\documentclass[a4paper,oneside,fleqn,11pt]{article}
\usepackage{amsmath,amssymb,amsfonts,amsthm,type1cm,bm,xcolor,mathrsfs}
\usepackage{mathtools}
\usepackage{graphicx,psfrag,epsf}
\usepackage{rotating}
\usepackage{authblk}
\usepackage{yfonts}
\usepackage{pdflscape}
\usepackage{setspace}
\usepackage{natbib}
\usepackage{float} 
\usepackage{caption}
\usepackage{subcaption}
\usepackage{booktabs}
\usepackage{siunitx}
\usepackage{threeparttable}
\RequirePackage[colorlinks,citecolor=blue,urlcolor=pink]{hyperref}
\usepackage{bigstrut}
\usepackage{comment}
\mathtoolsset{showonlyrefs=true}

\DeclareTextFontCommand{\texttt}{\ttfamily\upshape}
\DeclareMathOperator*\argmin{\arg\min}

\DeclareMathOperator\tr{tr}

\DeclareMathOperator\E{\mathbb{E}}
\DeclareMathOperator\Var{Var}
\DeclareMathOperator\Corr{Corr}
\DeclareMathOperator\Pro{\mathbb{P}}

\DeclareMathOperator\sgn{sgn}

\DeclareMathAlphabet\mathbfcal{OMS}{cmsy}{b}{n}
\def\b0{\mathbf{0}}
\def\b1{\mathbf{1}}

\def\bA{\mathbf{A}}
\def\bB{\mathbf{B}}
\def\bC{\mathbf{C}}
\def\bD{\mathbf{D}}

\def\bG{\mathbf{G}}

\def\bI{\mathbf{I}}
\def\bJ{\mathbf{J}}
\def\bK{\mathbf{K}}

\def\bM{\mathbf{M}}

\def\bQ{\mathbf{Q}}
\def\bR{\mathbf{R}}

\def\bV{\mathbf{V}}
\def\bW{\mathbf{W}}

\def\ba{\mathbf{a}}
\def\bb{\mathbf{b}}

\def\bd{\mathbf{d}}
\def\be{\mathbf{e}}

\def\bg{\mathbf{g}}

\def\bm{\mathbf{m}}

\def\br{\mathbf{r}}
\def\bs{\mathbf{s}}

\def\bu{\mathbf{u}}
\def\bv{\mathbf{v}}

\def\bx{\mathbf{x}}
\def\by{\mathbf{y}}
\def\bz{\mathbf{z}}

\def\cF{\mathcal{F}}
\def\cK{\mathcal{K}}

\def\bbR{\mathbb{R}}
\def\bbZ{\mathbb{Z}}
\def\bbN{\mathbb{N}}

\def\bGamma{\boldsymbol{\Gamma}}

\def\bbbeta{\boldsymbol{\eta}}

\def\btheta{\boldsymbol{\theta}}

\def\bOmega{\boldsymbol{\Omega}}

\def\bxi{\boldsymbol{\xi}}
\def\bSigma{\boldsymbol{\Sigma}}
\def\bGamma{\boldsymbol{\Gamma}}

\def\bzero{\mathbf{0}}
\def\bone{\mathbf{1}}
\def\what{\hat}
\def\wtilde{\tilde}

\newtheorem{thm}{Theorem}
\newtheorem{lem}{Lemma}
\newtheorem{cor}{Corollary}
\newtheorem{ass}{Assumption}
\newtheorem{rem}{Remark}
\newtheorem{prop}{Proposition}

\newcommand{\CD}{\stackrel{d}{\longrightarrow}}
\newcommand{\CP}{\stackrel{p}{\longrightarrow}}

\newcommand{\CDs}{\stackrel{d^*}{\longrightarrow}}
\newcommand{\CPs}{\stackrel{p^*}{\longrightarrow}}

\allowdisplaybreaks

\makeatletter
\def\section{\@startsection {section}{1}{\z@}{-3.5ex plus -1ex minus-.2ex}{2.3ex plus .2ex}{\large\bf}}
\makeatother

\makeatletter
\def\subsection{\@startsection {subsection}{1}{\z@}{-3.5ex plus -1ex minus-.2ex}{2.3ex plus .2ex}{\normalsize\bf}}
\makeatother

\title{Two-Point Dependent Wild Bootstrap for Weakly Dependent Estimating Equations}
\author[$\,\!$]{\textsc{Mikihito Nishi}$^*$}
\author[$\,\!$]{\textsc{Takashi Yamagata}$^\dagger$}
\affil[$*$]{\textit{Graduate School of Economics, University of Tokyo}}
\affil[$\dagger$]{\textit{Department of Economics and Related Studies, University of York}}
\affil[$\dagger$]{\textit{Graduate School of Economics and Management, Tohoku University}}

\date{October 2026}

\begin{document}
\maketitle

\begin{abstract}
\noindent 
This paper develops a general two-point dependent wild bootstrap (DWB) for weakly dependent estimating equations.
Its key feature is that the two-point marginal distribution and the latent serial dependence specification can be chosen separately.
The construction combines a normalized two-point distribution with a stationary latent Gaussian process via a Gaussian copula transformation, includes dependent Rademacher and Mammen multipliers, and nests the classical iid two-point wild bootstrap as the serially independent case.
The induced multiplier autocovariances determine the lag weights in a corresponding heteroskedasticity- and autocorrelation-consistent (HAC) covariance estimator, which coincides exactly with the conditional covariance of the bootstrap estimating-equation sum.
We establish first-order bootstrap validity for asymptotically linear estimators by showing that the matched-HAC estimator consistently estimates the long-run covariance and that the bootstrap estimating-equation sum converges conditionally to the same Gaussian limit as its original-sample counterpart, yielding valid HAC-studentized $z$-tests and the corresponding Wald and Lagrange multiplier tests.
Monte Carlo experiments in nonlinear generalized method of moments (GMM) and linear regression, with common HAC studentization and matched multiplier covariances, show that dependent Rademacher generally yields better size control than Mammen and Gaussian DWB, especially for restricted $z$-tests.
An application illustrates differences between asymptotic and bootstrap inference.
\end{abstract}


\section{Introduction}

Wild bootstrap methods provide a convenient approach to robust inference under heteroskedasticity when the disturbance distribution is otherwise left unspecified \citep{Wu1986,Liu1988,Mammen1993}.
Bounded two-point auxiliary distributions have played an important role in this literature.
Mammen's two-point distribution reproduces the first three moments relevant for classical higher-order arguments in heteroskedastic regression, while \citet{DavidsonMonticiniPeel2007} study a general class of mean-zero, unit-variance two-point distributions and the trade-off between their higher moments.
The symmetric Rademacher distribution provides the particularly simple $\{-1,1\}$ case and has attractive finite-sample properties in conventional wild bootstrap testing \citep{davidsonflachaire2008}.
Thus, even after imposing mean zero and unit variance, the choice of multiplier distribution leaves scope to vary higher-order marginal features such as skewness and kurtosis.

For dependent data, independent wild bootstrap weights do not reproduce the serial covariance relevant for inference.
The dependent wild bootstrap (DWB) addresses this problem by introducing serial dependence into the auxiliary variables.
A particularly convenient way of generating such dependence is through Gaussian multipliers, but this also fixes the marginal distribution of the multiplier: standardized Gaussian multipliers are unbounded and have fixed higher-order moments.
This removes the freedom to retain bounded two-point marginal distributions such as the Rademacher and Mammen laws.
This paper addresses this problem by showing that the two-point marginal law and the latent serial dependence specification can be chosen separately. 

We develop a two-point DWB using a latent stationary Gaussian process and a Gaussian copula threshold transformation.
For a normalized two-point distribution with distribution function $F$, we define
\begin{align}
    \xi_t
    =
    F^{-1}\{\Phi(Z_t)\},
    \label{eq:intro_copula_transform}
\end{align}
where $(Z_t)$ is a stationary Gaussian process.
The probability mechanism underlying Gaussian thresholding is classical \citep{EmrichPiedmonte1991}; we use it here to construct a dependent multiplier process while preserving exactly the chosen two-point marginal law.
The marginal distribution of $\xi_t$ is exactly $F$, while its serial dependence is induced by the correlation structure of the latent Gaussian process through the threshold transformation. 
Thus the chosen two-point marginal law can be preserved while the latent dependence specification is varied. 
The same device accommodates Rademacher, Mammen, and other normalized bounded two-point multipliers, within the same dependence construction. 
When the latent Gaussian process is serially independent, the construction reduces exactly to the classical iid two-point wild bootstrap, so the conventional iid Rademacher and Mammen wild bootstraps are nested as special cases.

The construction is related to, but distinct from, existing DWB and dependent multiplier procedures. 
The DWB of \citet{Shao2010} introduces serial dependence into the auxiliary variables, and related dependent multiplier procedures have subsequently been developed for other statistics and empirical processes \citep{LeuchtNeumann2013,doukhannemann2015,buecherkojadinovic2016}. 
Recent work by \citet{HounyoLin2026} develops multiway wild-cluster bootstrap procedures for linear regression with serially correlated common time effects, including a dependent Rademacher construction tailored to the time dimension.
Their objective is to reproduce the dependence induced by multiway clustering and serially correlated time effects, whereas our construction provides a general device for weakly dependent estimating equations that accommodates arbitrary normalized two-point laws and allows the marginal law and serial dependence profile to be specified separately.
To our knowledge, this is the first general two-point DWB that accommodates an arbitrary normalized two-point marginal law while allowing its serial dependence profile to be specified separately.

The construction also has a natural connection to heteroskedasticity-and-autocorrelation-consistent (HAC) covariance estimation.
The latent Gaussian correlation determines the covariance of the transformed two-point multipliers, which in turn determines the corresponding HAC lag weights.
Using these multiplier autocovariances as lag weights yields a matched-HAC estimator whose finite-sample value coincides exactly with the conditional covariance of the bootstrap score sum. For consistency, we separate the standard population-score HAC approximation from the additional effects of estimating and recentering the score contributions, and show that the latter are asymptotically negligible. 
The general connection between DWB and HAC covariance estimation is already present in \citet{Shao2010} and is made explicit by \citet{davidsonmonticini2014}; in the present framework, however, the multiplier process, its autocovariance sequence, and the associated HAC lag weights arise jointly from the specified two-point marginal law and latent dependence structure.

We then establish first-order validity of the resulting two-point DWB in a general weakly dependent estimating-equation framework.
In particular, we show that the bootstrap estimating-equation sum converges conditionally to the same Gaussian limit as its original-sample counterpart, yielding valid bootstrap distributions for studentized statistics and the corresponding Wald and Lagrange multiplier statistics based on asymptotically linear estimators.
The proof requires an additional step beyond existing DWB results because the Gaussian copula dependence profile is not finite-range.
We construct a finite-dependent approximation that preserves the exact two-point marginal distribution, control the resulting approximation error, and establish the required conditional Gaussian limit using a blocking argument.
For linear regression, the validity results also extend to the null-imposed residual bootstrap under primitive conditions.

The Monte Carlo analysis considers nonlinear
generalized method of moments (GMM) and linear regression
under serial dependence, stationary conditional
heteroskedasticity, and non-Gaussian innovations.
To compare multiplier distributions, the experiments
use a common HAC studentizer and exactly or approximately
matched multiplier covariance profiles.
Rademacher DWB generally yields more accurate
finite-sample size control than Mammen and Gaussian DWB,
particularly for the restricted $z$-statistic
$\wtilde z_{LM}$.
Size-adjusted power is broadly similar across procedures
for the restricted statistic, and the main size findings
are broadly robust to bandwidth choice.
The corresponding two-step GMM results give
qualitatively similar conclusions.

We illustrate the method using a nonlinear GMM specification based on the conditional-mean implication of the Cox--Ingersoll--Ross short-rate model.
The application shows that asymptotic and bootstrap approximations can lead to different conclusions at conventional significance levels.

The rest of the paper is organized as follows.
Section~\ref{sec:setup} introduces the estimating-equation framework and motivating examples.
Section~\ref{sec:constructing_multipliers} develops the dependent two-point multiplier construction, its matched-HAC representation, and its Rademacher and Mammen special cases. 
Section~\ref{sec:estimating_equations} develops the bootstrap inference procedures for asymptotically linear estimators, including unrestricted and restricted studentized procedures and the regression-specific residual bootstrap. 
Section~\ref{sec:asymptotics} establishes the asymptotic validity of the proposed procedures.
Section~\ref{sec:mc} presents the Monte Carlo experiments.
Section~\ref{sec:empirical} gives the empirical illustration.

\section{Setup and First-Order Framework}\label{sec:setup}

To fix ideas, suppose that an estimator $\what{\btheta}$ for a parameter of interest $\btheta_0\in\Theta\subset\bbR^k$ is asymptotically linear:
\begin{align}
    \sqrt{T}(\what{\btheta}-\btheta_0)
    =
    \bB_0\frac{1}{\sqrt{T}}\sum_{t=1}^T\bg_t(\btheta_0)+o_p(1),
    \label{eq:asym_linear_theta}
\end{align}
where $\bg_t(\btheta)$ is a $d$-dimensional score, moment, or estimating-equation contribution and $\bB_0$ is a $k\times d$ population linearization matrix.
The true value satisfies $\E\bg_t(\btheta_0)=\bzero$. 
We assume that $\{\bg_t(\btheta_0)\}$ is covariance-stationary, weakly serially dependent, and satisfies
\begin{align}
    \frac{1}{\sqrt{T}}\sum_{t=1}^T\bg_t(\btheta_0)
    \CD
    N(\bzero,\bOmega_0),
    \label{eq:score_clt}
\end{align}
where $\bOmega_0
    =
    \sum_{h=-\infty}^{\infty}\bGamma_0(h)>0$, $
    \bGamma_0(h)=\E\{\bg_t(\btheta_0)\bg_{t-h}(\btheta_0)'\}$. 
Then
\begin{align}
    \sqrt{T}(\what{\btheta}-\btheta_0)
    \CD
    N(\bzero,\bV_{\theta}),
    \label{eq:theta_limit_generic}
\end{align}
where $\bV_{\theta}
    :=
    \bB_0\bOmega_0\bB_0'$. 

Under serial dependence, feasible inference therefore requires estimation of the long-run covariance matrix $\bOmega_0$.

\subsection{HAC inference under serial dependence}\label{subsec:hac_inference}

Let $\what{\bg}_t=\bg_t(\what{\btheta})$ and $\overline{\what{\bg}}=T^{-1}\sum_{t=1}^T\what{\bg}_t$.
For $h\geq0$, define the sample autocovariance matrix
\begin{align}
    \what{\bGamma}(h)
    &=
    \frac{1}{T}\sum_{t=h+1}^T
    (\what{\bg}_t-\overline{\what{\bg}})
    (\what{\bg}_{t-h}-\overline{\what{\bg}})'.
    \label{eq:generic_sample_autocovariance}
\end{align}
Let $w_T(h)$ denote symmetric lag weights satisfying $w_T(0)=1$.
A generic HAC estimator of $\bOmega_0$ is
\begin{align}
    \what{\bOmega}_{\mathrm{HAC}}
    &=
    \what{\bGamma}(0)
    +
    \sum_{h=1}^{T-1}
    w_T(h)
    \left\{
    \what{\bGamma}(h)+\what{\bGamma}(h)'
    \right\}.
    \label{eq:generic_hac}
\end{align}
For example, the Bartlett-type \citet{NeweyWest1987} estimator uses $w_T(h)=(1-|h|/b_T)\mathbf{1}\{|h|\leq b_T\}$, where $b_T$ is the bandwidth or truncation parameter.
Let $\what{\bB}$ be a consistent sample analogue of $\bB_0$ and define
\begin{align}
    \what{\bV}_{\theta,\mathrm{HAC}}
    &=
    \what{\bB}\what{\bOmega}_{\mathrm{HAC}}\what{\bB}'.
    \label{eq:generic_hac_vtheta}
\end{align}
When $\what{\bOmega}_{\mathrm{HAC}}\CP\bOmega_0$ and $\what{\bB}\CP\bB_0$, replacing $\bV_{\theta}$ by $\what{\bV}_{\theta,\mathrm{HAC}}$ gives conventional first-order HAC inference for $\btheta_0$.
Such asymptotic inference can nevertheless be sensitive in finite samples to long-run variance estimation and to the normal or chi-square approximation.
This motivates a dependent wild bootstrap that reproduces the serial covariance relevant for HAC inference while retaining bounded two-point multipliers.
The construction below chooses the HAC lag weights from the autocovariances of the bounded dependent multipliers and reproduces exactly the same weights in the bootstrap.

\subsection{Examples}\label{subsec:examples}

\paragraph{Linear regression.}
Consider
\begin{align}
    y_t=\bx_t'\btheta_0+u_t,
    \qquad
    \E(\bx_tu_t)=\bzero,
\end{align}
where $\bx_t$ is a $k$-dimensional regressor and $\btheta_0$ is the parameter of interest.
Define
\begin{align}
    \bg_t(\btheta)=\bx_t(y_t-\bx_t'\btheta).
\end{align}
If $\bQ=\E(\bx_t\bx_t')$ is nonsingular, the ordinary least squares (OLS) estimator satisfies
\begin{align}
    \sqrt{T}(\what{\btheta}-\btheta_0)
    =
    \bQ^{-1}\frac{1}{\sqrt{T}}\sum_{t=1}^T\bx_tu_t+o_p(1).
\end{align}
Thus $\bB_0=\bQ^{-1}$ and $\bg_t(\btheta_0)=\bx_tu_t$.
The estimated contribution is $\what{\bg}_t=\bg_t(\what{\btheta})=\bx_t\what u_t$, $\what u_t=y_t-\bx_t'\what{\btheta}$.

\paragraph{Maximum likelihood and smooth $M$-estimation.}
Let $\ell_t(\btheta)$ be a log-likelihood, quasi-log-likelihood, or smooth criterion contribution and define the score
\begin{align}
    \bg_t(\btheta)=\partial \ell_t(\btheta)/\partial\btheta.
\end{align}
Suppose that $\what{\btheta}$ solves $T^{-1}\sum_{t=1}^T\bg_t(\what{\btheta})=\bzero$. 
Let $\bG_0=\E\left\{\partial\bg_t(\btheta_0) /\partial\btheta'\right\}$. 
If $\bG_0$ is nonsingular and the standard Taylor expansion is valid, then
\begin{align}
    \sqrt{T}(\what{\btheta}-\btheta_0)
    =
    -\bG_0^{-1}\frac{1}{\sqrt{T}}\sum_{t=1}^T\bg_t(\btheta_0)+o_p(1),
\end{align}
so $\bB_0=-\bG_0^{-1}$.
Under correct likelihood specification, $\bG_0$ is the negative information matrix; under quasi-likelihood or general smooth $M$-estimation, it is the corresponding sensitivity matrix.

\paragraph{GMM.}
Let $\btheta_0$ be defined by $\E\bg_t(\btheta_0)=\bzero$, where $\bg_t(\btheta)$ is a $d$-dimensional moment vector. 
Consider a GMM estimator that minimizes
\begin{align}
    \overline{\bg}_T(\btheta)'\bW_T\overline{\bg}_T(\btheta),
    \qquad
    \overline{\bg}_T(\btheta)=T^{-1}\sum_{t=1}^T\bg_t(\btheta),
\end{align}
where $\bW_T$ is a symmetric positive-definite weighting matrix satisfying $\bW_T\CP\bW_0$ for some symmetric positive-definite matrix $\bW_0$. 
Define 
    $\bD_0
    =
    \E\left\{
        \partial\bg_t(\btheta_0)/\partial\btheta'
    \right\}$. 
If $\bD_0'\bW_0\bD_0$ is nonsingular, then
\begin{align}
    \sqrt{T}(\what{\btheta}-\btheta_0)
    =
    -(\bD_0'\bW_0\bD_0)^{-1}\bD_0'\bW_0
    \frac{1}{\sqrt{T}}\sum_{t=1}^T\bg_t(\btheta_0)
    +o_p(1),
\end{align}
so 
    $\bB_0
    =
    -(\bD_0'\bW_0\bD_0)^{-1}\bD_0'\bW_0$.

\paragraph{Panel data regression.}
The framework also covers panel settings in which cross-sectional score or moment contributions are aggregated into a time-indexed estimating-equation sequence. Provided that the resulting sequence satisfies the first-order conditions above, the dependent two-point multiplier can be applied along the time dimension in the same way as for the preceding examples. 
Panel-specific theory and methods for large panels are developed by Dai, Matsushita, and Yamagata~(\citeyear{DaiMatsushitaYamagata2026}).

\section{Constructing Dependent Wild Bootstrap Multipliers}\label{sec:constructing_multipliers}

This section constructs the multiplier process used for score and moment bootstrap inference and links its covariance to HAC estimation.
The construction first specifies a bounded two-point marginal distribution and then introduces serial dependence through a latent Gaussian copula while preserving the exact marginal law.
The resulting covariance kernel defines the original-sample HAC estimator and, by construction, the conditional covariance of the bootstrap score sum.

\subsection{Bounded two-point wild bootstrap distributions}\label{sec:twopoint}

Fix $p\in(0,1)$. Let $\xi$ take two values $a_p<0<b_p$ with probabilities
\begin{align}
    \Pro(\xi=a_p)=p,
    \qquad
    \Pro(\xi=b_p)=1-p.
\end{align}
Imposing $\E\xi=0$ and $\E\xi^2=1$ gives
\begin{align}
    a_p=-\sqrt{\frac{1-p}{p}},
    \qquad
    b_p=\sqrt{\frac{p}{1-p}}.
    \label{eq:two_point_values}
\end{align}
The corresponding distribution function is denoted by $F_p$. 

The third and fourth moments are
\begin{align}
    \E\xi^3
    =
    \frac{2p-1}{\sqrt{p(1-p)}},
    \qquad
    \E\xi^4
    =
    \frac{(1-p)^2}{p}+\frac{p^2}{1-p}.
    \label{eq:two_point_moments}
\end{align}
Two choices are central.

\paragraph{Davidson--Flachaire/Rademacher two-point multiplier.}
For $p=1/2$, $a_p=-1$ and $b_p=1$. This is the Rademacher two-point
wild bootstrap multiplier. It is the symmetric bounded choice emphasized by
\citet{davidsonflachaire2008}. It satisfies $\E\xi=0$, $\E\xi^2=1$,
$\E\xi^3=0$, and $\E\xi^4=1$.

\paragraph{Mammen skewness-replicating two-point multiplier.}
Mammen's two-point distribution is obtained by choosing $p$ so that
$\E\xi^3=1$. This gives
\begin{align}
    p_M=\frac{\sqrt{5}+1}{2\sqrt{5}},
    \qquad
    a_M=\frac{1-\sqrt{5}}{2},
    \qquad
    b_M=\frac{1+\sqrt{5}}{2}.
    \label{eq:mammen_distribution}
\end{align}
It satisfies $\E\xi=0$, $\E\xi^2=\E\xi^3=1$, and $\E\xi^4=2$. Thus Mammen's distribution is the skewness-replicating bounded two-point choice.

The distinction between these two cases is useful. The Rademacher choice is symmetric and maximally simple. The Mammen choice preserves the classical third-moment matching condition. Both are bounded and both are special cases of the same two-point family.

\subsection{Copula-based bounded dependent two-point multipliers}
\label{sec:construction}
Let $\{Z_t\}_{t\in\bbZ}$ be a stationary standard Gaussian sequence satisfying
\begin{align}
    \E Z_t=0,
    \qquad
    \E Z_t^2=1,
    \qquad
    \Corr(Z_t,Z_{t-h})=\rho_h,
    \qquad
    h\in\bbZ.
    \label{eq:latent_gaussian_general}
\end{align}
By stationarity, $\rho_{-h}=\rho_h$. 
For a sample of size $T$, define the $T\times1$ latent Gaussian vector
\begin{align}
    \bz_T
    =
    (Z_1,\ldots,Z_T)',
    \qquad
    \bz_T
    \sim
    N(\bzero,\bR_T),
    \qquad
    (\bR_T)_{t,s}
    = 
    \rho_{|t-s|}.
\end{align}

Apply the Gaussian copula transformation coordinatewise:
\begin{align}
    U_t
    =
    \Phi(Z_t),
    \qquad
    \xi_t
    =
    F_p^{-1}(U_t),
    \qquad
    t=1,\ldots,T.
\end{align}
Equivalently, with $q_p=\Phi^{-1}(p)$,
\begin{align}
    \xi_t
    =
    a_p\,\bone\{Z_t\leq q_p\}
    +
    b_p\,\bone\{Z_t>q_p\}.
    \label{eq:threshold_twopoint}
\end{align}
Thus each $\xi_t$ has exactly the prescribed two-point marginal distribution $F_p$.
Writing
\begin{align}
    \bxi_T
    =
    (\xi_1,\ldots,\xi_T)',
\end{align}
the transformation maps the Gaussian vector $\bz_T$ into a $T$-dimensional random vector whose coordinates each take the two values $a_p$ and $b_p$, with dependence inherited from the joint Gaussian distribution of $\bz_T$. 

For any two coordinates separated by lag $h$, $(Z_t,Z_{t-h})$ is bivariate standard normal with correlation $\rho_h$.
Since $\xi_t
    =
    F_p^{-1}\{\Phi(Z_t)\}$ and  
    $\xi_{t-h}
    =
    F_p^{-1}\{\Phi(Z_{t-h})\}$, 
their correlation is determined by $\rho_h$, given $p$.
We denote the resulting correlation transformation by $K_p$, so that $K_p(\rho_h)$ is the correlation between $\xi_t$ and $\xi_{t-h}$. 

The following proposition gives the explicit form of this correlation transformation and establishes that the resulting multiplier autocovariance sequence is valid.

\begin{prop}[Copula-induced correlation map]
\label{prop:copula_kernel}
For the normalized two-point distribution $F_p$,
\begin{align}
    K_p(\rho)
    =
    \frac{\Phi_2(q_p,q_p;\rho)-p^2}{p(1-p)},
    \qquad
    -1\leq\rho\leq1,
    \label{eq:general_copula_kernel}
\end{align}
where $\Phi_2(\cdot,\cdot;\rho)$ is the bivariate standard normal distribution function with correlation $\rho$.
Consequently, since $\{\xi_t\}$ is stationary, its autocovariance sequence is given by
\begin{align}
    \E(\xi_t\xi_{t-h})
    =
    K_p(\rho_h),
    \qquad
    h\in\bbZ.
\end{align}
Moreover, 
    $\bK_{p,T}
    :=
    \E(\bxi_T\bxi_T')$ 
is positive semi-definite for every finite $T$.
\end{prop}

\begin{proof}
Let $I_t=\bone\{Z_t\leq q_p\}$.
Using $\xi_t=a_pI_t+b_p(1-I_t)$, $pa_p+(1-p)b_p=0$, and the definitions of $a_p$ and $b_p$,
\begin{align}
    \xi_t
    =
    \frac{p-I_t}{\sqrt{p(1-p)}}.
\end{align}
Therefore
\begin{align}
    \E(\xi_t\xi_{t-h})
    &=
    \frac{\E\{(p-I_t)(p-I_{t-h})\}}{p(1-p)}
    \nonumber\\
    &=
    \frac{\E(I_tI_{t-h})-p^2}{p(1-p)}.
\end{align}
Since
    $\E(I_tI_{t-h})
    =
    \Phi_2(q_p,q_p;\rho_h)$, 
the expression for $K_p(\rho_h)$ follows. 
Finally, for any $\ba=(a_1,\ldots,a_T)'\in\bbR^T$, 
    $\ba'\bK_{p,T}\ba
    =
    \E[
        (
            \sum_{t=1}^T a_t\xi_t
        )^2
    ]
    \geq0$; thus, $\bK_{p,T}$ is positive semi-definite for every finite $T$.
\end{proof}

For implementation, we specify the symmetric autocorrelation sequence of the latent Gaussian process as
\begin{align}
    \rho_{T,h}
    &=
    \exp\left\{
        -\left(\frac{h}{b_T}\right)^2
    \right\},
    \qquad
    h\in\bbZ,
    \label{eq:latent_gaussian_corr}
\end{align}
where $b_T>0$ is the dependence-scale parameter.
In practice, generate $\bz_T\sim N(\bzero,\bR_T)$,
with $(\bR_T)_{t,s}=\exp\left\{-\left(|t-s|/b_T\right)^2\right\}$, and apply the
two-point transformation \eqref{eq:threshold_twopoint}
coordinatewise.

For $h\geq0$, Proposition~\ref{prop:copula_kernel} gives
\begin{align}
    \E(\xi_t\xi_{t-h})
    =
    K_p(\rho_{T,h})
    =
    K_p\left[
        \exp\left\{
            -\left(\frac{h}{b_T}\right)^2
        \right\}
    \right].
    \label{eq:induced_kernel_general}
\end{align}
For notational simplicity, the dependence of $Z_t$ and $\xi_t$ on $T$ is suppressed when \eqref{eq:latent_gaussian_corr} is used.
The induced multiplier autocovariance $K_p(\rho_{T,h})$ will also serve as the HAC lag weight at lag $h$.

\subsection{From multiplier dependence to matched-HAC inference}\label{sec:matched}

\paragraph{Original-sample HAC.}
In the generic HAC estimator \eqref{eq:generic_hac}, choose
\begin{align}
    w_T(h)
    &=
    K_p(\rho_{T,h}).
    \label{eq:matched_hac_weight}
\end{align}
The resulting original-sample HAC estimator is
\begin{align}
    \what{\bOmega}_p
    &:=
    \what{\bGamma}(0)
    +
    \sum_{h=1}^{T-1}
    K_p(\rho_{T,h})
    \left\{
    \what{\bGamma}(h)+\what{\bGamma}(h)'
    \right\}.
    \label{eq:matched_hac}
\end{align}
Equivalently,
\begin{align}
    \what{\bOmega}_p
    &=
    \frac{1}{T}\sum_{t=1}^T\sum_{s=1}^T
    K_p(\rho_{T,t-s})
    (\what{\bg}_t-\overline{\what{\bg}})
    (\what{\bg}_s-\overline{\what{\bg}})'.
    \label{eq:double_sum_hac}
\end{align}
The corresponding original-sample covariance estimator for the linearized estimator is
\begin{align}
    \what{\bV}_{\theta}
    &=
    \what{\bB}\what{\bOmega}_p\what{\bB}'.
    \label{eq:vtheta_hat}
\end{align}
Thus $\what{\bOmega}_p$ is the matched-HAC estimator targeting the
long-run covariance $\bOmega_0$, with lag weight at $h$ given by the
multiplier autocovariance $K_p(\rho_{T,h})$.
Section~\ref{sec:asymptotics} establishes conditions under which
$\what{\bOmega}_p\CP\bOmega_0$.

\paragraph{Bootstrap match.}
For each bootstrap replication, generate $\{\xi_t^*\}_{t=1}^T$ independently of the data from the same Gaussian copula threshold construction, so that
\begin{align}
    \E^*(\xi_t^*)
    &=
    0,
    \qquad
    \E^*(\xi_t^*\xi_{t-h}^*)
    =
    K_p(\rho_{T,h}),
    \label{eq:bootstrap_multiplier_covariance}
\end{align}
and define
\begin{align}
    \what{\bg}_t^*
    &=
    \xi_t^*(\what{\bg}_t-\overline{\what{\bg}}).
    \label{eq:score_bootstrap_theta}
\end{align}
Conditional on the data,
\begin{align}
    \E^*(\what{\bg}_t^*\what{\bg}_{t-h}^{*'})
    &=
    K_p(\rho_{T,h})
    (\what{\bg}_t-\overline{\what{\bg}})
    (\what{\bg}_{t-h}-\overline{\what{\bg}})'.
    \label{eq:bootstrap_score_pair_covariance}
\end{align}
Consequently, the conditional covariance of the bootstrap score sum exactly matches the original-sample HAC estimator:
\begin{align}
    \Var^*\left(
    \frac{1}{\sqrt{T}}\sum_{t=1}^T\what{\bg}_t^*
    \right)
    &=
    \frac{1}{T}\sum_{t=1}^T\sum_{s=1}^T
    K_p(\rho_{T,t-s})
    (\what{\bg}_t-\overline{\what{\bg}})
    (\what{\bg}_s-\overline{\what{\bg}})'\\
    &=
    \what{\bOmega}_p.
    \label{eq:exact_hac_match}
\end{align}
Hence the original-sample HAC estimator and the conditional covariance of the bootstrap score sum coincide exactly.
This exact equality is the matching property: the HAC weight $K_p(\rho_{T,h})$ at lag $h$ is also the covariance of two bootstrap multipliers $h$ periods apart.

Consequently, once matched-HAC consistency is established, the conditional covariance of the bootstrap score sum also converges to $\bOmega_0$. 
The asymptotic results below further show that the bootstrap score sum converges conditionally to the same $N(\bzero,\bOmega_0)$ limit as its original-sample counterpart.

For draw-specific studentization, define the bootstrap analogue of $\what{\bOmega}_p$ from the bootstrap score array:
\begin{align}
    \what{\bOmega}_p^*
    &=
    \what{\bGamma}^*(0)
    +
    \sum_{h=1}^{T-1}
    K_p(\rho_{T,h})
    \left\{
    \what{\bGamma}^*(h)+\what{\bGamma}^*(h)'
    \right\},
    \label{eq:bootstrap_matched_hac}
\end{align}
where, for $h\geq0$, 
    $\what{\bGamma}^*(h)
    =
    \frac{1}{T}\sum_{t=h+1}^T
    (\what{\bg}_t^*-\overline{\what{\bg}^*})
    (\what{\bg}_{t-h}^*-\overline{\what{\bg}^*})'$, $
    \overline{\what{\bg}^*}
    =
    \frac{1}{T}\sum_{t=1}^T\what{\bg}_t^*$. 
The associated draw-specific covariance estimator of $\what{\btheta}^*$ is
\begin{align}
    \what{\bV}_{\theta}^*
    &=
    \what{\bB}\what{\bOmega}_p^*\what{\bB}'.
    \label{eq:bootstrap_theta_covariance}
\end{align}
Note that $\what{\bOmega}_p^*$ and $\what{\bV}_{\theta}^*$ vary across bootstrap replications.

\begin{rem}[Positive semi-definiteness]
Equation \eqref{eq:exact_hac_match} shows that $\what{\bOmega}_p$ is a conditional covariance matrix and is therefore positive semi-definite for every finite sample.
At the same time, \eqref{eq:matched_hac} identifies it as a HAC estimator with lag weights $K_p(\rho_{T,h})$.
\end{rem}

\subsubsection{The Rademacher special case}

When $p=1/2$, $q_p=0$, $a_p=-1$, and $b_p=1$.
Then $\xi_t=\sgn(Z_t)$ up to the value assigned at zero.
The Gaussian sign-correlation identity gives
\begin{align}
    K_{1/2}(\rho)=\frac{2}{\pi}\arcsin(\rho).
    \label{eq:arcsine_special_case}
\end{align}
Writing $x=h/b_T$, the corresponding lag-weight function is
\begin{align}
    K_{1/2}\{\exp(- x^2)\}
    =
    \frac{2}{\pi}\arcsin\{\exp(- x^2)\}.
    \label{eq:arcsine_kernel}
\end{align}
Near the origin,
\begin{align}
    \frac{2}{\pi}\arcsin\{\exp(- x^2)\}
    =
    1-\frac{2\sqrt{2}}{\pi}|x|+o(|x|),
    \qquad x\to0.
    \label{eq:kernel_cusp}
\end{align}
Thus the Rademacher case has an explicit arcsine lag-weight function with a Bartlett-like cusp.

\subsubsection{The Mammen special case}

When $p=p_M$ in \eqref{eq:mammen_distribution}, the multiplier process has exact Mammen marginals at every date.
The covariance transformation is
\begin{align}
    K_{p_M}(\rho)=\frac{\Phi_2(q_{p_M},q_{p_M};\rho)-p_M^2}{p_M(1-p_M)}.
    \label{eq:mammen_kernel}
\end{align}
There is no arcsine simplification in this asymmetric case, but the transformation is explicit and can be evaluated numerically. For a specified multiplier correlation at a single lag, the corresponding latent Gaussian correlation can be obtained numerically by inverting $K_{p_M}$. In our construction, however, the latent Gaussian correlation sequence is specified directly through \eqref{eq:latent_gaussian_corr}, so no pointwise inversion of a prespecified sequence of multiplier correlations across lags is required.

\section{Dependent Wild Bootstrap for Estimating Equations}\label{sec:estimating_equations}

The setup in Section~\ref{sec:setup} gives 
\begin{align}
    \sqrt{T}(\what{\btheta}-\btheta_0)
    &=
    \bB_0\frac{1}{\sqrt{T}}\sum_{t=1}^T\bg_t(\btheta_0)+o_p(1),
\end{align}
and $\sqrt{T}(\what{\btheta}-\btheta_0)\CD N(\bzero,\bV_{\theta})$ with $\bV_{\theta}=\bB_0\bOmega_0\bB_0'$.
The dependent wild bootstrap applies the bootstrap contributions $\{\what{\bg}_t^*\}$ defined in \eqref{eq:score_bootstrap_theta} to this first-order representation.
The bootstrap analogue of the asymptotic linear representation is
\begin{align}
    \sqrt{T}(\what{\btheta}^*-\what{\btheta})
    &=
    \what{\bB}
    \frac{1}{\sqrt{T}}\sum_{t=1}^T
    \what{\bg}_t^*,
    \label{eq:bootstrap_theta_update}
\end{align}
where $\what{\bB}$ is the sample analogue of $\bB_0$.
The notation $\what{\btheta}^*$ is symbolic in this general formulation: the procedure resamples the first-order representation directly and does not require re-estimation from bootstrap data. 
In the linear regression implementation considered in Subsection~\ref{subsec:reg_residual}, where the observed regressors are held fixed across bootstrap draws, this direct bootstrap update coincides exactly with the unrestricted OLS estimator computed from the generated bootstrap sample.

For likelihood scores or just-identified smooth estimating equations, $\what{\bB}=-\what{\bG}^{-1}$ with $\what{\bG}=T^{-1}\sum_{t=1}^T\partial\bg_t(\what{\btheta})/\partial\btheta'$.
For efficient two-step GMM,
$\what{\bB}=-(\what{\bD}'\bW_T\what{\bD})^{-1}
\what{\bD}'\bW_T$, where $\what{\bD}=T^{-1}\sum_{t=1}^T\partial\bg_t(\what{\btheta})/\partial\btheta'$.
The weighting matrix $\bW_T$ is constructed from a first-step GMM estimator and satisfies $\bW_T\CP\bOmega_0^{-1}$.
When a matched-HAC estimator is used in this construction,
$\bW_T-\what{\bOmega}_p^{-1}=o_p(1)$.
The matched original-sample covariance estimator $\what{\bV}_{\theta}$ is given in \eqref{eq:vtheta_hat}, and \eqref{eq:exact_hac_match} shows that its moment covariance component is reproduced exactly by the bootstrap. 

\subsection{Unrestricted and restricted studentized implementations}\label{subsec:wald_lm_implementations}

We consider a fixed number $\ell$ of linear restrictions, $H_0:\bA\btheta=\ba_0$, where $\bA\in\bbR^{\ell\times k}$ has rank $\ell$ and $\ba_0\in\bbR^\ell$.
It is useful to formulate inference first in terms of feasible studentized $\ell$-dimensional vectors and then obtain the conventional Wald and LM statistics by taking quadratic forms.
The unrestricted Wald implementation uses the unrestricted estimator, whereas the restricted LM implementation imposes the null before forming the bootstrap moments. 
We write $\what{\bz}_{W}$ for the feasible unrestricted studentized vector and $\wtilde{\bz}_{LM}$ for the feasible restricted studentized vector.
For $\ell=1$, we refer to the scalar quantities $\what z_W$ and $\wtilde z_{LM}$ as the unrestricted and restricted $z$-statistics, respectively. 

For scalar restrictions, retaining the signed studentized statistic is
particularly useful because it preserves the direction of departures from
the null, which can be economically meaningful in regression and structural
parameter applications. It also permits the lower and upper tails of the
sampling approximation to be assessed separately, whereas the corresponding
quadratic Wald and LM statistics fold the two directions together.

All inverse square roots below denote the symmetric positive-definite inverse square root of the corresponding covariance matrix.

For the unrestricted implementation, define the feasible studentized restriction vector
\begin{align}
    \what{\bz}_{W}
    &=
    (\bA\what{\bV}_{\theta}\bA')^{-1/2}
    \sqrt{T}(\bA\what{\btheta}-\ba_0),
    \qquad
    W_T
    =
    \what{\bz}_{W}'\what{\bz}_{W}.
    \label{eq:wald_stat_generic}
\end{align}
Thus the usual Wald statistic is the squared Euclidean norm of the studentized restriction vector.
To approximate the distribution of $\what{\bz}_{W}$, use the bootstrap analogue in \eqref{eq:bootstrap_theta_update} and, for recomputed studentization, define
\begin{align}
    \what{\bz}_{W}^*
    &=
    (\bA\what{\bV}_{\theta}^*\bA')^{-1/2}
    \sqrt{T}\bA(\what{\btheta}^*-\what{\btheta}),
    \qquad
    W_T^*
    =
    \what{\bz}_{W}^{*'}\what{\bz}_{W}^*.
    \label{eq:wald_bootstrap_generic}
\end{align}
For fixed studentization, replace $\what{\bV}_{\theta}^*$ in \eqref{eq:wald_bootstrap_generic} by the original-sample matrix $\what{\bV}_{\theta}$.

The restricted implementation imposes the null before the moment contribution is formed.
Since $\operatorname{rank}(\bA)=\ell$, use an invertible linear reparameterization $\boldsymbol{\eta}=(\boldsymbol{\eta}_1',\boldsymbol{\eta}_2')'$ such that $\boldsymbol{\eta}_2=\bA\btheta$, where $\boldsymbol{\eta}_2\in\bbR^\ell$.
Under the null, $\boldsymbol{\eta}_2=\ba_0$, while $\boldsymbol{\eta}_1$ contains the $k-\ell$ nuisance parameters.
Write $\boldsymbol{\eta}_{1,0}$ for the true nuisance parameter, $\wtilde{\boldsymbol{\eta}}_1$ for its restricted estimate, and $\wtilde{\btheta}$ for the corresponding restricted estimate in the original parameterization, and define $\wtilde{\bg}_t=\bg_t(\wtilde{\btheta})$ and $\overline{\wtilde{\bg}}=T^{-1}\sum_{t=1}^T\wtilde{\bg}_t$.
The matched-HAC estimator based on the restricted moments is
\begin{align}
    \wtilde{\bOmega}_p
    &=
    \frac{1}{T}\sum_{t=1}^T\sum_{s=1}^T
    K_p(\rho_{T,t-s})
    (\wtilde{\bg}_t-\overline{\wtilde{\bg}})
    (\wtilde{\bg}_s-\overline{\wtilde{\bg}})'.
    \label{eq:restricted_matched_hac}
\end{align}
Let $\wtilde{\bG}=(\wtilde{\bG}_1,\wtilde{\bG}_2)$ be the sample Jacobian with respect to $(\boldsymbol{\eta}_1',\boldsymbol{\eta}_2')'$ at $\wtilde{\btheta}$, where $\wtilde{\bG}_1\in\bbR^{d\times(k-\ell)}$ and $\wtilde{\bG}_2\in\bbR^{d\times\ell}$.
Following \citet{NeweyMcFadden1994}, define
\begin{align}
    \wtilde{\bG}_{2\cdot1}
    &=
    \wtilde{\bG}_2
    -
    \wtilde{\bG}_1
    (\wtilde{\bG}_1'\wtilde{\bOmega}_p^{-1}\wtilde{\bG}_1)^{-1}
    \wtilde{\bG}_1'\wtilde{\bOmega}_p^{-1}\wtilde{\bG}_2.
    \label{eq:g2_partial_generic}
\end{align}
By construction, $\wtilde{\bG}_1'\wtilde{\bOmega}_p^{-1}\wtilde{\bG}_{2\cdot1}=\bzero$. 
The corresponding $\ell$-dimensional restricted studentized vector is 
\begin{align}
    \wtilde{\bz}_{LM}
    &=
    -\left(
    \wtilde{\bG}_{2\cdot1}'
    \wtilde{\bOmega}_p^{-1}
    \wtilde{\bG}_{2\cdot1}
    \right)^{-1/2}
    \wtilde{\bG}_{2\cdot1}'
    \wtilde{\bOmega}_p^{-1}
    \sqrt{T}\,\overline{\wtilde{\bg}},
    \qquad
    LM_T
    =
    \wtilde{\bz}_{LM}'\wtilde{\bz}_{LM}.
    \label{eq:lm_stat_generic}
\end{align}
For a scalar restriction, the leading minus sign is chosen so that, to first order, the sign of the restricted $z$-statistic $\wtilde z_{LM}$ agrees with that of the unrestricted $z$-statistic $\what z_W$, while leaving the LM statistic unchanged. 
Thus $LM_T$ is the usual LM statistic, given by the squared Euclidean norm of the restricted studentized vector. 

For the direct bootstrap, define the restricted bootstrap moments and their sample average by
\begin{align}
    \wtilde{\bg}_t^*
    &=
    \xi_t^*(\wtilde{\bg}_t-\overline{\wtilde{\bg}}),
    \qquad
    \overline{\wtilde{\bg}^*}
    =
    T^{-1}\sum_{t=1}^T\wtilde{\bg}_t^*.
    \label{eq:restricted_moment_bootstrap}
\end{align}
Note that, conditional on the data,
\begin{align*}
    \Var^*\left(
    \frac{1}{\sqrt{T}}\sum_{t=1}^T\wtilde{\bg}_t^*
    \right)
    &=
    \wtilde{\bOmega}_p
\end{align*}
holds exactly.
For draw-specific studentization, define
\begin{align}
    \wtilde{\bOmega}_p^*
    &=
    \frac{1}{T}\sum_{t=1}^T\sum_{s=1}^T
    K_p(\rho_{T,t-s})
    (\wtilde{\bg}_t^*-\overline{\wtilde{\bg}^*})
    (\wtilde{\bg}_s^*-\overline{\wtilde{\bg}^*})'.
    \label{eq:restricted_bootstrap_matched_hac}
\end{align}
The original-sample projection $\wtilde{\bG}_{2\cdot1}'\wtilde{\bOmega}_p^{-1}$ is held fixed across bootstrap draws, and the recomputed feasible bootstrap restricted studentized vector is 
\begin{align}
    \wtilde{\bz}_{LM}^*
    &=
    -\left\{
    \wtilde{\bG}_{2\cdot1}'
    \wtilde{\bOmega}_p^{-1}
    \wtilde{\bOmega}_p^*
    \wtilde{\bOmega}_p^{-1}
    \wtilde{\bG}_{2\cdot1}
    \right\}^{-1/2}
    \wtilde{\bG}_{2\cdot1}'
    \wtilde{\bOmega}_p^{-1}
    \sqrt{T}\,\overline{\wtilde{\bg}^*},
    \qquad
    LM_T^*
    =
    \wtilde{\bz}_{LM}^{*'}\wtilde{\bz}_{LM}^*.
    \label{eq:lm_bootstrap_generic}
\end{align}
For fixed studentization, replace $\wtilde{\bOmega}_p^*$ in the covariance matrix of \eqref{eq:lm_bootstrap_generic} by $\wtilde{\bOmega}_p$, leaving the same original-sample projection fixed.

For restricted efficient two-step GMM, the role of the projection $\wtilde{\bG}_{2\cdot1}$ can also be seen from the first-order conditions. 
The weighting matrix constructed from the first-step GMM estimator satisfies $\bW_T-\wtilde{\bOmega}_p^{-1}=o_p(1)$, so the first-order conditions imply $\wtilde{\bG}_1'
    \wtilde{\bOmega}_p^{-1}
    \overline{\wtilde{\bg}}
    =
    o_p(T^{-1/2})$. Hence, to first order,
    $\wtilde{\bG}_{2\cdot1}'
    \wtilde{\bOmega}_p^{-1}
    \sqrt{T}\,\overline{\wtilde{\bg}}
    =
    \wtilde{\bG}_2'
    \wtilde{\bOmega}_p^{-1}
    \sqrt{T}\,\overline{\wtilde{\bg}}$;
see \citet[][p.~2230]{NeweyMcFadden1994}. 
However, the direct bootstrap does not re-estimate the restricted nuisance parameter in each draw, so the corresponding bootstrap first-order condition is not available. 
The projection $\wtilde{\bG}_{2\cdot1}'\wtilde{\bOmega}_p^{-1}$ is therefore retained in \eqref{eq:lm_bootstrap_generic} to account for the first-order effect of restricted nuisance-parameter estimation.


\subsection{Regression-specific residual resampling}\label{subsec:reg_residual}

Regression models permit an additional implementation in which the same fixed $\ell$-dimensional linear restriction is imposed and the multiplier is applied to residuals rather than directly to the score.

\subsubsection{Original-sample statistics}
Consider
\begin{align}
    y_t
    =
    \bx_t'\btheta+u_t,
    \qquad
    H_0:\bA\btheta=\ba_0,
    \qquad
    t=1,\ldots,T,
    \label{eq:reg_model}
\end{align}
where $\bx_t=(1,\bz_t')'\in\bbR^k$. 
We assume throughout the regression framework that the intercept is included under both the null and alternative and is left unrestricted by $\bA$.
Let $\what{\btheta}$ denote the unrestricted OLS estimator and define
\begin{align}
    \what{\bQ}_x
    &=
    T^{-1}\sum_{t=1}^T\bx_t\bx_t',
    \\
    \what u_t
    &=
    y_t-\bx_t'\what{\btheta},
    \qquad
    \what{\bg}_t
    =
    \bx_t\what u_t.
\end{align}
The unrestricted normal equations imply $\overline{\what{\bg}}=\bzero$.
Let $\what{\bOmega}_p$ denote the matched-HAC estimator based on $\{\what{\bg}_t\}$ and set
\begin{align}
    \what{\bV}_{\theta}
    =
    \what{\bQ}_x^{-1}\what{\bOmega}_p\what{\bQ}_x^{-1}.
    \label{def:vcov_wald_linear}
\end{align}
The regression Wald statistic is therefore \eqref{eq:wald_stat_generic} with $\what{\bB}=\what{\bQ}_x^{-1}$.

For the restricted implementation, choose a fixed nonsingular matrix $\bC=(\bC_1,\bC_2)$ such that $\bA\bC_1=\bzero$ and $\bA\bC_2=\bI_\ell$, and write $\btheta=\bC_1\boldsymbol{\eta}_1+\bC_2\boldsymbol{\eta}_2$ so that $\boldsymbol{\eta}_2=\bA\btheta$.
Let $\wtilde{\boldsymbol{\eta}}_1$ denote the restricted OLS estimator under $\boldsymbol{\eta}_2=\ba_0$ and define
\begin{align}
    \wtilde{\btheta}
    &=
    \bC_1\wtilde{\boldsymbol{\eta}}_1+\bC_2\ba_0,
    \\
    \wtilde u_t
    &=
    y_t-\bx_t'\wtilde{\btheta},
    \qquad
    \wtilde{\bg}_t
    =
    \bx_t\wtilde u_t,
    \qquad
    \overline{\wtilde{\bg}}
    =
    T^{-1}\sum_{t=1}^T\wtilde{\bg}_t.
\end{align}
Because the intercept is an unrestricted nuisance parameter, the restricted normal equations imply $\sum_{t=1}^T\wtilde u_t=0$.
Thus the restricted residuals themselves require no additional recentering before multiplication, although the restricted moment array is centered in the direct bootstrap and matched-HAC studentizer.
For the matched-HAC implementation, let $\wtilde{\bOmega}_p$ be defined from $\{\wtilde{\bg}_t-\overline{\wtilde{\bg}}\}$ as in \eqref{eq:restricted_matched_hac}.
For $\bg_t(\btheta)=\bx_t(y_t-\bx_t'\btheta)$,
\begin{align}
    \wtilde{\bG}
    =
    (\wtilde{\bG}_1,\wtilde{\bG}_2)
    =
    -\what{\bQ}_x(\bC_1,\bC_2).
\end{align}
Using these regression-specific derivative matrices, define $\wtilde{\bG}_{2\cdot1}$ by \eqref{eq:g2_partial_generic}.
The regression LM statistic is therefore \eqref{eq:lm_stat_generic} with these regression-specific objects.

\subsubsection{Residual resampling}
For a Wald-type residual bootstrap, set $u_t^\dagger=\xi_t^*\what u_t$ and generate
\begin{align}
    y_t^*
    =
    \bx_t'\what{\btheta}+u_t^\dagger.
\end{align}
Here $u_t^\dagger$ is the generated bootstrap disturbance.

Let $\what{\btheta}^*$ denote the unrestricted OLS estimator from the bootstrap sample.
As shown below, this estimator coincides exactly with the direct-score bootstrap update denoted by $\what{\btheta}^*$ in \eqref{eq:bootstrap_theta_update}, so we use the same notation.
Define the post-estimation bootstrap residual and residual-bootstrap moment contribution by
\begin{align}
    \what u_t^*
    &=
    y_t^*-\bx_t'\what{\btheta}^*,
    \qquad
    \what{\bg}_{t,\mathrm{res}}^*
    =
    \bx_t\what u_t^*.
\end{align}
The generated disturbance $u_t^\dagger$ is generally different from the post-estimation residual $\what u_t^*$.
The OLS normal equations give exactly
\begin{align}
    \sqrt T\bA(\what{\btheta}^*-\what{\btheta})
    =
    \bA\what{\bQ}_x^{-1}
    T^{-1/2}\sum_{t=1}^T\xi_t^*\what{\bg}_t.
    \label{eq:reg_resid_wald_numerator}
\end{align}
Thus the Wald numerator from refitting the bootstrap sample is exactly the same as that obtained by direct score resampling in \eqref{eq:bootstrap_theta_update}, since 
$\overline{\what{\bg}}=\bzero$.
For fixed studentization, replace $\what{\bV}_{\theta}^*$ in \eqref{eq:wald_bootstrap_generic} by the original-sample matrix
$\what{\bV}_{\theta}$ defined in \eqref{def:vcov_wald_linear}; for recomputed studentization, construct $\what{\bV}_{\theta}^*$
from the draw-specific moment array $\{\what{\bg}_{t,\mathrm{res}}^*\}$.


For an LM-type residual bootstrap, impose the restriction before resampling by setting $\wtilde u_t^\dagger=\xi_t^*\wtilde u_t$ and generate
\begin{align}
    y_t^*
    =
    \bx_t'\wtilde{\btheta}+\wtilde u_t^\dagger.
\end{align}
For each bootstrap sample, re-estimate $\boldsymbol{\eta}_1$ under $\boldsymbol{\eta}_2=\ba_0$, let $\wtilde{\boldsymbol{\eta}}_1^*$ denote the resulting nuisance estimate, and define
\begin{align}
    \wtilde{\btheta}^*
    &=
    \bC_1\wtilde{\boldsymbol{\eta}}_1^*+\bC_2\ba_0,
    \\
    \wtilde u_t^*
    &=
    y_t^*-\bx_t'\wtilde{\btheta}^*,
    \qquad
    \wtilde{\bg}_{t,\mathrm{res}}^*
    =
    \bx_t\wtilde u_t^*.
\end{align}
The definition of $\wtilde{u}_t^*$ and $\wtilde{\bG}_{2\cdot1}'\wtilde{\bOmega}_p^{-1}\wtilde{\bG}_1=\bzero$ give the exact decomposition\footnote{This decomposition does not require the regression to contain an intercept; the intercept condition above is used only to ensure that the restricted residuals have zero sample mean.}
\begin{align}
    \wtilde{\bG}_{2\cdot1}'\wtilde{\bOmega}_p^{-1}
    T^{-1/2}\sum_{t=1}^T\wtilde{\bg}_{t,\mathrm{res}}^*
    =
    \wtilde{\bG}_{2\cdot1}'\wtilde{\bOmega}_p^{-1}
T^{-1/2}\sum_{t=1}^T
\xi_t^*(\wtilde{\bg}_t-\overline{\wtilde{\bg}})
+
\wtilde{\bG}_{2\cdot1}'\wtilde{\bOmega}_p^{-1}
\overline{\wtilde{\bg}}\,
T^{-1/2}\sum_{t=1}^T\xi_t^* .
    \label{eq:reg_resid_lm_numerator}
\end{align} 
The first term on the right-hand side is the centered direct-bootstrap projected moment, while the second is a centering correction that is asymptotically negligible under $H_0$. 
For recomputed studentization, construct $\wtilde{\bOmega}_p^*$ from $\{\wtilde{\bg}_{t,\mathrm{res}}^*\}$ after subtracting its bootstrap
sample average, and use it as the middle covariance matrix of \eqref{eq:lm_bootstrap_generic}, while $\wtilde{\bG}_{2\cdot1}$ and $\wtilde{\bOmega}_p^{-1}$ remain fixed at their original restricted-sample values.
Primitive conditions for residual resampling are given in Subsection~\ref{subsec:reg_validity}.

\section{Asymptotic Results}\label{sec:asymptotics}

In this section, we derive the asymptotic properties of the dependent wild bootstrap.
The primary inferential results are Gaussian approximations for the studentized estimating-equation vector and the feasible unrestricted and restricted studentized vectors; the familiar chi-square limits for the Wald and LM statistics then follow by continuous mapping. 

\begin{ass}
\begin{itemize}
    \item[(i)] The asymptotic expansion \eqref{eq:asym_linear_theta} and the score central limit theorem (CLT) \eqref{eq:score_clt} hold, with $\bOmega_0$ positive definite.
    \item[(ii)] $\sup_t\E\|\bg_t(\btheta_0)\|^2<\infty$.
    \item[(iii)] $\sup_t\E\sup_{\btheta\in\Theta}\|(\partial/\partial\btheta')\bg_t(\btheta)\|^2<\infty$.
    \item[(iv)] $\btheta_0$ is an interior point of $\Theta$, and $\bg_t(\btheta)$ is continuously differentiable in a neighborhood of $\btheta_0$.
\end{itemize}
\label{ass:moment_weakdepend}
\end{ass}

Assumption~\ref{ass:moment_weakdepend}(i) restates the first-order framework in Section~\ref{sec:setup}.
Assumption~\ref{ass:moment_weakdepend}(iii) is a standard condition used to establish HAC consistency and asymptotic normality of smooth estimators \citep{Andrews1991,NeweyMcFadden1994}.

\begin{ass}[Bandwidth]
The bandwidth $b_T$ used in \eqref{eq:double_sum_hac} satisfies $b_T\to\infty$ and $b_T/T\to0$.
\label{ass:bandwidth}
\end{ass}

Using the same lag weights $K_p(\rho_{T,h})$ as in the feasible matched-HAC estimator, 
define the corresponding infeasible finite-sample HAC matrix based on $\bg_t(\btheta_0)$ by
\begin{align} \bOmega_{p,T} = \bGamma_T(0)+ \sum_{h=1}^{T-1}K_p(\rho_{T,h})\left\{\bGamma_T(h)+\bGamma_T(h)'\right\}, \end{align} 
where $\bGamma_T(h) = T^{-1}\sum_{t=h+1}^T\bg_t(\btheta_0)\bg_{t-h}(\btheta_0)'$.
Thus $\bOmega_{p,T}$ separates the standard HAC approximation problem from the additional effects of estimating and recentering the score contributions. 
Consistency of the feasible matched-HAC estimator, established in Theorem~\ref{thm:matched_hac_consistency} below, therefore has two ingredients: consistency of $\bOmega_{p,T}$ for $\bOmega_0$, and asymptotic negligibility of replacing the true-parameter scores by their estimated and recentered counterparts. The first ingredient is stated in the next assumption. 
Throughout this section, the latent correlation sequence $\{\rho_{T,h}\}$ is given by \eqref{eq:latent_gaussian_corr} and is used both to construct the bootstrap multipliers and to define the matched-HAC weights. 

The population-score HAC consistency condition is as follows.

\begin{ass}
$\bOmega_{p,T}\CP\bOmega_0$.
\label{ass:consistency_infHAC}
\end{ass}
Assumption~\ref{ass:consistency_infHAC} holds under standard weak-dependence conditions.
For example, Assumption A of \citet{Andrews1991} implies this condition when the lag-weight function $x\mapsto K_p\{\exp(-x^2)\}$ belongs to the $\cK_1$ class.
Lemma~\ref{lem:kernel_radem} verifies this property.

\begin{lem}
The function $x\mapsto K_p\{\exp(-x^2)\}$ is continuous, equals one at $x=0$, is symmetric in $x$, and satisfies $\int_{-\infty}^{\infty}|K_p\{\exp(-x^2)\}|^qdx<\infty$ for each $p\in(0,1)$ and finite $q>0$.
Therefore $x\mapsto K_p\{\exp(-x^2)\}$ belongs to the $\cK_1$ class of \citet{Andrews1991}.
\label{lem:kernel_radem}
\end{lem}

Lemma~\ref{lem:kernel_radem} shows that the HAC lag-weight function induced by the two-point multiplier construction satisfies the standard kernel regularity conditions. Hence, under conventional weak-dependence conditions, Assumption~\ref{ass:consistency_infHAC} follows from existing HAC theory. The next result shows that estimation and recentering do not affect this consistency asymptotically.


\begin{thm}[matched-HAC consistency]
Suppose that Assumptions~\ref{ass:moment_weakdepend}--\ref{ass:consistency_infHAC} hold. 
If $b_T/T^{1/2}\to0$, then $\what{\bOmega}_p\CP\bOmega_0$.
\label{thm:matched_hac_consistency}
\end{thm}

Matched-HAC consistency provides the covariance approximation needed for studentization, but bootstrap validity additionally requires a conditional Gaussian approximation to the bootstrap score sum. 
For this purpose, we strengthen the weak-dependence and moment conditions as follows.

\begin{ass}
\begin{itemize}

    \item[(i)] $\bg_t(\btheta_0)$ is strongly mixing with mixing coefficients $\alpha(h)$.
    \item[(ii)] There exists $\delta\geq2$ such that $\sup_t\E\|\bg_t(\btheta_0)\|^{2+\delta}<\infty$ and $\sum_{h=1}^{\infty}\alpha(h)^{\delta/(2+\delta)}<\infty$.
\end{itemize}
\label{ass:moment_stronger}
\end{ass}

Assumption~\ref{ass:moment_stronger} is the same as Assumption 3.1 of \citet{Shao2010}.
Under this stronger condition, the next theorem establishes the central bootstrap approximation at the level of the estimating-equation sum.

\begin{thm}[Conditional Gaussian approximation]
Suppose that Assumptions~\ref{ass:moment_weakdepend}--\ref{ass:moment_stronger} hold. 
If $b_T/T^{\delta/(2+2\delta)}\to0$, then
    ${T}^{-\frac{1}{2}}\sum_{t=1}^T\bg_t(\btheta_0)
    \CD
    N(\bzero,\bOmega_0)$, 
    ${T}^{-\frac{1}{2}}\sum_{t=1}^T\what{\bg}_t^*
    \CDs
    N(\bzero,\bOmega_0)$ in probability, 
and
\begin{align}
    \sup_{\bx\in\bbR^d}
    \left|
    \Pro\left(\frac{1}{\sqrt{T}}\sum_{t=1}^T\bg_t(\btheta_0)\leq\bx\right)
    -
    \Pro^*\left(
    \frac{1}{\sqrt{T}}\sum_{t=1}^T\what{\bg}_t^*
    \leq\bx
    \right)
    \right|
    \CP0.
    \label{eq:score_bootstrap_cdf_validity}
\end{align}
\label{thm:score_boot_consistent}
\end{thm}

Theorem~\ref{thm:score_boot_consistent} is stated at the score level and is the key bootstrap result. 
For asymptotically linear estimators, it transfers directly through the corresponding linear representation, yielding the following result.

\begin{cor}[Gaussian approximation for estimators]
Suppose that Assumptions~\ref{ass:moment_weakdepend}--\ref{ass:moment_stronger} hold, $\operatorname{rank}(\bB_0)=k$, 
$b_T/T^{\delta/(2+2\delta)}\to0$, and $\what{\bB}\CP\bB_0$.
Then 
    $\sqrt{T}(\what{\btheta}-\btheta_0)
    \CD N(\bzero,\bV_{\theta})$, 
    $\sqrt{T}(\what{\btheta}^*-\what{\btheta})
    \CDs
    N(\bzero,\bV_{\theta})$ 
    in probability, 
and
\begin{align}
    \sup_{\bx\in\bbR^k}
    \left|
    \Pro\left(\sqrt{T}(\what{\btheta}-\btheta_0)\leq\bx\right)
    -
    \Pro^*\left(\sqrt{T}(\what{\btheta}^*-\what{\btheta})\leq\bx\right)
    \right|
    \CP0.
    \label{eq:theta_bootstrap_cdf_validity}
\end{align}
\label{cor:theta_hat_boot_consistent}
\end{cor}

\subsection{Gaussian validity of the unrestricted and restricted studentized vectors}\label{subsec:wald_lm_validity}

The Gaussian approximation is the primary first-order result for inference, while the conventional Wald and LM chi-square laws follow by taking squared norms.

\begin{cor}[Unrestricted bootstrap validity]\label{cor:wald}
Assume that $H_0:\bA\btheta=\ba_0$ holds, where $\bA\in\bbR^{\ell\times k}$ has rank $\ell$, $\ba_0\in\bbR^\ell$, and $\ell$ is fixed.
Suppose that $\bA\bV_{\theta}\bA'$ is positive definite and $\what{\bV}_{\theta}\CP\bV_{\theta}$.
Under the conditions of Corollary~\ref{cor:theta_hat_boot_consistent}, let the bootstrap vector use either the fixed studentizer $\what{\bV}_{\theta}$ or a recomputed studentizer satisfying $\what{\bV}_{\theta}^*\CPs\bV_{\theta}$ in probability.
Then 
    $\what{\bz}_{W}
    \CD
    N(\bzero,\bI_\ell)$,
    $\what{\bz}_{W}^*
    \CDs
    N(\bzero,\bI_\ell)$ in probability, 
and
\begin{align}
    \sup_{\bx\in\bbR^\ell}
    \left|
    \Pro^*(\what{\bz}_{W}^*\leq\bx)-\Pro(\what{\bz}_{W}\leq\bx)
    \right|
    \CP0.
    \label{eq:wald_bootstrap_gaussian_validity}
\end{align}
Consequently, $W_T=\what{\bz}_{W}'\what{\bz}_{W}\CD\chi_\ell^2$, $W_T^*=\what{\bz}_{W}^{*'}\what{\bz}_{W}^*\CDs\chi_\ell^2$ in probability, and
\begin{align}
    \sup_{x \geq 0}
    \left|
    \Pro^*(W_T^*\leq x)-\Pro(W_T\leq x)
    \right|
    \CP0.
    \label{eq:wald_bootstrap_validity}
\end{align}
\end{cor}

For the restricted-score LM result, let
$\bG=(\bG_1,\bG_2)$ denote the probability limit of
$\wtilde{\bG}=(\wtilde{\bG}_1,\wtilde{\bG}_2)$ under $H_0$, and let
$\bG_{2\cdot1}
=
\bG_2
-
\bG_1
(\bG_1'\bOmega_0^{-1}\bG_1)^{-1}
\bG_1'\bOmega_0^{-1}\bG_2$.

\begin{cor}[Restricted bootstrap validity]\label{cor:lm}
Assume $H_0:\bA\btheta=\ba_0$ holds as in Corollary~\ref{cor:wald}.
Suppose that the conditions of Theorem~\ref{thm:score_boot_consistent} hold, $\wtilde{\bG}_j\CP\bG_j$ for $j\in\{1,2\}$, $\wtilde{\bOmega}_p\CP\bOmega_0$, $\bG_1'\bOmega_0^{-1}\bG_1$ is nonsingular, and $\bG_{2\cdot1}'\bOmega_0^{-1}\bG_{2\cdot1}$ is positive definite.
Suppose that the restricted score sum admits the following expansion with respect to $\bbbeta$ under the null:
\begin{align}
    \sqrt{T}\,\overline{\wtilde{\bg}}
    =
    \frac{1}{\sqrt{T}}\sum_{t=1}^T\bg_t(\btheta_0)
    +
    \bG_1\sqrt{T}(\wtilde{\boldsymbol{\eta}}_1-\boldsymbol{\eta}_{1,0})
    +
    o_p(1),
    \qquad
    \sqrt{T}(\wtilde{\boldsymbol{\eta}}_1-\boldsymbol{\eta}_{1,0})=O_p(1).
    \label{eq:lm_restricted_moment_linearization}
\end{align}
For recomputed studentization, suppose additionally that $\wtilde{\bOmega}_p^*\CPs\bOmega_0$ in probability.
Then the fixed-studentizer and recomputed-studentizer versions both satisfy 
    $\wtilde{\bz}_{LM}
    \CD
    N(\bzero,\bI_\ell)$,
    $\wtilde{\bz}_{LM}^*
    \CDs
    N(\bzero,\bI_\ell)$ 
    in probability, 
and
\begin{align}
    \sup_{\bx\in\bbR^\ell}
    \left|
    \Pro^*(\wtilde{\bz}_{LM}^*\leq\bx)-\Pro(\wtilde{\bz}_{LM}\leq\bx)
    \right|
    \CP0.
    \label{eq:lm_bootstrap_gaussian_validity}
\end{align}
Consequently, $LM_T=\wtilde{\bz}_{LM}'\wtilde{\bz}_{LM}\CD\chi_\ell^2$, $LM_T^*=\wtilde{\bz}_{LM}^{*'}\wtilde{\bz}_{LM}^*\CDs\chi_\ell^2$ in probability, and
\begin{align}
    \sup_{x \geq 0}
    \left|
    \Pro^*(LM_T^*\leq x)-\Pro(LM_T\leq x)
    \right|
    \CP0.
    \label{eq:lm_bootstrap_validity}
\end{align}
\end{cor}
For recomputed studentization, Corollaries~\ref{cor:wald} and \ref{cor:lm} impose high-level conditions requiring consistency of the recomputed studentizers under the bootstrap law. 
Section~\ref{subsec:reg_validity} provides primitive sufficient conditions for these requirements in the regression residual-bootstrap setting.


\subsection{Primitive validity of regression residual resampling}\label{subsec:reg_validity}

The preceding results concern direct score resampling, and we now give primitive conditions for the regression-specific residual bootstrap in Subsection~\ref{subsec:reg_residual}.
Under $H_0$, write
\begin{align}
    y_t=\bx_t'\btheta_0+u_t,
    \qquad
    \bA\btheta_0=\ba_0,
\end{align}
and retain the fixed matrix $\bC=(\bC_1,\bC_2)$ from Subsection~\ref{subsec:reg_residual}.
Throughout this subsection, $\bx_t=(1,\bz_t')'$ and the intercept is left unrestricted by $\bA$.
Write $\bQ_x=\E(\bx_t\bx_t')$ and $\bg_t=\bg_t(\btheta_0)=\bx_tu_t$.
For least squares, $\bB_0=\bQ_x^{-1}$, so $\bV_{\theta}=\bQ_x^{-1}\bOmega_0\bQ_x^{-1}$.
Let
\begin{align}
    \bG_1=-\bQ_x\bC_1,
    \qquad
    \bG_2=-\bQ_x\bC_2,
\end{align}
and define
\begin{align}
    \bG_{2\cdot1}
    =
    \bG_2
    -
    \bG_1(\bG_1'\bOmega_0^{-1}\bG_1)^{-1}\bG_1'\bOmega_0^{-1}\bG_2.
\end{align}
For this regression model,
\begin{align}
    \left(\bG_{2\cdot1}'\bOmega_0^{-1}\bG_{2\cdot1}\right)^{-1}
    =
    \bA\bV_{\theta}\bA'.
    \label{eq:reg_lm_wald_variance_relation}
\end{align}
Let $g_{t,j}$ denote the $j$th component of $\bg_t$, and let $\operatorname{cum}(g_{s,a},g_{t,b},g_{u,c},g_{v,d})$ denote the fourth-order cumulant of $(g_{s,a},g_{t,b},g_{u,c},g_{v,d})$.
The following conditions provide primitive sufficient conditions for the original-sample regression results.
\begin{ass}[Primitive regression conditions]
\label{ass:reg_additional}
\begin{itemize}
    \item[(i)] $\E(\bx_tu_t)=\bzero$ and $\bQ_x$ is positive definite.
    \item[(ii)] The process $\{(\bx_t,u_t)\}_{t\in\bbZ}$ is strictly stationary and strongly mixing with coefficients $\alpha(h)$, and for some $\delta\geq2$, $\E\|\bx_tu_t\|^{2+\delta}<\infty$ and $\sum_{h=1}^{\infty}\alpha(h)^{\delta/(2+\delta)}<\infty$.
    \item[(iii)] $\E\|\bx_t\|^4<\infty$ and $\bOmega_0$ is positive definite.
    \item[(iv)] For every $j_1,j_2,j_3,j_4\in\{1,\ldots,k\}$, 
        $\sum_{h_1,h_2,h_3\in\bbZ}
        \left|
        \operatorname{cum}(g_{0,j_1},g_{h_1,j_2},g_{h_2,j_3},g_{h_3,j_4})
        \right|
        <\infty$.
    \item[(v)] $b_T\to\infty$ and $b_T/T^{1/2}\to0$.
\end{itemize}
\end{ass}

For fixed $p\in(0,1)$, define the observed matched two-point HAC estimators
\begin{align}
    \what{\bOmega}_p
    &=
    {1\over T}\sum_{t=1}^T\sum_{s=1}^T
    K_p(\rho_{T,t-s})
    \what{\bg}_t\what{\bg}_s',
    \\
    \wtilde{\bOmega}_p
    &=
    {1\over T}\sum_{t=1}^T\sum_{s=1}^T
    K_p(\rho_{T,t-s})
    (\wtilde{\bg}_t-\overline{\wtilde{\bg}})
    (\wtilde{\bg}_s-\overline{\wtilde{\bg}})'.
    \label{eq:reg_observed_twopoint_hac}
\end{align}
We use the regression Wald and LM statistics defined in Subsection~\ref{subsec:reg_residual}.

The residual bootstrap requires the additional conditions below.
Condition (i) is needed for the conditional multiplier approximation, including fixed studentization, while conditions (ii)--(iii) are imposed only when the bootstrap HAC studentizer is recomputed.

\begin{ass}[Residual-bootstrap conditions]
\label{ass:reg_bootstrap}
\begin{itemize}
    \item[(i)] With $\delta$ as in Assumption~\ref{ass:reg_additional}(ii), $b_T/T^{\delta/(2+2\delta)}\to0$.
    \item[(ii)] For a recomputed bootstrap HAC studentizer, $\E\|\bx_t\|^{4+2\delta}+\E|u_t|^{4+2\delta}<\infty$.
    \item[(iii)] If $b_T^*$ denotes the bootstrap HAC bandwidth, then $b_T^*\to\infty$ and $b_T^*=O(b_T)$.
\end{itemize}
\end{ass}

Let $\what{\bOmega}_p^*$ denote the matched-HAC estimator recomputed from the unrestricted residual-bootstrap moment array $\{\what{\bg}_{t,\mathrm{res}}^*\}$ using bandwidth $b_T^*$, and set $\what{\bV}_{\theta}^*=\what{\bQ}_x^{-1}\what{\bOmega}_p^*\what{\bQ}_x^{-1}$.
Let $\wtilde{\bOmega}_p^*$ denote the matched-HAC estimator recomputed from the restricted residual-bootstrap moment array $\{\wtilde{\bg}_{t,\mathrm{res}}^*\}$ after subtracting its bootstrap sample average, using bandwidth $b_T^*$.
For the recomputed residual-bootstrap LM statistic, retain the original-sample projection $\wtilde{\bG}_{2\cdot1}'\wtilde{\bOmega}_p^{-1}$ and use the draw-specific projected covariance $\wtilde{\bG}_{2\cdot1}'\wtilde{\bOmega}_p^{-1}\wtilde{\bOmega}_p^*\wtilde{\bOmega}_p^{-1}\wtilde{\bG}_{2\cdot1}$ in \eqref{eq:lm_bootstrap_generic}.

\begin{prop}[Regression and residual-bootstrap Gaussian validity]
\label{prop:reg_validity}
Suppose Assumption~\ref{ass:reg_additional} holds. 
The studentized vectors below are the feasible statistics defined in Subsection~\ref{subsec:wald_lm_implementations}. 
\begin{itemize}
    \item[(i)] The unrestricted OLS estimator satisfies
    \begin{align}
        \sqrt T(\what{\btheta}-\btheta_0)
        =
        \bQ_x^{-1}{1\over\sqrt T}\sum_{t=1}^T\bg_t+o_p(1),
        \qquad
        \sqrt T(\what{\btheta}-\btheta_0)
        \CD
        N(\bzero,\bV_{\theta}).
        \label{eq:reg_theta_primitive_expansion}
    \end{align}
    Moreover,
    \begin{align}
        \what{\bOmega}_p&\CP\bOmega_0,
        \qquad
        \wtilde{\bOmega}_p\CP\bOmega_0,
        \qquad
        \what{\bV}_{\theta}\CP\bV_{\theta},
        \\
        \wtilde{\bG}_{2\cdot1}'\wtilde{\bOmega}_p^{-1}\wtilde{\bG}_{2\cdot1}
        &\CP
        \bG_{2\cdot1}'\bOmega_0^{-1}\bG_{2\cdot1}.
        \label{eq:reg_observed_hac_consistency}
    \end{align}
    Hence under $H_0$,
    \begin{align}
        \what{\bz}_{W}
        &\CD
        N(\bzero,\bI_\ell),
        \qquad
        \wtilde{\bz}_{LM}
        \CD
        N(\bzero,\bI_\ell),
        \label{eq:reg_observed_gaussian_limits}
    \end{align}
    and consequently $W_T\CD\chi_\ell^2$ and $LM_T\CD\chi_\ell^2$.
    \item[(ii)] If Assumption~\ref{ass:reg_bootstrap}(i) also holds, then the residual-bootstrap unrestricted and restricted studentized vectors with the observed matched-HAC studentizers held fixed satisfy
    \begin{align}
        \what{\bz}_{W}^*
        &\CDs
        N(\bzero,\bI_\ell),
        \qquad
        \wtilde{\bz}_{LM}^*
        \CDs
        N(\bzero,\bI_\ell)
        \label{eq:reg_boot_fixed_clt}
    \end{align}
    in probability.
    \item[(iii)] If Assumption~\ref{ass:reg_bootstrap}(ii)--(iii) also holds, then
    \begin{align}
        \what{\bOmega}_p^*&\CPs\bOmega_0,
        \qquad
        \wtilde{\bOmega}_p^*\CPs\bOmega_0,
        \qquad
        \what{\bV}_{\theta}^*\CPs\bV_{\theta},
        \\
        \wtilde{\bG}_{2\cdot1}'\wtilde{\bOmega}_p^{-1}\wtilde{\bOmega}_p^*\wtilde{\bOmega}_p^{-1}\wtilde{\bG}_{2\cdot1}
        &\CPs
        \bG_{2\cdot1}'\bOmega_0^{-1}\bG_{2\cdot1}.
        \label{eq:reg_boot_recomputed_hac_consistency}
    \end{align}
    Hence the recomputed residual-bootstrap studentized vectors also satisfy the conditional Gaussian limits in \eqref{eq:reg_boot_fixed_clt}.
    Consequently,
    \begin{align}
        \sup_{\bx\in\bbR^\ell}
        \left|
        \Pro^*(\what{\bz}_{W}^*\leq\bx)-\Pro(\what{\bz}_{W}\leq\bx)
        \right|
        &\CP0,
        \\
        \sup_{\bx\in\bbR^\ell}
        \left|
        \Pro^*(\wtilde{\bz}_{LM}^*\leq\bx)-\Pro(\wtilde{\bz}_{LM}\leq\bx)
        \right|
        &\CP0.
        \label{eq:reg_boot_cdf_validity}
    \end{align}
    In particular, $W_T^*\CDs\chi_\ell^2$ and $LM_T^*\CDs\chi_\ell^2$ in probability for both fixed and recomputed studentization.
\end{itemize}
\end{prop}

Proposition~\ref{prop:reg_validity} gives primitive Gaussian validity for the same fixed $\ell$ linear restrictions considered in the general theory, with the Wald and LM chi-square limits following as consequences. 
The proof is given in Appendix~\ref{app:reg_primitive}.

\subsection{Bandwidth choice}\label{sec:bandwidth}

The validity results above impose rate conditions on the bandwidth $b_T$.
For the Rademacher case, the arcsine kernel in \eqref{eq:arcsine_kernel} has a first-order cusp at the origin, as shown in \eqref{eq:kernel_cusp}, suggesting the standard first-order HAC bandwidth rate $b_T=\lceil c_bT^{1/3}\rceil$. 
This choice is permitted under Assumption~\ref{ass:moment_stronger} when $\delta>2$. 
In the main simulations, we set $c_b=1$.
Appendix~\ref{app:bandwidth_mc} also reports results for $c_b\in\{0.5,2\}$, showing that the main qualitative findings are broadly robust to the bandwidth choice.

For a general two-point law, the same latent Gaussian correlation profile induces the lag-weight function $x\mapsto K_p\{\exp(-x^2)\}$, with the transformation $K_p$ depending on the chosen marginal distribution.
In the simulations, the Mammen dependence scale is slightly adjusted to approximate the covariance profile induced by the dependent Rademacher multipliers, as described in Appendix~\ref{app:additional_mc}.

\section{Monte Carlo Experiments}\label{sec:mc}

This section examines the finite-sample performance of the proposed dependent two-point multiplier procedures in nonlinear GMM and linear regression models.
To keep the main comparison focused, we report results for a demanding design combining serial dependence, stationary conditional heteroskedasticity, and skewed innovations.
Results for alternative innovation distributions, together with two-step GMM results and a comparison with conventional LM and Wald bootstrap tests, are reported in Appendix~\ref{app:additional_mc}.

In both experiments, sample sizes are $T\in\{100,200,400\}$.
For each design, $R=5{,}000$ Monte Carlo replications are conducted with $B=499$ bootstrap draws per replication.
All reported tests concern a single restriction, $\ell=1$, and use either the restricted $z$-statistic $\wtilde z_{LM}$ or the unrestricted $z$-statistic $\what z_W$.
Unless otherwise stated, all tests are conducted at the nominal $5\%$ level.
We use the signed $z$-statistics as the primary finite-sample diagnostic because they preserve the direction of departures from the null and allow the lower and upper tails of the sampling approximation to be assessed separately.
The main results therefore report equal-tail two-sided $z$-tests: rejection occurs below the $\alpha/2$ critical value or above the $1-\alpha/2$ critical value, using the standard Normal reference distribution for asymptotic inference and the corresponding conditional bootstrap distribution for bootstrap inference.
For $\ell=1$, the conventional LM and Wald statistics are obtained by squaring $\wtilde z_{LM}$ and $\what z_W$, respectively, and hence fold the two directions together.
For comparison, the corresponding conventional LM and Wald bootstrap tests are reported in Appendix~\ref{app:symmetric_tests}.

The common bandwidth is $b_T=\lceil T^{1/3}\rceil$, which is also used as the moving-block-bootstrap block length.
The reference latent Gaussian correlation at lag $h\geq0$ is $\rho_{T,h}=\exp\{-(h/b_T)^2\}$.
The simulations compare asymptotic inference using the HAC lag weights induced by the dependent Rademacher multipliers with dependent Rademacher, dependent Mammen, Gaussian DWB, and moving-block-bootstrap procedures.

\subsection{Designs}\label{subsec:mc_designs}
To isolate the effects of the multiplier distribution,
all procedures within each experiment use the same observed
test statistic, with its HAC covariance estimated using the
lag weights $r_0(h)=K_{1/2}(\rho_{T,h})$ induced by the
dependent Rademacher multipliers
(see \eqref{eq:arcsine_kernel}).
The Gaussian multiplier covariance is matched exactly to
this profile, and the Mammen covariance is matched
approximately (see Appendix~\ref{app:additional_mc}).

\noindent\textit{Nonlinear GMM.}
The first experiment considers
    $y_t
    =
    \theta_{1,0}+\exp(\theta_{2,0}x_t)
    +
    \sigma_t u_t$, $t=1,\ldots,T$, 
where $\sigma_t^2=(1+x_{t-1}^2)/2$.
We set $\theta_{1,0}=0$ throughout and consider the test $H_0:\theta_{2,0}=0.5$, with $\theta_{2,0}=0.5$ for size and $\theta_{2,0}=0.7$ for power.
The regressor follows 
    $x_t
    =
    \rho_x x_{t-1}
    +
    \sqrt{1-\rho_x^2}\,v_t, 
    v_t\sim N(0,1)$,
with $\rho_x=0.8$, while 
    $u_t
    =
    \rho_u u_{t-1}
    +
    \sqrt{1-\rho_u^2} \varepsilon_t$,
with $\rho_u=0.5$. 
The innovation $\varepsilon_t$ follows one of three distributions, each standardized to have mean zero and variance one: Normal, Student-$t_5$, or centered $\chi_1^2$.
The processes $x_t$ and $u_t$ are initialized at zero and simulated for an additional 300 periods, which are discarded as burn-in before retaining the final sample of size $T$.
The main text reports the centered $\chi_1^2$ results, while the Normal and Student-$t_5$ results are given in Appendix~\ref{app:additional_mc}.

Let $\bz_t=(1,x_t,x_{t-1})'$ and define the moment contribution $\bg_t(\btheta)=\bz_t\{y_t-\theta_1-\exp(\theta_2x_t)\}$.
There are three moment conditions and two parameters, and the one-step GMM estimator uses the identity weighting matrix.
For $\wtilde z_{LM}$, $\theta_2$ is fixed at its null value and the nuisance parameter $\theta_1$ is estimated subject to this restriction; $\what z_W$ is based on the unrestricted one-step GMM estimator.
For the bootstrap versions, the corresponding centered
restricted or unrestricted moment contributions are resampled,
with the HAC studentizer recomputed in each draw using $r_0$.
As the moving-block-bootstrap (MBB) benchmark, we use the Hall--Horowitz recentered procedure \citep{HallHorowitz1996}.
Throughout the nonlinear GMM experiment and the empirical application, MBB refers to this recentered version.
The recentering subtracts the conditional bootstrap mean induced by unequal representation of observations near the sample boundaries, thereby imposing the bootstrap analogue of the moment condition.
Further details are reported in Appendix~\ref{app:mbb_recentering}.

\medskip
\noindent\textit{Linear regression.}
The second experiment considers 
    $y_t
    =
    \gamma+\theta x_t+\sigma_tu_t$, 
    $t=1,\ldots,T$, 
with $\gamma=0$.
The regressor $x_t$, error process $u_t$, and conditional heteroskedasticity are generated as in the nonlinear GMM experiment.
We test $H_0:\theta=0$ against $H_1:\theta\neq0$, evaluating size at $\theta=0$ and power at $\theta=0.20$.
For $\wtilde z_{LM}$, the regression is estimated under $H_0$ and the score is constructed using the restricted residuals, while $\what z_W$ is based on unrestricted OLS and its HAC covariance.
The dependent multiplier procedures use residual resampling with the $r_0$ HAC studentizer recomputed in each draw.
The MBB benchmark instead resamples blocks of the corresponding centered score contributions and uses the same studentizer.
Results for Normal and Student-$t_5$ innovations are reported in Appendix~\ref{app:additional_mc}.

\subsection{Results}\label{subsec:mc_results}

Table~\ref{tab:mc_size_main} reports empirical size for both experiments for tests conducted at the nominal $5\%$ level.
Size-adjusted power results are reported in Subsection~\ref{app:main_power}.

\begin{table}[!htb]
\centering
\caption{Empirical size (\%) of equal-tail two-sided $z$-tests under serially dependent heteroskedastic $\chi_1^2$ innovations}
\label{tab:mc_size_main}
\begin{threeparttable}
\small
\setlength{\tabcolsep}{4.2pt}
\renewcommand{\arraystretch}{1.03}

{
\begin{tabular}{
l
S[table-format=2.2]
S[table-format=2.2]
S[table-format=2.2]
S[table-format=2.2]
S[table-format=2.2]
S[table-format=2.2]
}
\toprule
&
\multicolumn{3}{c}{Nonlinear GMM}
&
\multicolumn{3}{c}{Linear regression}
\\
\cmidrule(lr){2-4}
\cmidrule(lr){5-7}
\multicolumn{1}{r}{$T$}
& {100} & {200} & {400}
& {100} & {200} & {400} \\
\midrule

\multicolumn{7}{l}{\textit{Restricted $z$-statistic, $\wtilde z_{LM}$}} \\

\multicolumn{7}{l}{\textit{Asymptotic}} \\
\quad ASY $r_0$
& 9.68 & 8.30 & 6.34
& 9.60 & 8.34 & 6.76 \\

\multicolumn{7}{l}{\textit{Bootstrap}} \\
\quad Dep.\ Rademacher
& 7.60 & 6.94 & 5.58
& 5.88 & 6.08 & 5.14 \\
\quad Dep.\ Mammen
& 14.08 & 12.18 & 9.54
& 13.38 & 12.10 & 9.88 \\
\quad Gaussian DWB
& 7.86 & 7.24 & 6.16
& 6.74 & 6.92 & 5.84 \\
\quad MBB
& 13.20 & 11.32 & 8.82
& 12.12 & 11.18 & 8.88 \\

\midrule

\multicolumn{7}{l}{\textit{Unrestricted $z$-statistic, $\what z_W$}} \\

\multicolumn{7}{l}{\textit{Asymptotic}} \\
\quad ASY $r_0$
& 15.70 & 12.24 & 8.66
& 14.42 & 11.84 & 8.58 \\

\multicolumn{7}{l}{\textit{Bootstrap}} \\
\quad Dep.\ Rademacher
& 14.06 & 11.22 & 7.84
& 8.68 & 7.58 & 6.02 \\
\quad Dep.\ Mammen
& 18.52 & 14.24 & 10.44
& 12.48 & 10.94 & 8.92 \\
\quad Gaussian DWB
& 14.76 & 11.60 & 8.46
& 9.54 & 8.30 & 6.88 \\
\quad MBB
& 15.32 & 12.10 & 9.22
& 14.92 & 13.04 & 9.72 \\
\bottomrule
\end{tabular}}

\begin{tablenotes}[flushleft]
\footnotesize
\item \textit{Notes:}
Entries are rejection frequencies at the nominal $5\%$ level
for $H_0:\theta_2=0.5$ (one-step GMM) and $H_0:\theta=0$ (OLS).
All procedures use the same observed statistic, with
HAC covariance estimated using the Rademacher-induced
lag weights $r_0$; the HAC estimator is recomputed
in each bootstrap draw.
Gaussian DWB has the same multiplier covariance as dependent
Rademacher; Mammen is approximately covariance-matched
(Appendix~\ref{app:additional_mc}).
MBB denotes the recentered score moving-block bootstrap.
Power results are reported in Subsection~\ref{app:main_power}.
\end{tablenotes}

\end{threeparttable}
\end{table}

For the restricted statistic $\wtilde z_{LM}$, the nonlinear GMM results show improved size control under Dep.\ Rademacher and Gaussian DWB relative to asymptotic inference.
At $T=100$, the asymptotic $r_0$ test rejects at $9.68\%$, whereas Dep.\ Rademacher and Gaussian DWB reduce the rejection frequencies to $7.60\%$ and $7.86\%$, respectively.
At $T=400$, the corresponding bootstrap rejection frequencies are $5.58\%$ and $6.16\%$.
Dep.\ Mammen and MBB remain appreciably oversized, particularly at the smaller sample sizes.

Turning to the linear regression experiment, for $\wtilde z_{LM}$, Dep.\ Rademacher gives rejection frequencies of $5.88\%$, $6.08\%$, and $5.14\%$ for $T=100$, $200$, and $400$, respectively.
Gaussian DWB is somewhat more oversized, while Dep.\ Mammen and MBB exhibit substantially larger size distortions.

Size distortions are generally larger for the unrestricted statistic $\what z_W$ than for the restricted statistic $\wtilde z_{LM}$.
Among the bootstrap procedures, Dep.\ Rademacher gives the smallest absolute size distortion in every comparison reported in Table~\ref{tab:mc_size_main}, although its advantage over Gaussian DWB is often modest.
Overall, the restricted statistic combined with Dep.\ Rademacher provides particularly accurate finite-sample size control in both experiments.

Size-adjusted power is broadly similar across procedures for the restricted statistic, although some differences remain for the unrestricted statistic, particularly at smaller sample sizes (Appendix~\ref{app:main_power}).

The results for Normal and Student-$t_5$ innovations reported in Appendix~\ref{app:additional_mc} are qualitatively similar: Dep.\ Rademacher continues to provide accurate size control for the restricted statistic, while size distortions are generally larger for the unrestricted statistic.
The corresponding two-step GMM results in Appendix~\ref{app:twostep_mc} also give qualitatively similar conclusions.

\section{Empirical Application: Mean Reversion in Short-Term Interest Rates}
\label{sec:empirical}

We illustrate the proposed bootstrap procedures using a nonlinear GMM specification motivated by the Cox--Ingersoll--Ross (CIR) model \citep{CoxIngersollRoss1985}. 
In continuous time, measured in years, the short rate $r(\tau)$ follows 
    $dr(\tau)
    =
    \kappa\{\mu-r(\tau)\}\,d\tau
    +
    \sigma\sqrt{r(\tau)}\,dW(\tau)$, 
where $\mu$ is the long-run mean level, $\kappa>0$ is the speed of mean reversion, $\sigma>0$ governs the instantaneous volatility, and $W(\tau)$ is a standard Brownian motion. 
The model implies the conditional mean 
    $\E\{r(\tau+\Delta)\mid\cF_\tau\}
    =
    \mu
    +
    \{r(\tau)-\mu\}\exp(-\kappa\Delta)$, 
where $\cF_\tau$ denotes the information available at time $\tau$ and $\Delta$ is the horizon measured in years. 
Thus, $\mu$ determines the level toward which the short rate reverts, while $\kappa$ determines the speed at which deviations from this level decay, with the associated mean-reversion half-life given by $\log(2)/\kappa$ years. 
We use only this conditional-mean implication and do not impose the CIR conditional-variance specification. 

The empirical analysis uses the monthly three-month Treasury bill secondary-market rate as an observable proxy for the short rate. 
The sample runs from January 1985 to August 2026; see Appendix~\ref{app:cir_data} for further details on the data and their construction. 
Let $r_t$ denote the observed rate in month $t$, and let $h$ denote the horizon measured in months. 
Evaluating the CIR conditional mean at the monthly observation dates gives 
    $\E_t(r_{t+h})
    =
    \mu
    +
    (r_t-\mu)\exp(-\kappa h/12)$, 
where $\E_t(\cdot):=\E(\cdot\mid\cF_t)$. 

Accordingly, we construct unconditional moment restrictions from this conditional-mean implication using functions of the current short rate as instruments. 
Let $x_t=r_t/10$, where this normalization defines the relative weighting of the moments under the identity-weighted criterion. We retain the same normalization in the two-step comparison. 
The implied moment restrictions are $\E\bg_t(\btheta_0)=\bzero$, where
\begin{align}
    \bg_t(\btheta)
    &=
    \bz_t
    \left[
        r_{t+h}
        -
        \mu
        -
        (r_t-\mu)\exp\left(-\kappa h/12\right)
    \right],
    \qquad
    \btheta=(\mu,\kappa)',
    \label{eq:cir_moment}
\end{align}
with $\bz_t=(1,x_t,x_t^2)'$. 
This gives three moment conditions for the two parameters $\mu$ and $\kappa$. 

We examine the model for $h=6,9,12$ months and test the null hypothesis of a two-year mean-reversion half-life, $H_0:\kappa=\log(2)/2$. 
For the empirical comparison, we focus on the restricted $z$-statistic $\wtilde z_{LM}$, which delivered the most accurate size control in the Monte Carlo experiments, particularly for the dependent Rademacher procedure. 
Table~\ref{tab:cir-h9} reports the main specification with $h=9$ months, while the corresponding results for $h=6$ and $h=12$ are reported in Appendix~\ref{app:cir_additional}.

The one-step GMM estimate implies a mean-reversion
half-life substantially longer than the two-year value
under the null.
The asymptotic test rejects at the $5\%$ level
($p=0.035$), whereas none of the four bootstrap
procedures rejects, with $p$-values between $0.076$
and $0.089$.
Thus, the empirical conclusion differs between
asymptotic and bootstrap inference, while the
bootstrap procedures give similar results.

Appendix~\ref{app:cir_additional} reports the one-step results for $h=6$ and $h=12$, together with a two-step GMM robustness check for
the main $h=9$ specification. 
The evidence varies with the horizon, while the two-step results preserve the main qualitative comparison between the asymptotic and bootstrap procedures.

\begin{table}[!htb]
\centering
\caption{Equal-tail two-sided $z$-tests of a two-year mean-reversion half-life, $h=9$ months}
\label{tab:cir-h9}
\small
\begin{threeparttable}
{
\begin{tabular}{lrrr}
\toprule
Method & $p$-value & $q_{0.025}$ & $q_{0.975}$ \\
\midrule
\multicolumn{4}{l}{\textit{Asymptotic}} \\
\quad ASY $r_0$ & \multicolumn{1}{l}{$0.035^{**}$} & -1.96 & 1.96 \\
\multicolumn{4}{l}{\textit{Bootstrap}} \\
\quad Dep.\ Rademacher & \multicolumn{1}{l}{0.089} & -2.41 & 2.40 \\
\quad Dep.\ Mammen & \multicolumn{1}{l}{0.080} & -2.34 & 2.26 \\
\quad Gaussian DWB & \multicolumn{1}{l}{0.076} & -2.32 & 2.39 \\
\quad MBB & \multicolumn{1}{l}{0.084} & -2.39 & 2.31 \\
\bottomrule
\end{tabular}}
\begin{tablenotes}[flushleft]
\footnotesize
\item \textit{Notes:} The sample is January 1985--August 2026 ($T=491$).
The one-step GMM estimates are $\what\mu=2.65$ and $\what\kappa=0.18$, implying a mean-reversion half-life of 3.77 years.
All procedures use the same observed restricted $z$-statistic $\wtilde z_{LM}=-2.11$, computed using the HAC covariance estimator with Rademacher-induced lag weights, as in Section~\ref{sec:mc}.
Bootstrap $p$-values and quantiles are based on 9,999 draws.
$^{**}$ denotes rejection at the 5\% level.
\end{tablenotes}
\end{threeparttable}
\end{table}

\section{Conclusion}

In this paper, we have developed a general two-point dependent wild bootstrap (DWB) for weakly dependent estimating equations. 
Its key feature is that the two-point marginal law and the latent dependence profile can be specified separately. 
The construction combines a normalized two-point distribution with a stationary latent Gaussian process through a Gaussian copula transformation, includes dependent Rademacher and Mammen multipliers as special cases, and nests the classical iid two-point wild bootstrap under serial independence. 
The induced multiplier autocovariances determine the corresponding HAC lag weights, and the resulting matched-HAC estimator coincides exactly with the conditional covariance of the bootstrap estimating-equation sum. 
We establish consistency of this covariance estimator and a conditional Gaussian approximation for the bootstrap score sum, yielding first-order validity for HAC-studentized $z$-statistics and the corresponding Wald and LM statistics.

The Monte Carlo results complement the first-order
validity theory by showing that the choice of multiplier
distribution can affect finite-sample inference.
With a common HAC studentizer and matched multiplier
covariance profiles, dependent Rademacher generally
provides more accurate size control than Mammen and
Gaussian DWB, particularly for the restricted
$z$-statistic $\wtilde z_{LM}$.
Size-adjusted power is broadly comparable for this
statistic, and the main findings are reasonably
robust to bandwidth choice.
In the short-rate application, the asymptotic test
rejects the two-year mean-reversion half-life at the
main specification, whereas none of the four bootstrap
procedures rejects.


The generality of the two-point construction motivates several extensions currently under development. 
These include higher-order analysis of dependent wild bootstrap inference, extensions to large panels, and bounded two-point multipliers under both cross-sectional and temporal dependence.
These developments are being pursued separately and extend the scope of the proposed construction beyond the weakly dependent time-series setting considered here.

\section*{Acknowledgments}
We are grateful to Yukitoshi Matsushita for helpful discussions and useful comments. 
This work was supported by JSPS KAKENHI (grant numbers 25KJ0041, 25K00625, 25K05036 and 25H00544). 

\section*{Data Availability and Disclosure Statement}
All data used in this study are publicly available. Data sources and replication code are provided in the Online Appendix and supplementary material. 
The authors declare no conflicts of interest.

\section*{Generative AI disclosure}
Generative AI tools, including ChatGPT (OpenAI) and Claude (Anthropic), were used for language editing and coding assistance in the Monte Carlo experiments and empirical analysis.
All research design, methodology, analysis, interpretation, and verification were undertaken by the authors, who take full responsibility for the paper and code.

\bibliographystyle{apalike}
\bibliography{DWB}

\newpage
\appendix
\setcounter{page}{1}
\setcounter{section}{0}
\setcounter{footnote}{0}
\renewcommand{\thesection}{\Alph{section}}
\renewcommand{\theequation}{\thesection.\arabic{equation}}
\renewcommand{\thetable}{\thesection.\arabic{table}}
\renewcommand{\thefigure}{\thesection.\arabic{figure}}

\makeatletter
\@addtoreset{equation}{section}
\@addtoreset{table}{section}
\@addtoreset{figure}{section}
\makeatother
\begin{center}
	{\Large Supplementary Material for} \\[7mm]
	\textbf{\Large Two-Point Dependent Wild Bootstrap for Weakly Dependent Estimating Equations} \\[10mm]
	\textsc{\large Mikihito Nishi$^*$} {\large and} \textsc{\large Takashi Yamagata$^\dagger$} \\[5mm]
	*\textit{\large Graduate School of Economics, University of Tokyo} \\[1mm]
	$\dagger$\textit{\large Department of Economics and Related Studies, University of York} \\[1mm]
	$\dagger$\textit{\large Graduate School of Economics and Management, Tohoku University}
\end{center}

\section{Supplementary Implementation and Mathematical Proofs}\label{sec:proof}

\subsection{One-step GMM benchmark}\label{app:onestep_gmm}

The main nonlinear-GMM experiment and empirical application use identity-weighted one-step GMM to separate the choice of estimation weight from bootstrap approximation. This subsection gives the general-weight formulas; efficient two-step GMM results are reported as robustness checks.
Let $\bW_1$ be a fixed positive-definite matrix, with $\bW_1=\bI_d$ in the simulations, and define
\begin{align}
    \what{\btheta}^{(1)}
    &=
    \arg\min_{\btheta\in\Theta}
    \overline{\bg}_T(\btheta)'\bW_1\overline{\bg}_T(\btheta).
    \label{eq:one_step_gmm_estimator}
\end{align}
Writing $\bD_0=\E\{\partial\bg_t(\btheta_0)/\partial\btheta'\}$, the first-order map is
\begin{align}
    \bB_0^{(1)}
    &=
    -(\bD_0'\bW_1\bD_0)^{-1}\bD_0'\bW_1,
\end{align}
so the asymptotic covariance is $\bV_{\theta}^{(1)}=\bB_0^{(1)}\bOmega_0\bB_0^{(1)\prime}$.
In the sample, define $\what{\bg}_t^{(1)}=\bg_t(\what{\btheta}^{(1)})$, $\what{\bD}^{(1)}=T^{-1}\sum_t\partial\bg_t(\what{\btheta}^{(1)})/\partial\btheta'$, and
\begin{align}
    \what{\bB}^{(1)}
    &=
    -\{\what{\bD}^{(1)\prime}\bW_1\what{\bD}^{(1)}\}^{-1}\what{\bD}^{(1)\prime}\bW_1.
\end{align}
Let $\what{\bOmega}_p^{(1)}$ be the matched-HAC estimator constructed from $\{\what{\bg}_t^{(1)}-\overline{\what{\bg}^{(1)}}\}_{t=1}^T$ and set $\what{\bV}_{\theta}^{(1)}=\what{\bB}^{(1)}\what{\bOmega}_p^{(1)}\what{\bB}^{(1)\prime}$.
The one-step Wald statistic is obtained from \eqref{eq:wald_stat_generic} by replacing $\what{\btheta}$ and $\what{\bV}_{\theta}$ with $\what{\btheta}^{(1)}$ and $\what{\bV}_{\theta}^{(1)}$.

The direct one-step bootstrap uses
\begin{align}
    \what{\bg}_t^{(1)*}
    &=
    \xi_t^*\{\what{\bg}_t^{(1)}-\overline{\what{\bg}^{(1)}}\},
    \\
    \sqrt{T}\{\what{\btheta}^{(1)*}-\what{\btheta}^{(1)}\}
    &=
    \what{\bB}^{(1)}T^{-1/2}\sum_{t=1}^T\what{\bg}_t^{(1)*}.
    \label{eq:one_step_gmm_bootstrap_update}
\end{align}
For recomputed studentization, construct $\what{\bOmega}_p^{(1)*}$ from the bootstrap moment array and set $\what{\bV}_{\theta}^{(1)*}=\what{\bB}^{(1)}\what{\bOmega}_p^{(1)*}\what{\bB}^{(1)\prime}$, while the fixed version uses $\what{\bV}_{\theta}^{(1)}$ throughout.
Thus the estimated first-order map $\what{\bB}^{(1)}$ is held fixed across bootstrap draws exactly as in the efficient two-step implementation.

For the restricted one-step implementation, let $\wtilde{\btheta}^{(1)}$ minimize the same criterion subject to $\boldsymbol{\eta}_2=\ba_0$, define $\wtilde{\bg}_t^{(1)}=\bg_t(\wtilde{\btheta}^{(1)})$, and write the restricted Jacobian as $\wtilde{\bG}^{(1)}=(\wtilde{\bG}_1^{(1)},\wtilde{\bG}_2^{(1)})$. 
Define
\begin{align}
    \wtilde{\bG}_{2\cdot1}^{(1)}
    &=
    \wtilde{\bG}_2^{(1)}
    -
    \wtilde{\bG}_1^{(1)}
    \{\wtilde{\bG}_1^{(1)\prime}\bW_1\wtilde{\bG}_1^{(1)}\}^{-1}
    \wtilde{\bG}_1^{(1)\prime}\bW_1\wtilde{\bG}_2^{(1)}.
    \label{eq:one_step_gmm_partial_direction}
\end{align}

Let $\wtilde{\bOmega}_p^{(1)}$ be the matched-HAC estimator based on the centered restricted one-step moments.
Because $\bW_1$ is not generally equal to $\bOmega_0^{-1}$, the projected-score covariance retains the sandwich form
\begin{align}
    \wtilde{\bJ}^{(1)}
    &=
    \wtilde{\bG}_{2\cdot1}^{(1)\prime}
    \bW_1
    \wtilde{\bOmega}_p^{(1)}
    \bW_1
    \wtilde{\bG}_{2\cdot1}^{(1)}.
\end{align}
The corresponding one-step restricted studentized vector is
\begin{align}
    \wtilde{\bz}_{LM}^{(1)}
    &=
    -\{\wtilde{\bJ}^{(1)}\}^{-1/2}
    \wtilde{\bG}_{2\cdot1}^{(1)\prime}
    \bW_1
    \sqrt{T}\,
    \overline{\wtilde{\bg}^{(1)}},
\end{align}
with
\begin{align}
    LM_T^{(1)}
    &=
    \wtilde{\bz}_{LM}^{(1)\prime}
    \wtilde{\bz}_{LM}^{(1)}.
\end{align}
For the direct one-step LM bootstrap, define $\wtilde{\bg}_t^{(1)*}=\xi_t^*\{\wtilde{\bg}_t^{(1)}-\overline{\wtilde{\bg}^{(1)}}\}$ and let $\wtilde{\bOmega}_p^{(1)*}$ be its draw-specific HAC estimator.
Holding $\bW_1$ and $\wtilde{\bG}_{2\cdot1}^{(1)}$ fixed, the recomputed bootstrap statistic is
\begin{align}
    LM_T^{(1)*}
    &=
    T\overline{\wtilde{\bg}^{(1)*}}'
    \bW_1
    \wtilde{\bG}_{2\cdot1}^{(1)}
    \left\{
    \wtilde{\bG}_{2\cdot1}^{(1)\prime}
    \bW_1
    \wtilde{\bOmega}_p^{(1)*}
    \bW_1
    \wtilde{\bG}_{2\cdot1}^{(1)}
    \right\}^{-1}
    \wtilde{\bG}_{2\cdot1}^{(1)\prime}
    \bW_1
    \overline{\wtilde{\bg}^{(1)*}}.
    \label{eq:one_step_gmm_lm_bootstrap}
\end{align}
The fixed-studentizer version replaces $\wtilde{\bOmega}_p^{(1)*}$ in the middle matrix by $\wtilde{\bOmega}_p^{(1)}$.
If the weighting matrix is replaced by an efficient sequence converging to $\bOmega_0^{-1}$, this sandwich collapses to the simpler $\bOmega_0^{-1}$ geometry used in the main two-step development.
The efficient formulas in the theoretical development are therefore a special case; the main numerical results use the general-weight formulas above, and the appendix compares them with two-step GMM.
In Experiment 1 the one-step benchmark sets $\bW_1=\bI_3$ in both unrestricted and restricted estimation, with all DGPs, multipliers, bandwidths, and bootstrap draw counts otherwise identical to the two-step specification.

\subsection{General proofs}

\noindent\begin{proof}[Proof of Lemma \ref{lem:kernel_radem}]
    The first three assertions are immediate. Let us prove the last statement. Fix $p\in(0,1)$ and a finite $q>0$. Since $K_p(0)=0$, the derivative of the bivariate normal distribution function and a first-order Taylor expansion at zero give
    \begin{align}
        K_p(u)
        =
        \frac{\exp(-q_p^2)}{2\pi p(1-p)}u+o(u),
        \qquad u\downarrow0.
    \end{align}
    Hence, for some $\eta\in(0,1)$ and constant $C_p\in(0,\infty)$,
    \begin{align}
        |K_p(u)|\leq C_pu
        \quad\text{for all}\quad u\in[0,\eta].
    \end{align}
    Because $0\leq \exp(-x^2)\leq\eta$ for $|x|\geq\sqrt{\log\eta^{-1}}\eqqcolon\delta$, we have $|K_p\{\exp(-x^2)\}|\leq C_p \exp(-x^2)$ for all $|x|\geq \delta$. Additionally, noting that $\exp(x^2)|K_p\{\exp(-x^2)\}|$ is continuous on $|x|\leq \delta$ and thus bounded on this interval, we have $|K_p\{\exp(-x^2)\}|\leq M\exp(-x^2)$ for some constant $M\in(0,\infty)$. Letting $C\coloneqq \max\{C_p,M\}<\infty$, we get
    \begin{align}
        |K_p\{\exp(-x^2)\}|\leq C\exp(-x^2), \label{bound:kp_exp_uni}
    \end{align}
    for all $x\in\bbR$, and consequently
    \begin{align}
        \int_{-\infty}^{\infty}|K_p\{\exp(-x^2)\}|^qdx \leq C^q\int_{-\infty}^{\infty}\exp(-qx^2)dx<\infty.
    \end{align}
    
    This completes the proof.
\end{proof}

\noindent\begin{proof}[Proof of Theorem \ref{thm:matched_hac_consistency}]
For $h\geq0$, define
\begin{align}
    \bGamma_T(h;\btheta)
    &=
    \frac{1}{T}\sum_{t=h+1}^T
    \bg_t(\btheta)\bg_{t-h}(\btheta)',
    \qquad
    \bGamma_T(-h;\btheta)
    =
    \bGamma_T(h;\btheta)'.
\end{align}
Define
\begin{align}
    \bOmega_{p,T}(\btheta)
    &=
    \bGamma_T(0;\btheta)
    +
    \sum_{h=1}^{T-1}K_p(\rho_{T,h})
    \left\{\bGamma_T(h;\btheta)+\bGamma_T(-h;\btheta)\right\}.
\end{align}
Thus $\bOmega_{p,T}=\bOmega_{p,T}(\btheta_0)$.
We decompose \begin{align} \what{\bOmega}_p-\bOmega_0 &= \left\{ \what{\bOmega}_p-\bOmega_{p,T}(\what{\btheta}) \right\} + \left\{ \bOmega_{p,T}(\what{\btheta})-\bOmega_{p,T} \right\} + \left\{ \bOmega_{p,T}-\bOmega_0 \right\}. \label{eq:matched_hac_decomposition} \end{align} The last term is $o_p(1)$ by Assumption~\ref{ass:consistency_infHAC}. 
We show below that the middle term, due to parameter estimation, and the first term, due to recentering, are also $o_p(1)$. 
We first show $\bOmega_{p,T}(\what{\btheta})-\bOmega_{p,T}\stackrel{p}{\to}\bzero$.
Assumption \ref{ass:consistency_infHAC} then implies $\bOmega_{p,T}(\what{\btheta})\stackrel{p}{\to}\bOmega_0$.
It is sufficient to establish the first convergence after premultiplication and postmultiplication by an arbitrary fixed $\bb\in\bbR^d$, so in what follows we suppress this scalar projection.

Given $b_T/T^{1/2}\to0$, it suffices to show
\begin{align}
    \frac{\sqrt{T}}{b_T}
    \left\{
    \bOmega_{p,T}(\what{\btheta})-\bOmega_{p,T}
    \right\}
    =
    O_p(1).
\end{align}
A mean value expansion of $\bOmega_{p,T}(\what{\btheta})$ around $\btheta_0$ yields
\begin{align}
    \frac{\sqrt{T}}{b_T}
    \left\{
    \bOmega_{p,T}(\what{\btheta})-\bOmega_{p,T}
    \right\}
    &=
    \frac{1}{b_T}
    \frac{\partial}{\partial\btheta'}
    \bOmega_{p,T}(\overline{\btheta})
    \sqrt{T}(\what{\btheta}-\btheta_0)\\
    &=
    \frac{1}{b_T}
    \sum_{h=-T+1}^{T-1}
    K_p(\rho_{T,h})
    \left.
    \frac{\partial}{\partial\btheta'}
    \bGamma_T(h;\btheta)
    \right|_{\btheta=\overline{\btheta}}
    \sqrt{T}(\what{\btheta}-\btheta_0),
    \label{eq:mv_expansion_ddot}
\end{align}
for some $\overline{\btheta}$ on the line segment joining $\what{\btheta}$ and $\btheta_0$.

Applying the Cauchy--Schwarz inequality gives
\begin{align}
    &\sup_{|h|<T}
    \left\|
    \left.
    \frac{\partial}{\partial\btheta'}
    \bGamma_T(h;\btheta)
    \right|_{\btheta=\overline{\btheta}}
    \right\|\\
    &\quad\leq
    2\left\{
    \frac{1}{T}\sum_{t=1}^T
    \sup_{\btheta\in\Theta}
    \left\|
    \frac{\partial}{\partial\btheta'}
    \bg_t(\btheta)
    \right\|^2
    \right\}^{1/2}
    \left\{
    \frac{1}{T}\sum_{t=1}^T
    \sup_{\btheta\in[\btheta_0,\what{\btheta}]}
    \|\bg_t(\btheta)\|^2
    \right\}^{1/2}
    =
    O_p(1),
    \label{eq:gamma_derivative_bound}
\end{align}
where the last probability order follows from Assumptions \ref{ass:moment_weakdepend}(i)--(iii).
In particular, the mean value theorem gives
\begin{align}
    T^{-1}\sum_{t=1}^T
    \sup_{\btheta\in[\btheta_0,\what{\btheta}]}
    \|\bg_t(\btheta)\|^2
    &\leq
    2T^{-1}\sum_{t=1}^T
    \|\bg_t(\btheta_0)\|^2\\
    &\quad+
    2\|\what{\btheta}-\btheta_0\|^2
    T^{-1}\sum_{t=1}^T
    \sup_{\vartheta\in\Theta}
    \|\bD_t(\vartheta)\|^2
    =
    O_p(1),
\end{align}
where
\begin{align}
    \bD_t(\btheta)
    =
    \frac{\partial\bg_t(\btheta)}{\partial\btheta'}.
\end{align}
This result, Assumption \ref{ass:moment_weakdepend}(i), and
\begin{align}
    \frac{1}{b_T}
    \sum_{h=-T+1}^{T-1}
    \left|
    K_p(\rho_{T,h})
    \right|
    =
    O(1),
    \label{eq:kernel_rowsum_scaled}
\end{align}
by \eqref{bound:kp_exp_uni} imply that the right-hand side of \eqref{eq:mv_expansion_ddot} is $O_p(1)$.
Hence
\begin{align}
    \bOmega_{p,T}(\what{\btheta})-\bOmega_{p,T}
    =
    O_p\left(\frac{b_T}{\sqrt{T}}\right)
    =
    o_p(1),
    \label{eq:ddot_tilde_difference}
\end{align}
and Assumption \ref{ass:consistency_infHAC} yields $\bOmega_{p,T}(\what{\btheta})\stackrel{p}{\to}\bOmega_0$.

It remains to show that recentering is asymptotically negligible, namely,
\begin{align}
    \what{\bOmega}_p-\bOmega_{p,T}(\what{\btheta})
    \stackrel{p}{\to}
    \bzero.
\end{align}
For each $t$, the mean value theorem gives some $\overline{\btheta}_t$ on the line segment joining $\btheta_0$ and $\what{\btheta}$ such that
\begin{align}
    \bg_t(\what{\btheta})
    =
    \bg_t(\btheta_0)
    +
    \bD_t(\overline{\btheta}_t)
    (\what{\btheta}-\btheta_0).
\end{align}
Therefore,
\begin{align}
    \overline{\what{\bg}}
    &=
    \frac{1}{T}\sum_{t=1}^T
    \bg_t(\btheta_0)
    +
    \left\{
    \frac{1}{T}\sum_{t=1}^T
    \bD_t(\overline{\btheta}_t)
    \right\}
    (\what{\btheta}-\btheta_0)
    =
    O_p(T^{-1/2}),
    \label{eq:gbar_rate}
\end{align}
because
\begin{align}
    \frac{1}{T}\sum_{t=1}^T
    \|\bD_t(\overline{\btheta}_t)\|^2
    &\leq
    \frac{1}{T}\sum_{t=1}^T
    \sup_{\btheta\in\Theta}
    \|\bD_t(\btheta)\|^2
    =
    O_p(1),
    \label{eq:intermediate_jacobian_bound}
\end{align}
while $T^{-1}\sum_{t=1}^T\bg_t(\btheta_0)=O_p(T^{-1/2})$ and $\what{\btheta}-\btheta_0=O_p(T^{-1/2})$ by Assumption \ref{ass:moment_weakdepend}(i).

Expanding the centered products gives
\begin{align}
    \what{\bOmega}_p-\bOmega_{p,T}(\what{\btheta})
    ={}&
    -\frac{1}{T}
    \sum_{t=1}^T\sum_{s=1}^T
    K_p(\rho_{T,t-s})
    \what{\bg}_t
    \overline{\what{\bg}}'\\
    &-
    \frac{1}{T}
    \sum_{t=1}^T\sum_{s=1}^T
    K_p(\rho_{T,t-s})
    \overline{\what{\bg}}
    \what{\bg}_s'\\
    &+
    \frac{1}{T}
    \sum_{t=1}^T\sum_{s=1}^T
    K_p(\rho_{T,t-s})
    \overline{\what{\bg}}
    \overline{\what{\bg}}'.
    \label{eq:centered_uncentered}
\end{align}
Equation \eqref{eq:kernel_rowsum_scaled} implies
\begin{align}
    \max_{1\leq t\leq T}
    \sum_{s=1}^T
    |K_p(\rho_{T,t-s})|
    =
    O(b_T).
    \label{eq:kernel_rowsum}
\end{align}
Moreover, $T^{-1}\sum_{t=1}^T\|\what{\bg}_t\|=O_p(1)$ by the mean value theorem.
Therefore, \eqref{eq:gbar_rate} and \eqref{eq:centered_uncentered} imply
\begin{align}
    \left\|
    \what{\bOmega}_p-\bOmega_{p,T}(\what{\btheta})
    \right\|
    =
    O_p\left(
    \frac{b_T}{\sqrt{T}}
    +
    \frac{b_T}{T}
    \right)
    =
    o_p(1).
    \label{eq:center_difference_rate}
\end{align}
Combining this result with \eqref{eq:ddot_tilde_difference} and Assumption \ref{ass:consistency_infHAC} proves $\what{\bOmega}_p\stackrel{p}{\to}\bOmega_0$.
\end{proof}

\noindent\begin{proof}[Proof of Theorem \ref{thm:score_boot_consistent}]
For this proof, write $\bs_T^*=T^{-1/2}\sum_{t=1}^T\what{\bg}_t^*$.
Let $\by\sim N(\bzero,\bOmega_0)$. Assumption
\ref{ass:moment_weakdepend}(i) and the P\'olya theorem imply
\begin{align}
    \sup_{\bx\in\bbR^d}\left|
    \Pro\left(\frac1{\sqrt{T}}\sum_{t=1}^T\bg_t(\btheta_0)\leq \bx\right)
    -\Pro(\by \leq \bx)\right|\to0.                              
    \label{eq:polya_original}
\end{align}
It therefore suffices to prove that, under the bootstrap law conditional on the
data, $\bs_T^*\CDs N(\bzero,\bOmega_0)$ in probability.

We first establish the conditional CLT for the infeasible population-score bootstrap vector
\begin{align}
    \bm_T^*=\frac1{\sqrt{T}}\sum_{t=1}^T
       \xi_t^*\bg_t(\btheta_0).                              
    \label{eq:MTstar_def}
\end{align}
The core idea is to follow the proof of Theorem 3.1 of \citet{Shao2010}. However, his result cannot be directly used because the lag-weight function $x\mapsto K_p\{\exp(-x^2)\}$ has noncompact support, and hence the multipliers are not $b_T$-dependent. Borrowing from \citet{buckleyWindingStationaryGaussian2018} and \citet{kurisuGaussianApproximationSpatially2024}, we handle this by using an $m$-dependent approximation and controlling the omitted part in conditional $L^2$.

\noindent\textbf{Step 1: Approximation by the truncated multiplier process.}
Because a zero-mean stationary Gaussian process is determined by its covariance, we may write the latent Gaussian process, without changing its distribution, as
\begin{align}
    Z_{T,t}=\int_{\bbR}h_T(t-u)\,dW(u),\qquad
    h_T(u)=\left(\frac{4}{\pi b_T^2}\right)^{1/4}
            \exp\left(-\frac{2u^2}{b_T^2}\right),
    \label{eq:gaussian_ma_rep}
\end{align}
where $W$ is a two-sided standard Brownian motion.  Direct calculation gives
$\int h_T^2=1$ and $\E(Z_{T,t}Z_{T,s})=\exp\{-((t-s)/b_T)^2\}=\rho_{T,t-s}$. For any $M>0$, set
\begin{align}
    \nu_M^2&=\int_{|u|\leq Mb_T}h_T(u)^2\,du,\\
    h_{T,M}(u)&=\nu_M^{-1}h_T(u)1\{|u|\leq Mb_T\},\\
    Z_{T,t}^{(M)}&=\int_{\bbR}h_{T,M}(t-u)\,dW(u),\qquad
    \xi_{T,t}^{*(M)}=F_p^{-1}\{\Phi(Z_{T,t}^{(M)})\}.
    \label{eq:truncated_multiplier}
\end{align}
Each $\xi_{T,t}^{*(M)}$ has the same marginal two-point law $F_p$ as $\xi_t^*$, and the sequence is $m_{T,M}$-dependent with $m_{T,M}=\lceil2Mb_T\rceil$. In this sense, $\xi_{T,t}^{*(M)}$ is the ``truncated'' version of $\xi_{T,t}^*$.

The truncation error is exponentially small in $M$, uniformly in $T$. Indeed, $\Corr(Z_{T,t},Z_{T,t}^{(M)})=\nu_M$ and Gaussian tail bounds give $1-\nu_M\leq C\exp(-c_0M^2)$.  For any $\varepsilon>0$, $Z_{T,t}$ and $Z_{T,t}^{(M)}$ are on different sides of $q_p$ only if $|Z_{T,t}-q_p|\leq\varepsilon$ or $|Z_{T,t}-Z_{T,t}^{(M)}|>\varepsilon$. Gaussian anti-concentration and Chebyshev's inequality therefore yield
\begin{align}
 \E^*|\xi_t^*-\xi_{T,t}^{*(M)}|^2 
 &\leq \Pro\left(\left|Z_{T,t} - q_p\right|\leq \varepsilon \text{ or } \left|Z_{T,t} - Z_{T,t}^{(M)}\right|>\varepsilon\right) \left|a_p-b_p\right|^2 \\
 &\leq C_p\left(\varepsilon + \frac{\Var(Z_{T,t} - Z_{T,t}^{(M)})}{\varepsilon^2}\right) \\
 &\leq C_p'\left\{\varepsilon+
       \frac{1-\nu_M}{\varepsilon^2}\right\},
\end{align}
where $C_p'$ is a finite constant depending only on $p$. Taking $\varepsilon=(1-\nu_M)^{1/3}$, we have that, for constants $C,c_1>0$,
\begin{align}
    \eta_M:=\sup_{T,t}\E^*|\xi_t^*-\xi_{T,t}^{*(M)}|^2
    \leq C\exp(-c_1M^2)\longrightarrow0 \ \text{ as } M\to \infty.
    \label{eq:multiplier_trunc_error}
\end{align}

Since $b_T/T^{\delta/(2+2\delta)}\to0$, choose the deterministic sequence
\begin{align}
    R_T=\left(\frac{T^{\delta/(2+2\delta)}}{b_T}\right)^{1/2}.
    \label{eq:diagonal_RT}
\end{align}
Then $R_Tb_T/T^{\delta/(2+2\delta)}=(b_T/T^{\delta/(2+2\delta)})^{1/2}\to0$ by the hypothesis.  Fix $\ba\in\bbR^d$ and put
$Y_t=\ba'\bg_t(\btheta_0)$, $m_T^*(\ba)=\ba'\bm_T^*=T^{-1/2}\sum_{t=1}^T\xi_t^*Y_t$, and
\begin{align}
    m_{T,R_T}^*(\ba)=\frac1{\sqrt{T}}\sum_{t=1}^T
       \xi_{T,t}^{*(R_T)}Y_t.
    \label{eq:diagonal_multiplier_stat}
\end{align}
Assumption \ref{ass:moment_stronger} and Davydov's inequality imply $\sum_{h\in\bbZ}|\E(Y_tY_{t-h})|<\infty$. Since the data and multipliers are independent,
\begin{align}
 \E\left[\E^*|m_T^*(\ba)-m_{T,R_T}^*(\ba)|^2\right]
 &\leq \frac{\eta_{R_T}}{T}\sum_{t,s=1}^T|\E(Y_tY_s)|\\
 &\leq \eta_{R_T}\sum_{h\in\bbZ}|\E(Y_tY_{t-h})|\longrightarrow0,
\end{align}
where the last convergence holds because $R_T\to\infty$ by the hypothesis and hence $\eta_{R_T}\to0$. Markov's inequality therefore gives
\begin{align}
    \E^*|m_T^*(\ba)-m_{T,R_T}^*(\ba)|^2=o_p(1). \label{eqn:l2_approximation}
\end{align}
Thus the contribution outside the expanding interval
$[-R_Tb_T,R_Tb_T]$ is negligible in conditional $L^2$ while the truncated
multiplier sequence remains sufficiently short-memory for the blocking argument of \citet{Shao2010}. Below, we establish CLT for $m_{T,R_T}^*(\ba)$ (conditional on data), which together with \eqref{eqn:l2_approximation} implies CLT for $m_T^*(\ba)$.

\noindent\textbf{Step 2: Conditional CLT for $m_{T,R_T}^*(\ba)$.}
Because both multiplier statistics have conditional mean zero, the triangle
inequality in conditional $L^2$ and the bound established in Step 1 imply
\begin{align}
 \left|\{\Var^*\{m_{T,R_T}^*(\ba)\}\}^{1/2}
       -\{\Var^*(m_T^*(\ba))\}^{1/2}\right|
 &\leq
 \left\{\E^*|m_{T,R_T}^*(\ba)-m_T^*(\ba)|^2\right\}^{1/2}\\
 &=o_p(1).
\end{align}
Moreover,
\begin{align}
    \Var^*\{m_T^*(\ba)\}=
\ba'\bOmega_{p,T}\ba
\stackrel{p}{\longrightarrow}
\ba'\bOmega_0\ba
\end{align}
by Assumption \ref{ass:consistency_infHAC}.  Since $\ba'\bOmega_0 \ba>0$ for
$\ba\neq \bzero$, it follows that
\begin{align}
    \Var^*\{m_{T,R_T}^*(\ba)\}\stackrel{p}{\longrightarrow}\ba'\bOmega_0 \ba.
    \label{eq:compact_variance_consistency}
\end{align}

Let
\begin{align}
    m_{T,R_T}=\lceil2R_Tb_T\rceil,
    \qquad
    L_T=\left\lfloor T^{1/2}m_{T,R_T}^{-1/\delta}\right\rfloor,
    \qquad
    k_T=\left\lfloor\frac{T}{L_T+m_{T,R_T}}\right\rfloor.
\end{align}
By \eqref{eq:diagonal_RT}, $m_{T,R_T}/T^{\delta/(2+2\delta)}\to0$.  Hence
$m_{T,R_T}=o(L_T)$, $L_T\to\infty$, and $k_T\to\infty$.
Following the large-block--small-block argument in the proof of Theorem 3.1 of \citet{Shao2010}, define
\begin{align}
    \mathcal L_{T,r}
    &=\left\{t\in\bbN:(r-1)(L_T+m_{T,R_T})+1\leq t
      \leq r(L_T+m_{T,R_T})-m_{T,R_T}\right\},
      \qquad 1\leq r\leq k_T,\\
    \mathcal S_{T,r}
    &=\left\{t\in\bbN:r(L_T+m_{T,R_T})-m_{T,R_T}+1\leq t
      \leq r(L_T+m_{T,R_T})\right\},
      \qquad 1\leq r\leq k_T-1,\\
    \mathcal S_{T,k_T}
    &=\left\{t\in\bbN:k_T(L_T+m_{T,R_T})-m_{T,R_T}+1\leq t\leq T\right\}.
\end{align}
Thus, $\mathcal L_{T,r}$ is a large block of length $L_T$, $\mathcal S_{T,r}$,
$r<k_T$, is a small block of length $m_{T,R_T}$, and $\mathcal S_{T,k_T}$ contains
the last small block and the terminal remainder. Let
\begin{align}
    U_{T,r}=\sum_{t\in\mathcal L_{T,r}}\xi_{T,t}^{*(R_T)}Y_t,
    \qquad
    V_{T,r}=\sum_{t\in\mathcal S_{T,r}}\xi_{T,t}^{*(R_T)}Y_t,
    \qquad r=1,\ldots,k_T.
\end{align}
Conditional on the data, the large-block sums
$U_{T,1},\ldots,U_{T,k_T}$ are independent.  Absolute summability of
$\{\E(Y_tY_{t-h})\}$ and $|\E^*(\xi_{T,t}^{*(R_T)}\xi_{T,s}^{*(R_T)})|\leq1$ imply
\begin{align}
    \E\E^*\left|\frac1{\sqrt{T}}\sum_rV_{T,r}\right|^2
    \leq C\frac{k_Tm_{T,R_T}+L_T}{T}=o(1),
    \label{eq:small_blocks_negligible}
\end{align}
so the small blocks are negligible in conditional $L^2$ in probability.

For each large block, let
\begin{align}
    \ell_{T,r}&=(r-1)(L_T+m_{T,R_T})+1,\\
    \mathcal L_{T,r,g}
    &=\left\{\ell_{T,r}+g-1+j(m_{T,R_T}+1):
      j=0,1,\ldots,
      \left\lfloor\frac{L_T-g}{m_{T,R_T}+1}\right\rfloor\right\},
      \qquad g=1,\ldots,m_{T,R_T}+1.
\end{align}
Then $\{\mathcal L_{T,r,g}\}_g$ form a partition of $\mathcal L_{T,r}$. Any two distinct indices in the same $\mathcal L_{T,r,g}$ are separated by at least $m_{T,R_T}+1$, so the variables $\{\xi_{T,t}^{*(R_T)}Y_t:t\in\mathcal L_{T,r,g}\}$ are independent conditional on the data. Thus each large block is partitioned into at most $m_{T,R_T}+1$ independent residue classes. Applying Rosenthal's inequality conditionally on the data gives, as in equation (A.3) of \citet{Shao2010},
\begin{align}
    \|U_{T,r}\|_{2+\delta,*}
    \leq C m_{T,R_T}^{1/2}
       \left(\sum_{t\in\mathcal L_{T,r}}Y_t^2\right)^{1/2},
    \label{eq:block_rosenthal}
\end{align}
where $C$ is a finite constant independent of $T$. Consequently,
\begin{align}
 &T^{-1-\delta/2}\sum_{r=1}^{k_T}\E^*|U_{T,r}|^{2+\delta}\\
 &\quad\leq
 C T^{-1-\delta/2}m_{T,R_T}^{1+\delta/2}L_T^{\delta/2}
   \sum_{t=1}^T|Y_t|^{2+\delta}\\
 &\quad=O_p\left(
   \frac{m_{T,R_T}^{(1+\delta)/2}}{T^{\delta/4}}
   \right)=o_p(1),
    \label{eq:conditional_lyapunov}
\end{align}
where $T^{-1}\sum_{t=1}^T|Y_t|^{2+\delta}=O_p(1)$ by Assumption \ref{ass:moment_stronger}, and the last equality is due to
$m_{T,R_T}/T^{\delta/(2+2\delta)}\to0$. Let $B_T^*=T^{-1/2}\sum_{r=1}^{k_T}U_{T,r}$ and let $S_T^*$ collect the small blocks and the terminal remainder. Equation \eqref{eq:small_blocks_negligible} gives $\E^*|S_T^*|^2=o_p(1)$. The triangle inequality and \eqref{eq:compact_variance_consistency} therefore imply
$\Var^*(B_T^*)\stackrel{p}{\to}\ba'\bOmega_0 \ba$. Using the Lindeberg--Feller theorem together with \eqref{eq:conditional_lyapunov}, we deduce
\begin{align}
    m_{T,R_T}^*(\ba)\CDs N(0,\ba'\bOmega_0 \ba)
    \quad\text{in probability}.
    \label{eq:compact_conditional_clt}
\end{align}
Together with \eqref{eqn:l2_approximation}, this yields
\begin{align}
    m_T^*(\ba)\CDs N(0,\ba'\bOmega_0 \ba)
    \quad\text{in probability}.
\end{align}
The Cram\'er--Wold device finally gives
\begin{align}
    \bm_T^*\CDs N(\bzero,\bOmega_0)
    \quad\text{in probability}.
    \label{eq:main_conditional_clt}
\end{align}

\noindent\textbf{Step 3: Estimated score and recentering remainders.}
Recall the matrix $\bD_t(\overline{\btheta}_t)$ defined in the proof of Theorem \ref{thm:matched_hac_consistency}. Then the mean value expansion of $\what{\bg}_t = \bg_t(\what{\btheta})$ gives
\begin{align}
    \what{\bg}_t
    =\bg_t(\btheta_0)+\bD_t(\overline{\btheta}_t)
    (\what{\btheta}-\btheta_0).
    \label{eq:componentwise_mve_score}
\end{align}
It follows that
\begin{align}
    \bs_T^*=\bm_T^*+\br_{1T}^*-\br_{2T}^*,
    \label{eq:sstar_decomposition}
\end{align}
where
\begin{align}
    \br_{1T}^*
    &=\left(\frac1T\sum_{t=1}^T\xi_t^*\bD_t(\overline{\btheta}_t)\right)
    \sqrt{T}(\what{\btheta}-\btheta_0),\\
    \br_{2T}^*
    &=\overline{\what{\bg}}\frac1{\sqrt{T}}\sum_{t=1}^T\xi_t^*.
\end{align}
By \eqref{eq:kernel_rowsum_scaled}, the inequality $|\tr(AB')|\leq(\|A\|^2+\|B\|^2)/2$, and \eqref{eq:intermediate_jacobian_bound}, we obtain
\begin{align}
    \E^*\|T^{-1}\sum_t\xi_t^*\bD_t(\overline{\btheta}_t)\|^2
    &=\frac1{T^2}\sum_{t,s=1}^T K_p(\rho_{T,t-s})
    \tr(\bD_t(\overline{\btheta}_t)\bD_t(\overline{\btheta}_s)')\\
    &\leq \frac{C b_T}{T}
    \left\{\frac1T\sum_{t=1}^T
    \sup_{\btheta\in\Theta}\|\bD_t(\btheta)\|^2\right\}
    =O_p\left(\frac{b_T}{T}\right)=o_p(1).
    \label{eq:derivative_double_sum}
\end{align}
Since $\sqrt{T}(\what{\btheta}-\btheta_0)=O_p(1)$,
\begin{align}
    \E^*\|\br_{1T}^*\|^2=o_p(1).
    \label{eq:A_negligible}
\end{align}

For the recentering term, \eqref{eq:gbar_rate} and
\eqref{eq:kernel_rowsum_scaled} give
\begin{align}
    \E^*\|\br_{2T}^*\|^2
    &=\|\overline{\what{\bg}}\|^2\frac1T
    \sum_{t,s=1}^TK_p(\rho_{T,t-s})
    \leq Cb_T\|\overline{\what{\bg}}\|^2
    =O_p\left(\frac{b_T}{T}\right)=o_p(1).
    \label{eq:center_remainder_negligible}
\end{align}
An application of Jensen's inequality implies that both $\E^*\|\br_{1T}^*\|$ and $\E^*\|\br_{2T}^*\|$ are $o_p(1)$.

Slutsky's theorem, \eqref{eq:main_conditional_clt},
\eqref{eq:A_negligible}, and \eqref{eq:center_remainder_negligible} now imply
$\bs_T^*\CDs N(\bzero,\bOmega_0)$ in probability.  Since $\bOmega_0$ is positive
definite, the P\'olya theorem applies conditionally, so
\begin{align}
    \sup_{\bx\in\bbR^d}\left|\Pro^*(\bs_T^*\leq \bx)-\Pro(\by\leq \bx)\right|
    \stackrel{p}{\longrightarrow}0.
\end{align}
Combining this display with \eqref{eq:polya_original} proves the theorem.
\end{proof}

\noindent\begin{proof}[Proof of Corollary \ref{cor:theta_hat_boot_consistent}]
By Assumption~\ref{ass:moment_weakdepend}(i) and the definition of $\bV_{\theta}$,
\begin{align}
    \sqrt{T}(\what{\btheta}-\btheta_0)
    =
    \bB_0\frac{1}{\sqrt{T}}\sum_{t=1}^T\bg_t(\btheta_0)+o_p(1)
    \CD
    N(\bzero,\bV_{\theta}).
    \label{eq:theta_original_limit_corproof}
\end{align}
By Theorem~\ref{thm:score_boot_consistent},
\begin{align}
    \frac{1}{\sqrt{T}}\sum_{t=1}^T\what{\bg}_t^*
    \CDs
    N(\bzero,\bOmega_0)
\end{align}
in probability.
Since $\what{\bB}\CP\bB_0$, Slutsky's theorem applied to \eqref{eq:bootstrap_theta_update} gives
\begin{align}
    \sqrt{T}(\what{\btheta}^*-\what{\btheta})
    =
    \what{\bB}
    \frac{1}{\sqrt{T}}\sum_{t=1}^T\what{\bg}_t^*
    \CDs
    N(\bzero,\bV_{\theta})
    \label{eq:theta_bootstrap_limit_corproof}
\end{align}
in probability.
Because $\operatorname{rank}(\bB_0)=k$ and $\bOmega_0$ is positive definite, $\bV_{\theta}$ is positive definite.
The P\'olya theorem applied to \eqref{eq:theta_original_limit_corproof} and conditionally to \eqref{eq:theta_bootstrap_limit_corproof}, followed by the triangle inequality, yields the result.
\end{proof}

\noindent\begin{proof}[Proof of Corollary \ref{cor:wald}]
Under $H_0$, let $\bSigma_A=\bA\bV_{\theta}\bA'$ and define the infeasible population-studentized vector
\begin{align}
    \bz_{W,T}^{0}
    &=
    \bSigma_A^{-1/2}
    \sqrt{T}\bA(\what{\btheta}-\btheta_0).
    \label{eq:wald_infeasible_original_proof}
\end{align}
Corollary~\ref{cor:theta_hat_boot_consistent} and the positive definiteness of $\bSigma_A$ give $\bz_{W,T}^{0}\CD N(\bzero,\bI_\ell)$.
Since $\bA\what{\bV}_{\theta}\bA'\CP\bSigma_A$ and $\sqrt{T}\bA(\what{\btheta}-\btheta_0)=O_p(1)$, the continuous mapping theorem gives
\begin{align}
    \what{\bz}_{W}-\bz_{W,T}^{0}
    &=
    \left\{
    (\bA\what{\bV}_{\theta}\bA')^{-1/2}
    -
    \bSigma_A^{-1/2}
    \right\}
    \sqrt{T}\bA(\what{\btheta}-\btheta_0)
    =
    o_p(1).
    \label{eq:wald_feasible_infeasible_original_proof}
\end{align}
Hence the feasible statistic satisfies
\begin{align}
    \what{\bz}_{W}
    \CD
    N(\bzero,\bI_\ell).
    \label{eq:wald_standardized_original_proof}
\end{align}

For the bootstrap, define the infeasible counterpart
\begin{align}
    \bz_{W,T}^{*,0}
    &=
    \bSigma_A^{-1/2}
    \sqrt{T}\bA(\what{\btheta}^*-\what{\btheta}).
    \label{eq:wald_infeasible_bootstrap_proof}
\end{align}
Corollary~\ref{cor:theta_hat_boot_consistent} gives $\bz_{W,T}^{*,0}\CDs N(\bzero,\bI_\ell)$ in probability.
Under fixed studentization, $\bA\what{\bV}_{\theta}\bA'\CP\bSigma_A$, while under recomputed studentization $\bA\what{\bV}_{\theta}^*\bA'\CPs\bSigma_A$ in probability.
Since $\sqrt{T}\bA(\what{\btheta}^*-\what{\btheta})=O_{p^*}(1)$ in probability, Slutsky's theorem gives $\what{\bz}_{W}^*-\bz_{W,T}^{*,0}=o_{p^*}(1)$ in probability for either studentization scheme.
Therefore the feasible bootstrap statistic satisfies
\begin{align}
    \what{\bz}_{W}^*
    \CDs
    N(\bzero,\bI_\ell)
    \quad\text{in probability}.
    \label{eq:wald_standardized_bootstrap_proof}
\end{align}
The P\'olya theorem applied to \eqref{eq:wald_standardized_original_proof} and conditionally to \eqref{eq:wald_standardized_bootstrap_proof}, followed by the triangle inequality, proves \eqref{eq:wald_bootstrap_gaussian_validity}.
Finally, the continuous mapping theorem gives $W_T\CD\chi_\ell^2$ and $W_T^*\CDs\chi_\ell^2$ in probability.
Because the $\chi_\ell^2$ distribution is continuous, the P\'olya theorem and the triangle inequality then give \eqref{eq:wald_bootstrap_validity}.
\end{proof}

\noindent\begin{proof}[Proof of Corollary \ref{cor:lm}]
Set
\begin{align}
    \bM_0
    &=
    \bG_{2\cdot1}'\bOmega_0^{-1},
    \qquad
    \wtilde{\bM}_p
    =
    \wtilde{\bG}_{2\cdot1}'\wtilde{\bOmega}_p^{-1},
\end{align}
and define
\begin{align}
    \bJ_0
    =
    \bM_0\bOmega_0\bM_0'
    =
    \bG_{2\cdot1}'\bOmega_0^{-1}\bG_{2\cdot1}.
    \label{eq:lm_projection_matrices_proof}
\end{align}
By construction,
\begin{align}
    \bM_0\bG_1
    =
    \bzero.
    \label{eq:lm_population_orthogonality_proof}
\end{align}
Premultiplying \eqref{eq:lm_restricted_moment_linearization} by $\wtilde{\bM}_p$ gives
\begin{align}
    \sqrt{T}\,\wtilde{\bM}_p\overline{\wtilde{\bg}}
    =
    \wtilde{\bM}_p\frac{1}{\sqrt{T}}\sum_{t=1}^T\bg_t(\btheta_0)
    +
    \wtilde{\bM}_p\bG_1\sqrt{T}(\wtilde{\boldsymbol{\eta}}_1-\boldsymbol{\eta}_{1,0})
    +
    o_p(1).
    \label{eq:lm_observed_projected_moment_expansion_proof}
\end{align}
The assumed consistency of $\wtilde{\bG}_1$, $\wtilde{\bG}_2$, and $\wtilde{\bOmega}_p$ implies that the first term equals $\bM_0T^{-1/2}\sum_t\bg_t(\btheta_0)+o_p(1)$, and the second is $o_p(1)$ by \eqref{eq:lm_population_orthogonality_proof} and \eqref{eq:lm_restricted_moment_linearization}.
Moreover, $\wtilde{\bG}_{2\cdot1}'\wtilde{\bOmega}_p^{-1}\wtilde{\bG}_{2\cdot1}\CP\bJ_0$.
Therefore the score CLT and the continuous mapping theorem yield
\begin{align}
    \wtilde{\bz}_{LM}
    \CD
    N(\bzero,\bI_\ell).
    \label{eq:lm_standardized_original_proof}
\end{align}

For the bootstrap, define the infeasible counterpart
\begin{align}
    \bz_{LM,T}^{*,0}
    &=
    -\bJ_0^{-1/2}
    \bM_0
    \frac{1}{\sqrt{T}}\sum_{t=1}^T\xi_t^*\bg_t(\btheta_0).
    \label{eq:lm_infeasible_bootstrap_proof}
\end{align}
The conditional CLT in \eqref{eq:main_conditional_clt} gives $\bz_{LM,T}^{*,0}\CDs N(\bzero,\bI_\ell)$ in probability.
For the numerator of \eqref{eq:lm_bootstrap_generic}, \eqref{eq:restricted_moment_bootstrap} gives
\begin{align}
    \sqrt{T}\,\wtilde{\bM}_p\overline{\wtilde{\bg}^*}
    &=
    \bM_0\frac{1}{\sqrt{T}}\sum_{t=1}^T\xi_t^*\bg_t(\btheta_0)
    +
    \boldsymbol{\rho}_{1T}^*
    +
    \boldsymbol{\rho}_{2T}^*
    -
    \boldsymbol{\rho}_{3T}^*,
    \label{eq:lm_boot_projected_moment_decomposition_proof}
\end{align}
where
\begin{align}
    \boldsymbol{\rho}_{1T}^*
    &=
    (\wtilde{\bM}_p-\bM_0)\frac{1}{\sqrt{T}}\sum_{t=1}^T\xi_t^*\bg_t(\btheta_0),
    \\
    \boldsymbol{\rho}_{2T}^*
    &=
    \wtilde{\bM}_p\frac{1}{\sqrt{T}}\sum_{t=1}^T\xi_t^*\{\bg_t(\wtilde{\btheta})-\bg_t(\btheta_0)\},
    \\
    \boldsymbol{\rho}_{3T}^*
    &=
    \wtilde{\bM}_p\overline{\wtilde{\bg}}\frac{1}{\sqrt{T}}\sum_{t=1}^T\xi_t^*.
    \label{eq:lm_boot_remainders_proof}
\end{align}
The conditional CLT in \eqref{eq:main_conditional_clt} and $\wtilde{\bM}_p\CP\bM_0$ give $\boldsymbol{\rho}_{1T}^*=o_{p^*}(1)$ in probability.
Root-$T$ consistency of $\wtilde{\btheta}$ and the same derivative argument as in \eqref{eq:derivative_double_sum} give $\boldsymbol{\rho}_{2T}^*=o_{p^*}(1)$ in probability.
For $\boldsymbol{\rho}_{3T}^*$, the fact that $\sqrt{T}\,\overline{\wtilde{\bg}}=O_p(1)$ by  \eqref{eq:lm_restricted_moment_linearization}, while $\Var^*(T^{-1/2}\sum_t\xi_t^*)=O(b_T)$ by \eqref{eq:kernel_rowsum_scaled} implies $\boldsymbol{\rho}_{3T}^*=o_{p^*}(1)$ because $b_T/T\to0$.
For fixed studentization, the bootstrap covariance in \eqref{eq:lm_bootstrap_generic} converges in probability to $\bJ_0$.
For recomputed studentization, $\wtilde{\bOmega}_p^*\CPs\bOmega_0$ and $\wtilde{\bM}_p\CP\bM_0$ imply $\wtilde{\bM}_p\wtilde{\bOmega}_p^*\wtilde{\bM}_p'\CPs\bJ_0$ in probability.
Slutsky's theorem therefore gives $\wtilde{\bz}_{LM}^*-\bz_{LM,T}^{*,0}=o_{p^*}(1)$ in probability for both studentization schemes.
Hence the feasible bootstrap statistic satisfies
\begin{align}
    \wtilde{\bz}_{LM}^*
    \CDs
    N(\bzero,\bI_\ell)
    \quad\text{in probability}.
    \label{eq:lm_standardized_bootstrap_proof}
\end{align}
The P\'olya theorem applied to \eqref{eq:lm_standardized_original_proof} and conditionally to \eqref{eq:lm_standardized_bootstrap_proof}, followed by the triangle inequality, proves \eqref{eq:lm_bootstrap_gaussian_validity}.
Finally, the results for $LM_T$ and $LM_T^*$ follow from the continuous mapping argument and the continuity of the $\chi_\ell^2$ distribution.
\end{proof}

\subsection{Proof of the regression validity result}
\label{app:reg_primitive}

We first record two technical lemmas used in the proof of Proposition~\ref{prop:reg_validity}.

\begin{lem}[HAC matrix replacement]
\label{lem:reg_hac_replacement}
For arbitrary arrays $\{\ba_t\}_{t=1}^T$ and $\{\bd_t\}_{t=1}^T$ in $\bbR^m$, where $m$ is fixed,
\begin{align}
    &\left\|
      {1\over T}\sum_{t=1}^T\sum_{s=1}^T
      K_p\!\left[\exp\left\{-\left({t-s\over \lambda_T}\right)^2\right\}\right]
      (\ba_t\ba_s'-\bd_t\bd_s')
    \right\|
    \label{eqn:diff_quadratic} \\
    &\qquad\leq
    C\lambda_T
      \left\{T^{-1}\sum_{t=1}^T\|\bd_t\|^2\right\}^{1/2}
      \left\{T^{-1}\sum_{t=1}^T\|\ba_t-\bd_t\|^2\right\}^{1/2}
    +C\lambda_T T^{-1}\sum_{t=1}^T\|\ba_t-\bd_t\|^2,
\end{align}
where $C<\infty$ does not depend on the arrays or $T$.
\end{lem}

\begin{proof}
Let $\be_t=\ba_t-\bd_t$ and write $w_{ts}=K_p[\exp\{-((t-s)/\lambda_T)^2\}]$.
Since $\ba_t\ba_s'-\bd_t\bd_s'=\bd_t\be_s'+\be_t\bd_s'+\be_t\be_s'$, the triangle inequality separates the difference in \eqref{eqn:diff_quadratic} into three terms. We bound the three terms in turn.
The Gaussian tail bound \eqref{bound:kp_exp_uni} gives
\[
    \max_t\sum_{s=1}^T|w_{ts}|
    =
    \max_s\sum_{t=1}^T|w_{ts}|
    \leq C\lambda_T.
\]
Since $\|\bu\bv'\|\leq\|\bu\|\,\|\bv\|$, the Cauchy--Schwarz inequality bounds the first term as
\begin{align*}
    \sum_{t,s=1}^T|w_{ts}|\,\|\bd_t\|\,\|\be_s\|
    &\leq
    \left\{\sum_{t,s=1}^T|w_{ts}|\|\bd_t\|^2\right\}^{1/2}
    \left\{\sum_{t,s=1}^T|w_{ts}|\|\be_s\|^2\right\}^{1/2} \\
    &\leq
    C\lambda_T
    \left(\sum_{t=1}^T\|\bd_t\|^2\right)^{1/2}
    \left(\sum_{t=1}^T\|\be_t\|^2\right)^{1/2}.
\end{align*}
The same bound applies to the second cross term, while
\[
    \sum_{t,s=1}^T|w_{ts}|\,\|\be_t\|\,\|\be_s\|
    \leq
    C\lambda_T\sum_{t=1}^T\|\be_t\|^2.
\]
Dividing by $T$ proves the result.
\end{proof}

\begin{lem}[Multiplier-product covariance bound]
\label{lem:multiplier_product_covariance}
For the dependent two-point multiplier process with latent correlations $\rho_{T,h}$ given by \eqref{eq:latent_gaussian_corr} and $t,s,u,v\in\{1,\ldots,T\}$, there exist constants $C,c_0>0$, independent of $T$ and the indices, such that
\begin{align}
    \left|\operatorname{Cov}^*(\xi_t^*\xi_s^*,\xi_u^*\xi_v^*)\right|
    \leq
    C\exp\left[-{c_0\over b_T^2}\min\{|t-u|,|t-v|,|s-u|,|s-v|\}^2\right].
    \label{eq:multiplier_product_covariance_decay}
\end{align}
Consequently, for all sufficiently large $T$, uniformly over $h,k\in\bbZ$ and $t$ satisfying $1\leq t,t-h\leq T$,
\begin{align}
    \sum_{\substack{1\leq u\leq T\\1\leq u-k\leq T}}
    \left|\operatorname{Cov}^*(\xi_t^*\xi_{t-h}^*,\xi_u^*\xi_{u-k}^*)\right|
    \leq Cb_T.
    \label{eq:multiplier_product_covariance_rowsum}
\end{align}
\end{lem}

\begin{proof}
Let $d=\min\{|t-u|,|t-v|,|s-u|,|s-v|\}$.
Boundedness of the two-point multipliers gives \eqref{eq:multiplier_product_covariance_decay} when $d=0$.
For $d>0$, set $M=d/(3b_T)$ and recall the truncated multiplier process $\xi_{T,j}^{*(M)}$ in \eqref{eq:truncated_multiplier}.
Its value at index $j$ depends only on the Brownian increments over $[j-Mb_T,j+Mb_T]$.
Since every index in $\{t,s\}$ is at least distance $d$ from every index in $\{u,v\}$, the corresponding unions of intervals are disjoint.
Hence $\xi_{T,t}^{*(M)}\xi_{T,s}^{*(M)}$ and $\xi_{T,u}^{*(M)}\xi_{T,v}^{*(M)}$ are independent under the bootstrap law.
Adding and subtracting the truncated products, and using boundedness and the Cauchy--Schwarz inequality, we get
\begin{align*}
    \left|\operatorname{Cov}^*(\xi_t^*\xi_s^*,\xi_u^*\xi_v^*)\right|
    &\leq
    C\sum_{j\in\{t,s,u,v\}}\left\{\E^*|\xi_j^*-\xi_{T,j}^{*(M)}|^2\right\}^{1/2} \\
    &\leq
    C\exp(-c_0M^2),
\end{align*}
where the last inequality follows from \eqref{eq:multiplier_trunc_error}, after adjusting $c_0$ if necessary.
Substituting $M=d/(3b_T)$ proves \eqref{eq:multiplier_product_covariance_decay}.
For the row-sum bound, the minimum distance between the pairs $\{t,t-h\}$ and $\{u,u-k\}$ is
\[
    \min\{|u-t|,|u-(t+k)|,|u-(t-h)|,|u-(t-h+k)|\}.
\]
Therefore, after another adjustment of constants,
\begin{align*}
    &\exp\left[-{c_0\over b_T^2}\min\{|u-t|,|u-(t+k)|,|u-(t-h)|,|u-(t-h+k)|\}^2\right] \\
    &\qquad\leq
    \exp\left(-c_0{(u-t)^2\over b_T^2}\right)
    +\exp\left(-c_0{(u-t-k)^2\over b_T^2}\right)\\  &\qquad\quad
    +\exp\left(-c_0{(u-t+h)^2\over b_T^2}\right)
    +\exp\left(-c_0{(u-t+h-k)^2\over b_T^2}\right).
\end{align*}
Each term has discrete row sum of order $b_T$, since
\[
    \sum_{j\in\bbZ}\exp(-c_0j^2/b_T^2)
    \leq
    1+\int_{-\infty}^{\infty}\exp(-c_0x^2/b_T^2)\,dx
    \leq Cb_T.
\]
This proves \eqref{eq:multiplier_product_covariance_rowsum}.
\end{proof}

\noindent\begin{proof}[Proof of Proposition~\ref{prop:reg_validity}]

\noindent\textbf{Part (i): Original-sample limits and matched-HAC consistency.}

Write $\what{\bQ}_x=T^{-1}\sum_{t=1}^T\bx_t\bx_t'$.
The mixing law of large numbers (LLN) and Assumption~\ref{ass:reg_additional} give $\what{\bQ}_x\CP\bQ_x$.
The strongly mixing CLT gives
\begin{align}
    {1\over\sqrt T}\sum_{t=1}^T\bg_t
    \CD
    N(\bzero,\bOmega_0).
    \label{eq:reg_g_clt}
\end{align}
The unrestricted OLS estimator satisfies
\begin{align}
    \what{\btheta}-\btheta_0
    =
    \what{\bQ}_x^{-1}T^{-1}\sum_{t=1}^T\bg_t
    =
    O_p(T^{-1/2}),
    \label{eq:reg_theta_rate}
\end{align}
which proves \eqref{eq:reg_theta_primitive_expansion}.
Let $\boldsymbol{\eta}_{1,0}$ denote the nuisance component of $\btheta_0=\bC_1\boldsymbol{\eta}_{1,0}+\bC_2\ba_0$.
The first-order condition of the restricted OLS estimator gives
\begin{align}
    \wtilde{\boldsymbol{\eta}}_1-\boldsymbol{\eta}_{1,0}
    =
    (\bC_1'\what{\bQ}_x\bC_1)^{-1}
    \bC_1'T^{-1}\sum_{t=1}^T\bg_t
    =
    O_p(T^{-1/2}).
    \label{eq:reg_eta1_rate}
\end{align}

We next establish that replacing the unobserved moment with its estimator has asymptotically no effect.
Using $\what{\bg}_t-\bg_t = -\bx_t\bx_t'(\what{\btheta}-\btheta_0)$, $\wtilde{\bg}_t-\bg_t = -\bx_t\bx_t'\bC_1(\wtilde{\boldsymbol{\eta}}_1-\boldsymbol{\eta}_{1,0})$, \eqref{eq:reg_theta_rate}, \eqref{eq:reg_eta1_rate}, the mixing LLN, and Assumption~\ref{ass:reg_additional}(iii), we obtain
\begin{align}
    {1\over T}\sum_{t=1}^T\|\what{\bg}_t-\bg_t\|^2
    &=
    O_p(T^{-1}),
    \\
    {1\over T}\sum_{t=1}^T\|\wtilde{\bg}_t-\bg_t\|^2
    &=
    O_p(T^{-1}).
    \label{eq:reg_uncentered_score_l2_replacement}
\end{align}
The restricted sample mean satisfies
\[
    \overline{\wtilde{\bg}}
    =
    T^{-1}\sum_{t=1}^T\bg_t
    -
    \what{\bQ}_x\bC_1(\wtilde{\boldsymbol{\eta}}_1-\boldsymbol{\eta}_{1,0})
    =
    O_p(T^{-1/2}),
\]
and hence
\begin{align}
    {1\over T}\sum_{t=1}^T
    \|\wtilde{\bg}_t-\overline{\wtilde{\bg}}-\bg_t\|^2
    =
    O_p(T^{-1}).
    \label{eq:reg_score_l2_replacement}
\end{align}

Consider the infeasible matched-HAC matrix
\begin{align}
    \bOmega_{p,T}
    =
    {1\over T}\sum_{t=1}^T\sum_{s=1}^T
    K_p(\rho_{T,t-s})
    \bg_t\bg_s'.
\end{align}
Let $\bGamma_0(h)=\E(\bg_t\bg_{t-h}')$ for $h\in\bbZ$.
Davydov's inequality and Assumption~\ref{ass:reg_additional}(ii) imply absolute summability of every coordinate sequence of $\{\bGamma_0(h)\}_{h\in\bbZ}$.
For each fixed $h$, $K_p(\rho_{T,h})\to1$, and boundedness of the kernel and dominated convergence give
\begin{align}
    \E\bOmega_{p,T}
    =
    \sum_{|h|<T}
    \left(1-{|h|\over T}\right)
    K_p(\rho_{T,h})
    \bGamma_0(h)
    \longrightarrow
    \bOmega_0.
    \label{eq:reg_hac_bias}
\end{align}
For each pair of coordinates of $\bOmega_{p,T}$, the covariance of the corresponding score products decomposes into one fourth-order cumulant and two products of second-order cross-covariances.
Assumption~\ref{ass:reg_additional}(iv), absolute summability of the second-order cross-covariances, and \eqref{eq:kernel_rowsum} therefore imply
\begin{align}
    \E\|\bOmega_{p,T}-\E\bOmega_{p,T}\|^2
    \leq
    C{1+b_T\over T}
    =
    o(1).
    \label{eq:reg_hac_variance}
\end{align}
Thus $\bOmega_{p,T}\CP\bOmega_0$.
An application of Lemma~\ref{lem:reg_hac_replacement} with $\lambda_T=b_T$, combined with \eqref{eq:reg_uncentered_score_l2_replacement}--\eqref{eq:reg_score_l2_replacement} and Assumption~\ref{ass:reg_additional}(v), gives
\[
    \what{\bOmega}_p-\bOmega_{p,T}=o_p(1),
    \qquad
    \wtilde{\bOmega}_p-\bOmega_{p,T}=o_p(1).
\]
Therefore $\what{\bOmega}_p\CP\bOmega_0$ and $\wtilde{\bOmega}_p\CP\bOmega_0$.
Since $\what{\bQ}_x\CP\bQ_x$, continuous mapping gives $\what{\bV}_{\theta}\CP\bV_{\theta}$.
Moreover, $\wtilde{\bG}_j\CP\bG_j$ for $j\in\{1,2\}$ and $\wtilde{\bOmega}_p\CP\bOmega_0$ imply $\wtilde{\bG}_{2\cdot1}'\wtilde{\bOmega}_p^{-1}\wtilde{\bG}_{2\cdot1}\CP\bG_{2\cdot1}'\bOmega_0^{-1}\bG_{2\cdot1}$. This proves \eqref{eq:reg_observed_hac_consistency}.

The unrestricted Gaussian limit in \eqref{eq:reg_observed_gaussian_limits} follows directly from \eqref{eq:reg_theta_primitive_expansion}, consistency of $\what{\bV}_{\theta}$, and Slutsky's theorem.
For the restricted studentized vector, we have
\begin{align}
    \sqrt T\,\overline{\wtilde{\bg}}
    =
    {1\over\sqrt T}\sum_{t=1}^T\bg_t
    +
    \wtilde{\bG}_1
    \sqrt T(\wtilde{\boldsymbol{\eta}}_1-\boldsymbol{\eta}_{1,0}).
    \label{eq:reg_restricted_moment_expansion}
\end{align}
Since $\wtilde{\bG}_{2\cdot1}'\wtilde{\bOmega}_p^{-1}\wtilde{\bG}_1=\bzero$ and $\wtilde{\bG}_{2\cdot1}'\wtilde{\bOmega}_p^{-1}\CP\bG_{2\cdot1}'\bOmega_0^{-1}$,
\begin{align}
    \sqrt T\left(\wtilde{\bG}_{2\cdot1}'\wtilde{\bOmega}_p^{-1}\right)\overline{\wtilde{\bg}}
    \CD
    N(\bzero,\bG_{2\cdot1}'\bOmega_0^{-1}\bG_{2\cdot1}).
    \label{eq:reg_lm_observed_projected_moment_clt}
\end{align}
Slutsky's theorem now gives $\wtilde{\bz}_{LM}\CD N(\bzero,\bI_\ell)$, which completes \eqref{eq:reg_observed_gaussian_limits}.
The chi-square limits $W_T\CD\chi_\ell^2$ and $LM_T\CD\chi_\ell^2$ follow by the continuous mapping theorem.

\medskip
\noindent\textbf{Part (ii): Residual bootstrap with fixed studentizers.}

Suppose additionally that Assumption~\ref{ass:reg_bootstrap}(i) holds.
The same truncation and conditional-CLT argument used in the proof of Theorem~\ref{thm:score_boot_consistent} applies to the $k$-dimensional moment process $\{\bg_t\}$, giving
\begin{align}
    {1\over\sqrt T}\sum_{t=1}^T\xi_t^*\bg_t
    \CDs
    N(\bzero,\bOmega_0)
    \label{eq:reg_ideal_boot_score_clt}
\end{align}
in probability.
By \eqref{eq:kernel_rowsum}, \eqref{eq:reg_uncentered_score_l2_replacement}, and the Cauchy--Schwarz inequality, we obtain
\begin{align}
    \E^*\left\|
    {1\over\sqrt T}\sum_{t=1}^T
    \xi_t^*(\what{\bg}_t-\bg_t)
    \right\|^2
    &\leq
    Cb_T
    {1\over T}\sum_{t=1}^T
    \|\what{\bg}_t-\bg_t\|^2
    =
    o_p(1),
\end{align}
and similarly
\[
    \E^*\left\|
    {1\over\sqrt T}\sum_{t=1}^T
    \xi_t^*(\wtilde{\bg}_t-\overline{\wtilde{\bg}}-\bg_t)
    \right\|^2
    =
    o_p(1).
\]
by \eqref{eq:reg_score_l2_replacement}. Consequently,
\begin{align}
    {1\over\sqrt T}\sum_{t=1}^T\xi_t^*\what{\bg}_t
    &\CDs
    N(\bzero,\bOmega_0),
    \\
    {1\over\sqrt T}\sum_{t=1}^T
    \xi_t^*(\wtilde{\bg}_t-\overline{\wtilde{\bg}})
    &\CDs
    N(\bzero,\bOmega_0).
    \label{eq:reg_direct_score_numerator_clt}
\end{align}
in probability.
Equation \eqref{eq:reg_resid_wald_numerator}, together with $\what{\bQ}_x\CP\bQ_x$, therefore gives
\begin{align}
    \sqrt T\bA(\what{\btheta}^*-\what{\btheta})
    \CDs
    N(\bzero,\bA\bV_{\theta}\bA')
    \label{eq:reg_residual_boot_numerator_wald_clt}
\end{align}
in probability.
For the residual LM statistic, we use \eqref{eq:reg_resid_lm_numerator}. Since $\|\wtilde{\bG}_{2\cdot1}'\wtilde{\bOmega}_p^{-1}\overline{\wtilde{\bg}}\|=O_p(T^{-1/2})$ and $\Var^*(T^{-1/2}\sum_{t=1}^T\xi_t^*)=O(b_T)$ by \eqref{eq:kernel_rowsum}, the second term in \eqref{eq:reg_resid_lm_numerator} is $o_{p^*}(1)$ in probability.
It follows that
\begin{align}
    \left(\wtilde{\bG}_{2\cdot1}'\wtilde{\bOmega}_p^{-1}\right)
    {1\over\sqrt T}\sum_{t=1}^T\wtilde{\bg}_{t,\mathrm{res}}^*
    \CDs
    N(\bzero,\bG_{2\cdot1}'\bOmega_0^{-1}\bG_{2\cdot1})
    \label{eq:reg_residual_boot_projected_moment_lm_clt}
\end{align}
in probability.
Part (i), \eqref{eq:reg_residual_boot_numerator_wald_clt}, and \eqref{eq:reg_residual_boot_projected_moment_lm_clt}, together with consistency of the fixed observed studentizers and Slutsky's theorem, prove the standardized Gaussian limits in \eqref{eq:reg_boot_fixed_clt}.

\medskip
\noindent\textbf{Part (iii): Recomputed bootstrap studentizers.}

Suppose additionally that Assumption~\ref{ass:reg_bootstrap}(ii)--(iii) holds.
First consider the ideal bootstrap HAC matrix
\begin{align}
    \bQ_T^*
    =
    {1\over T}\sum_{t=1}^T\sum_{s=1}^T
    K_p\!\left[\exp\left\{-\left({t-s\over b_T^*}\right)^2\right\}\right]
    \xi_t^*\xi_s^*\bg_t\bg_s'.
    \label{eq:reg_boot_hac_uncentered}
\end{align}
Conditional on the data,
\begin{align}
    \E^*\bQ_T^*
    =
    {1\over T}\sum_{t=1}^T\sum_{s=1}^T
    K_p\!\left[\exp\left\{-\left({t-s\over b_T^*}\right)^2\right\}\right]
    K_p(\rho_{T,t-s})
    \bg_t\bg_s'.
    \label{eq:reg_boot_hac_conditional_mean}
\end{align}
For each fixed lag, both kernel factors converge to one, and the Gaussian tail bound gives
\[
    \left|
    K_p\{\exp[-(h/b_T^*)^2]\}
    K_p(\rho_{T,h})
    \right|
    \leq
    C\exp\left[
    -c_1h^2\{(b_T^*)^{-2}+b_T^{-2}\}
    \right].
\]
Thus the product weight has absolute row sum $O(\min\{b_T^*,b_T\})$, and the same coordinatewise bias and cumulant-variance argument used in \eqref{eq:reg_hac_bias}--\eqref{eq:reg_hac_variance} gives
\begin{align}
    \E^*\bQ_T^*
    \CP
    \bOmega_0.
    \label{eq:reg_boot_hac_mean_limit}
\end{align}
Write $w_{T,h}=K_p\{\exp[-(h/b_T^*)^2]\}$.
The Gaussian tail bound implies
\[
    \sum_{h\in\bbZ}w_{T,h}^2
    \leq
    Cb_T^*,
    \qquad
    \sup_{q\in\bbZ}
    \sum_{h\in\bbZ}
    |w_{T,h}w_{T,h+q}|
    \leq
    Cb_T^*.
\]
For each coordinate of $\bQ_T^*$, Assumption~\ref{ass:reg_additional}(iv), absolute summability of the second-order cross-covariances, and Lemma~\ref{lem:multiplier_product_covariance} give
\begin{align}
    \E\left\{
    \E^*
    \|\bQ_T^*-\E^*\bQ_T^*\|^2
    \right\}
    \leq
    C{1+b_T+b_T^*\over T}
    =
    o(1).
    \label{eq:reg_boot_hac_expected_variance_bound}
\end{align}
Chebyshev's inequality therefore yields
\begin{align}
    \bQ_T^*
    \CPs
    \bOmega_0
    \qquad\text{in probability}.
    \label{eq:reg_boot_hac_uncentered_consistency}
\end{align}
The same result holds after centering the ideal bootstrap score array.
Indeed,
\begin{align}
    T\E^*\left\|
    T^{-1}\sum_{t=1}^T\xi_t^*\bg_t
    \right\|^2
    =
    \operatorname{tr}\left[
    {1\over T}\sum_{t=1}^T\sum_{s=1}^T
    K_p(\rho_{T,t-s})
    \bg_t\bg_s'
    \right]
    =
    \operatorname{tr}(\bOmega_{p,T})
    =
    O_p(1).
    \label{eq:reg_ideal_boot_mean_rate}
\end{align}
Hence $T^{-1}\sum_{t=1}^T\xi_t^*\bg_t=O_{p^*}(T^{-1/2})$ in probability.
Lemma~\ref{lem:reg_hac_replacement}, Assumption~\ref{ass:reg_bootstrap}(iii), and Assumption~\ref{ass:reg_additional}(v) then show that centering changes $\bQ_T^*$ by $o_{p^*}(1)$ in probability.

The desired results now follow if we can show that replacing the ideal bootstrap moments by the residual-bootstrap moment arrays does not asymptotically contribute.
The unrestricted bootstrap OLS estimator satisfies
\begin{align}
    \what{\btheta}^*-\what{\btheta}
    =
    \what{\bQ}_x^{-1}
    T^{-1}\sum_{t=1}^T
    \bx_t\xi_t^*\what u_t
    =
    O_{p^*}(T^{-1/2})
    \label{eq:reg_boot_coef_rate}
\end{align}
in probability.
The restricted bootstrap estimator similarly satisfies
\begin{align}
    \wtilde{\boldsymbol{\eta}}_1^*-\wtilde{\boldsymbol{\eta}}_1
    =
    (\bC_1'\what{\bQ}_x\bC_1)^{-1}
    \bC_1'
    T^{-1}\sum_{t=1}^T
    \bx_t\xi_t^*\wtilde u_t
    =
    O_{p^*}(T^{-1/2})
    \label{eq:reg_boot_restricted_coef_rate}
\end{align}
in probability.
For the Wald residual bootstrap, we have
\[
    \what u_t^*
    =
    \xi_t^*\what u_t
    -
    \bx_t'(\what{\btheta}^*-\what{\btheta}),
\]
so
\[
    \what{\bg}_{t,\mathrm{res}}^*
    -
    \xi_t^*\bg_t
    =
    \xi_t^*(\what{\bg}_t-\bg_t)
    -
    \bx_t\bx_t'
    (\what{\btheta}^*-\what{\btheta}).
\]
Boundedness of the two-point multipliers, Assumption~\ref{ass:reg_bootstrap}(ii), \eqref{eq:reg_uncentered_score_l2_replacement}, and \eqref{eq:reg_boot_coef_rate} imply
\begin{align}
    {1\over T}\sum_{t=1}^T
    \|\what{\bg}_{t,\mathrm{res}}^*-\xi_t^*\bg_t\|^2
    =
    O_{p^*}(T^{-1})
    \qquad\text{in probability}.
    \label{eq:reg_wald_boot_array_equivalence}
\end{align}
For the restricted residual bootstrap,
\[
    \wtilde u_t^*
    =
    \xi_t^*\wtilde u_t
    -
    \bx_t'\bC_1
    (\wtilde{\boldsymbol{\eta}}_1^*-\wtilde{\boldsymbol{\eta}}_1),
\]
so
\[
    \wtilde{\bg}_{t,\mathrm{res}}^*
    -
    \xi_t^*\bg_t
    =
    \xi_t^*(\wtilde{\bg}_t-\bg_t)
    -
    \bx_t\bx_t'\bC_1
    (\wtilde{\boldsymbol{\eta}}_1^*-\wtilde{\boldsymbol{\eta}}_1).
\]
Using \eqref{eq:reg_uncentered_score_l2_replacement} and \eqref{eq:reg_boot_restricted_coef_rate},
\begin{align}
    {1\over T}\sum_{t=1}^T
    \|\wtilde{\bg}_{t,\mathrm{res}}^*-\xi_t^*\bg_t\|^2
    =
    O_{p^*}(T^{-1})
    \qquad\text{in probability}.
    \label{eq:reg_lm_boot_array_equivalence}
\end{align}
The same rate holds after centering this array.
Applying Lemma~\ref{lem:reg_hac_replacement} with bandwidth $b_T^*$ to \eqref{eq:reg_wald_boot_array_equivalence} and \eqref{eq:reg_lm_boot_array_equivalence}, using $b_T^*=O(b_T)$ and $b_T/T^{1/2}\to0$, gives
\[
    \what{\bOmega}_p^*
    \CPs
    \bOmega_0,
    \qquad
    \wtilde{\bOmega}_p^*
    \CPs
    \bOmega_0
\]
in probability.
Since $\what{\bQ}_x\CP\bQ_x$, $\wtilde{\bOmega}_p^*\CPs\bOmega_0$, and $\wtilde{\bG}_{2\cdot1}'\wtilde{\bOmega}_p^{-1}\CP\bG_{2\cdot1}'\bOmega_0^{-1}$, we conclude $\what{\bV}_{\theta}^*\CPs\bV_{\theta}$ and $\wtilde{\bG}_{2\cdot1}'\wtilde{\bOmega}_p^{-1}\wtilde{\bOmega}_p^*\wtilde{\bOmega}_p^{-1}\wtilde{\bG}_{2\cdot1}\CPs\bG_{2\cdot1}'\bOmega_0^{-1}\bG_{2\cdot1}$ in probability by the continuous mapping theorem.
The Gaussian limits, \eqref{eq:reg_boot_cdf_validity}, and the chi-square results follow from Slutsky's theorem, the P\'olya theorem, and the continuous mapping theorem.
\end{proof}

\section{Additional Monte Carlo Results}\label{app:additional_mc}


This appendix reports Monte Carlo results complementary to those in Section~\ref{sec:mc}. 
The main text reports empirical size focusing on centered $\chi_1^2$ innovations under the combined AR(1) and stationary conditional heteroskedasticity design.
In this appendix, we first report size-adjusted power under the main designs.
We next examine robustness to Normal and Student-$t_5$ innovations under the same design, and then report corresponding two-step GMM results.
We also compare the equal-tail two-sided $z$-tests used in the main analysis with conventional LM and Wald bootstrap tests.
We finally report sensitivity to the bandwidth.

{
\noindent\textit{Covariance matching in the Monte Carlo experiments.}
The dependent Rademacher multiplier uses the latent correlation \eqref{eq:latent_gaussian_corr}, so that its covariance profile and covariance matrix are
\[
    r_0(h)=\frac{2}{\pi}\arcsin\left[\exp\{-(h/b_T)^2\}\right],
    \qquad
    \bR_0=\{r_0(|i-j|)\}_{i,j=1}^T.
\]
The Gaussian DWB uses $\bxi^G\sim N(\bzero,\bR_0)$, so its covariance matrix coincides exactly with that of the dependent Rademacher multiplier.
For the Mammen multiplier, the induced covariance $K_{p_M}$ in \eqref{eq:mammen_kernel} is evaluated at the rescaled latent correlation
\begin{align*}
    \rho^M_{T,h}&=\exp\{-c_M(h/b_T)^2\},\\
    c_M&=\argmin_{c>0}\int_0^\infty\Big[K_{p_M}\{\exp(-cx^2)\}-K_{1/2}\{\exp(-x^2)\}\Big]^2dx=0.9305,
\end{align*}
which gives a covariance profile within $0.0033$ of $r_0$ at every lag.
Exact lag-by-lag matching is not available, because inverting $K_{p_M}$ at $r_0$ does not yield a positive semi-definite latent correlation sequence.

\noindent\textit{Robustness to the covariance profile.}
Replacing the calibrated Mammen multiplier by the uncalibrated one, studentized with its own induced HAC weights, changes the rejection frequencies in Table~\ref{tab:mc_size_main} by at most $0.32$ percentage points.
Replacing the matched Gaussian multiplier by the conventional Gaussian DWB with Bartlett-correlated multipliers and Bartlett studentization changes them by at most $0.60$ percentage points; this comparison combines the change in the covariance profile with a change in studentization.
Under either Gaussian implementation, Dep.\ Rademacher rejects less often than the Gaussian procedure in all twelve comparisons of Table~\ref{tab:mc_size_main}.
}

\subsection{Power under the main design}\label{app:main_power}

Table~\ref{tab:mc_power_main} reports size-adjusted power under the alternatives corresponding to the main Monte Carlo designs in Section~\ref{sec:mc}.
For nonlinear GMM, power is evaluated at $\theta_{2,0}=0.7$ against $H_0:\theta_{2,0}=0.5$; for linear regression, power is evaluated at $\theta=0.20$ against $H_0:\theta=0$.
For each procedure, the size-adjusted test rejects when the equal-tail $p$-value is below the largest level whose rejection frequency in the corresponding null design does not exceed $5\%$.

\begin{table}[!htb]
\centering
\caption{Size-adjusted power (\%) of equal-tail two-sided $z$-tests under serially dependent heteroskedastic $\chi_1^2$ innovations}
\label{tab:mc_power_main}
\begin{threeparttable}
\small
\setlength{\tabcolsep}{4.2pt}
\renewcommand{\arraystretch}{1.03}

{
\begin{tabular}{
l
S[table-format=2.2]
S[table-format=2.2]
S[table-format=2.2]
S[table-format=2.2]
S[table-format=2.2]
S[table-format=2.2]
}
\toprule
&
\multicolumn{3}{c}{Nonlinear GMM}
&
\multicolumn{3}{c}{Linear regression}
\\
\cmidrule(lr){2-4}
\cmidrule(lr){5-7}
\multicolumn{1}{r}{$T$}
& {100} & {200} & {400}
& {100} & {200} & {400} \\
\midrule

\multicolumn{7}{l}{\textit{Restricted $z$-statistic, $\wtilde z_{LM}$}} \\

\multicolumn{7}{l}{\textit{Asymptotic}} \\
\quad ASY $r_0$
& 28.24 & 50.00 & 75.60
& 22.64 & 35.60 & 57.56 \\

\multicolumn{7}{l}{\textit{Bootstrap}} \\
\quad Dep.\ Rademacher
& 27.73 & 50.87 & 74.87
& 22.01 & 34.14 & 57.10 \\
\quad Dep.\ Mammen
& 28.00 & 49.91 & 75.01
& 21.23 & 32.83 & 55.15 \\
\quad Gaussian DWB
& 27.80 & 49.96 & 75.17
& 21.95 & 34.29 & 56.88 \\
\quad MBB
& 27.74 & 49.13 & 73.95
& 21.44 & 32.13 & 52.64 \\

\midrule

\multicolumn{7}{l}{\textit{Unrestricted $z$-statistic, $\what z_W$}} \\

\multicolumn{7}{l}{\textit{Asymptotic}} \\
\quad ASY $r_0$
& 40.04 & 60.90 & 83.54
& 20.90 & 34.12 & 53.74 \\

\multicolumn{7}{l}{\textit{Bootstrap}} \\
\quad Dep.\ Rademacher
& 33.35 & 57.75 & 82.67
& 18.60 & 30.19 & 52.55 \\
\quad Dep.\ Mammen
& 27.07 & 55.91 & 80.94
& 18.70 & 28.73 & 48.48 \\
\quad Gaussian DWB
& 37.60 & 58.84 & 82.98
& 19.26 & 31.53 & 52.88 \\
\quad MBB
& 32.23 & 53.86 & 79.17
& 18.06 & 27.72 & 47.24 \\
\bottomrule
\end{tabular}}

\begin{tablenotes}[flushleft]
\footnotesize
\item \textit{Notes:}
Entries are size-adjusted rejection frequencies under the alternatives $\theta_2=0.7$ for $H_0:\theta_2=0.5$ (one-step GMM) and $\theta=0.2$ for $H_0:\theta=0$ (OLS).
Each test is adjusted to an exact $5\%$ null rejection frequency in the corresponding null design, randomizing at the boundary level where bootstrap $p$-values are discrete.
The procedures are those of Table~\ref{tab:mc_size_main}.
\end{tablenotes}

\end{threeparttable}
\end{table}

Size-adjusted power is broadly similar across procedures for the restricted statistic $\wtilde z_{LM}$.
For the unrestricted statistic $\what z_W$, some differences remain at smaller sample sizes: in the nonlinear GMM experiment at $T=100$, size-adjusted power is $40.04\%$ for ASY $r_0$, $37.60\%$ for Gaussian DWB, $33.35\%$ for Dep.\ Rademacher, and $27.07\%$ for Dep.\ Mammen, and these differences narrow as $T$ increases.

\subsection{Additional nonlinear GMM results}

Tables~\ref{tab:nlgmm-normal} and~\ref{tab:nlgmm-t5} report the one-step GMM results under Normal and Student-$t_5$ innovations, respectively, using the same serial dependence and stationary conditional heteroskedasticity design as in the main nonlinear GMM experiment.
The corresponding two-step GMM results are reported in
Subsection~\ref{app:twostep_mc}.

\subsubsection{Hall--Horowitz recentering for MBB}\label{app:mbb_recentering}

For the nonlinear GMM implementation, we follow Hall and Horowitz (1996) and recenter the moving-block-bootstrap moment contributions by subtracting their conditional bootstrap mean.
This imposes the bootstrap analogue of the sample moment condition, so that the recentered bootstrap moment sequence has conditional mean zero. 
The recentered moments are used throughout the MBB estimation and studentization.

\begin{table}[!ht]
\centering
\caption{Empirical size and power of equal-tail two-sided $z$-tests based on the GMM estimator under serially dependent heteroskedastic Normal innovations}
\label{tab:nlgmm-normal}
\begin{threeparttable}
\small
\setlength{\tabcolsep}{5.0pt}
\renewcommand{\arraystretch}{1.05}

{
\begin{tabular}{
l
@{\hspace{1.0em}}
S[table-format=2.2]
S[table-format=2.2]
S[table-format=2.2]
@{\hspace{1.5em}}
S[table-format=2.2]
S[table-format=2.2]
S[table-format=2.2]
}
\toprule
Method
& \multicolumn{3}{c}{Size (\%)}
& \multicolumn{3}{c}{Power (\%)} \\
\cmidrule(lr){2-4}
\cmidrule(lr){5-7}
\multicolumn{1}{r}{$T$}
& {100} & {200} & {400}
& {100} & {200} & {400} \\
\midrule

\multicolumn{7}{l}{\textit{Panel A: Restricted $z$-statistic, $\wtilde{z}_{LM}$}} \\

\multicolumn{7}{l}{\textit{Asymptotic}} \\
\quad ASY $r_0$
& 9.32 & 8.46 & 6.94
& 36.66 & 53.54 & 77.34 \\

\multicolumn{7}{l}{\textit{Bootstrap}} \\
\quad Dep.\ Rademacher
& 7.10 & 6.82 & 6.10
& 31.44 & 49.00 & 75.18 \\
\quad Dep.\ Mammen
& 12.84 & 11.06 & 9.16
& 43.82 & 60.30 & 81.14 \\
\quad Gaussian DWB
& 7.34 & 7.10 & 6.72
& 32.12 & 50.10 & 75.52 \\
\quad MBB
& 11.28 & 10.18 & 8.70
& 41.62 & 58.70 & 79.52 \\

\midrule

\multicolumn{7}{l}{\textit{Panel B: Unrestricted $z$-statistic, $\what{z}_{W}$}} \\

\multicolumn{7}{l}{\textit{Asymptotic}} \\
\quad ASY $r_0$
& 14.94 & 11.40 & 8.92
& 55.28 & 69.70 & 85.90 \\

\multicolumn{7}{l}{\textit{Bootstrap}} \\
\quad Dep.\ Rademacher
& 13.00 & 9.88 & 7.96
& 52.00 & 67.20 & 84.36 \\
\quad Dep.\ Mammen
& 16.32 & 12.34 & 10.04
& 55.48 & 69.34 & 85.26 \\
\quad Gaussian DWB
& 14.02 & 10.76 & 8.56
& 52.98 & 68.62 & 85.16 \\
\quad MBB
& 13.52 & 10.16 & 8.80
& 50.46 & 65.70 & 82.80 \\

\bottomrule
\end{tabular}}

\begin{tablenotes}[flushleft]
\footnotesize
\item \textit{Notes:} All reported tests are equal-tail two-sided $z$-tests for $H_0:\theta_2=0.5$ at the nominal $\alpha=0.05$ level based on the one-step GMM estimator. The size is the rejection frequency when $\theta_2=0.5$, while the power is the rejection frequency when $\theta_2=0.7$. Critical values are given by the $\alpha/2$ and $1-\alpha/2$ quantiles of the respective reference distributions. ASY $r_0$ and the bootstrap procedures are as in Table~\ref{tab:mc_size_main}. Dep.\ Rademacher, Dep.\ Mammen, and Gaussian DWB denote the corresponding dependent wild bootstrap procedures, and MBB denotes the moving-block bootstrap.
\end{tablenotes}

\end{threeparttable}
\end{table}

\begin{table}[!ht]
\centering
\caption{Empirical size and power of equal-tail two-sided $z$-tests based on the GMM estimator under serially dependent heteroskedastic Student-$t_5$ innovations}
\label{tab:nlgmm-t5}
\begin{threeparttable}
\small
\setlength{\tabcolsep}{5.0pt}
\renewcommand{\arraystretch}{1.05}

{
\begin{tabular}{
l
@{\hspace{1.0em}}
S[table-format=2.2]
S[table-format=2.2]
S[table-format=2.2]
@{\hspace{1.5em}}
S[table-format=2.2]
S[table-format=2.2]
S[table-format=2.2]
}
\toprule
Method
& \multicolumn{3}{c}{Size (\%)}
& \multicolumn{3}{c}{Power (\%)} \\
\cmidrule(lr){2-4}
\cmidrule(lr){5-7}
\multicolumn{1}{r}{$T$}
& {100} & {200} & {400}
& {100} & {200} & {400} \\
\midrule

\multicolumn{7}{l}{\textit{Panel A: Restricted $z$-statistic, $\wtilde{z}_{LM}$}} \\

\multicolumn{7}{l}{\textit{Asymptotic}} \\
\quad ASY $r_0$
& 9.66 & 8.18 & 6.66
& 36.46 & 53.56 & 77.20 \\

\multicolumn{7}{l}{\textit{Bootstrap}} \\
\quad Dep.\ Rademacher
& 7.52 & 6.42 & 5.66
& 31.00 & 50.06 & 75.14 \\
\quad Dep.\ Mammen
& 14.26 & 11.52 & 9.80
& 44.06 & 61.12 & 81.42 \\
\quad Gaussian DWB
& 8.04 & 7.12 & 6.04
& 31.88 & 51.12 & 75.92 \\
\quad MBB
& 12.94 & 10.56 & 8.72
& 41.44 & 59.06 & 79.64 \\

\midrule

\multicolumn{7}{l}{\textit{Panel B: Unrestricted $z$-statistic, $\what{z}_{W}$}} \\

\multicolumn{7}{l}{\textit{Asymptotic}} \\
\quad ASY $r_0$
& 15.68 & 11.42 & 8.90
& 55.40 & 69.72 & 85.98 \\

\multicolumn{7}{l}{\textit{Bootstrap}} \\
\quad Dep.\ Rademacher
& 14.18 & 10.22 & 7.70
& 52.18 & 67.62 & 85.00 \\
\quad Dep.\ Mammen
& 17.82 & 13.22 & 10.08
& 55.86 & 69.94 & 85.72 \\
\quad Gaussian DWB
& 14.62 & 11.06 & 8.38
& 53.62 & 68.58 & 85.50 \\
\quad MBB
& 14.46 & 11.02 & 9.08
& 50.40 & 66.00 & 83.14 \\

\bottomrule
\end{tabular}}

\begin{tablenotes}[flushleft]
\footnotesize
\item \textit{Notes:} All reported tests are equal-tail two-sided $z$-tests for $H_0:\theta_2=0.5$ at the nominal $\alpha=0.05$ level based on the one-step GMM estimator. The size is the rejection frequency when $\theta_2=0.5$, while the power is the rejection frequency when $\theta_2=0.7$. Critical values are given by the $\alpha/2$ and $1-\alpha/2$ quantiles of the respective reference distributions. ASY $r_0$ and the bootstrap procedures are as in Table~\ref{tab:mc_size_main}. Dep.\ Rademacher, Dep.\ Mammen, and Gaussian DWB denote the corresponding dependent wild bootstrap procedures, and MBB denotes the moving-block bootstrap.
\end{tablenotes}

\end{threeparttable}
\end{table}

Tables~\ref{tab:nlgmm-normal} and \ref{tab:nlgmm-t5} show that the main conclusions are robust to the innovation distribution.
For both statistics and all sample sizes, Dep.\ Rademacher has the smallest size distortion among the bootstrap procedures and rejects less often than Gaussian DWB, while Dep.\ Mammen and MBB remain more oversized.
The unrestricted $z$-statistic $\what z_W$ generally exhibits larger size distortions.

\subsubsection{Two-step GMM robustness}\label{app:twostep_mc}

Table~\ref{tab:nlgmm-two-chi2} reports two-step GMM under the main serially dependent design with centered $\chi_1^2$ innovations.
The qualitative findings are similar to those for one-step GMM: size distortions are generally smaller for the restricted $z$-statistic $\wtilde z_{LM}$ than for the unrestricted $z$-statistic $\what z_W$, and Dep.\ Rademacher has the smallest size distortion among the bootstrap procedures for both statistics and all sample sizes.
For $\wtilde z_{LM}$, its rejection frequencies under the null are $7.44\%$, $6.90\%$, and $5.96\%$, compared with $7.60\%$, $6.94\%$, and $5.58\%$ for one-step GMM.
Thus, changing the GMM weighting scheme does not alter the main finite-sample comparison among the bootstrap procedures.

\begin{table}[!ht]
\centering
\caption{Empirical size and power of equal-tail two-sided $z$-tests based on two-step GMM under serially dependent heteroskedastic centered $\chi_1^2$ innovations}
\label{tab:nlgmm-two-chi2}
\begin{threeparttable}
\small
\setlength{\tabcolsep}{5.0pt}
\renewcommand{\arraystretch}{1.05}

{
\begin{tabular}{
l
@{\hspace{1.0em}}
S[table-format=2.2]
S[table-format=2.2]
S[table-format=2.2]
@{\hspace{1.5em}}
S[table-format=2.2]
S[table-format=2.2]
S[table-format=2.2]
}
\toprule
Method
& \multicolumn{3}{c}{Size (\%)}
& \multicolumn{3}{c}{Power (\%)} \\
\cmidrule(lr){2-4}
\cmidrule(lr){5-7}
\multicolumn{1}{r}{$T$}
& {100} & {200} & {400}
& {100} & {200} & {400} \\
\midrule

\multicolumn{7}{l}{\textit{Panel A: Restricted $z$-statistic, $\wtilde{z}_{LM}$}} \\

\multicolumn{7}{l}{\textit{Asymptotic}} \\
\quad ASY $r_0$
& 8.06 & 7.22 & 6.14
& 49.96 & 72.56 & 92.10 \\

\multicolumn{7}{l}{\textit{Bootstrap}} \\
\quad Dep.\ Rademacher
& 7.44 & 6.90 & 5.96
& 46.94 & 70.80 & 91.96 \\
\quad Dep.\ Mammen
& 12.74 & 10.54 & 8.48
& 58.00 & 77.68 & 93.76 \\
\quad Gaussian DWB
& 7.78 & 7.40 & 6.40
& 48.06 & 72.36 & 92.36 \\
\quad MBB
& 10.72 & 9.64 & 7.96
& 54.86 & 75.70 & 92.82 \\

\midrule

\multicolumn{7}{l}{\textit{Panel B: Unrestricted $z$-statistic, $\what{z}_{W}$}} \\

\multicolumn{7}{l}{\textit{Asymptotic}} \\
\quad ASY $r_0$
& 15.26 & 11.04 & 7.56
& 70.68 & 85.68 & 96.58 \\

\multicolumn{7}{l}{\textit{Bootstrap}} \\
\quad Dep.\ Rademacher
& 14.32 & 10.64 & 7.56
& 68.48 & 85.20 & 96.38 \\
\quad Dep.\ Mammen
& 17.82 & 13.34 & 9.32
& 71.00 & 85.78 & 96.42 \\
\quad Gaussian DWB
& 15.24 & 11.04 & 8.00
& 69.60 & 85.42 & 96.70 \\
\quad MBB
& 14.86 & 11.36 & 8.60
& 67.04 & 83.06 & 95.12 \\

\bottomrule
\end{tabular}}

\begin{tablenotes}[flushleft]
\footnotesize
\item \textit{Notes:} All reported tests are equal-tail two-sided $z$-tests for $H_0:\theta_2=0.5$ at the nominal $\alpha=0.05$ level based on the two-step GMM estimator. The size is the rejection frequency when $\theta_2=0.5$, while the power is the rejection frequency when $\theta_2=0.7$. Critical values are given by the $\alpha/2$ and $1-\alpha/2$ quantiles of the respective reference distributions. ASY $r_0$ and the bootstrap procedures are as in Table~\ref{tab:mc_size_main}. The second-step weighting matrix and the restricted projection use the inverse $r_0$ HAC estimator. Dep.\ Rademacher, Dep.\ Mammen, and Gaussian DWB denote the corresponding dependent wild bootstrap procedures, and MBB denotes the moving-block bootstrap.
\end{tablenotes}

\end{threeparttable}
\end{table}

\subsection{Additional regression results}

Tables~\ref{tab:reg-normal} and~\ref{tab:reg-t5} report the regression results under Normal and Student-$t_5$ innovations, respectively, using the same serial dependence and stationary conditional heteroskedasticity design and the same bootstrap implementations as in the main linear-regression experiment.

For MBB, blocks of the centered score contributions are resampled and recentered as described in Appendix~\ref{app:mbb_recentering}.
Residual block resampling is not used, because it moves blocks of residuals while keeping the regressor sequence fixed and therefore does not preserve the dependence between the regressor and the conditional variance.

\begin{table}[!ht]
\centering
\caption{Empirical size and power of equal-tail two-sided $z$-tests based on the OLS estimator under serially dependent heteroskedastic Normal innovations}
\label{tab:reg-normal}
\begin{threeparttable}
\small
\setlength{\tabcolsep}{5.0pt}
\renewcommand{\arraystretch}{1.05}

{
\begin{tabular}{
l
@{\hspace{1.0em}}
S[table-format=2.2]
S[table-format=2.2]
S[table-format=2.2]
@{\hspace{1.5em}}
S[table-format=2.2]
S[table-format=2.2]
S[table-format=2.2]
}
\toprule
Method
& \multicolumn{3}{c}{Size (\%)}
& \multicolumn{3}{c}{Power (\%)} \\
\cmidrule(lr){2-4}
\cmidrule(lr){5-7}
\multicolumn{1}{r}{$T$}
& {100} & {200} & {400}
& {100} & {200} & {400} \\
\midrule

\multicolumn{7}{l}{\textit{Panel A: Restricted $z$-statistic, $\wtilde{z}_{LM}$}} \\

\multicolumn{7}{l}{\textit{Asymptotic}} \\
\quad ASY $r_0$
& 9.00 & 8.28 & 7.76
& 27.86 & 41.32 & 59.76 \\

\multicolumn{7}{l}{\textit{Bootstrap}} \\
\quad Dep.\ Rademacher
& 5.34 & 5.68 & 6.24
& 20.86 & 34.24 & 54.54 \\
\quad Dep.\ Mammen
& 12.74 & 11.50 & 10.92
& 33.24 & 45.90 & 63.44 \\
\quad Gaussian DWB
& 6.94 & 6.86 & 6.84
& 23.20 & 36.64 & 56.26 \\
\quad MBB
& 11.64 & 10.54 & 9.98
& 31.06 & 43.24 & 60.72 \\

\midrule

\multicolumn{7}{l}{\textit{Panel B: Unrestricted $z$-statistic, $\what{z}_{W}$}} \\

\multicolumn{7}{l}{\textit{Asymptotic}} \\
\quad ASY $r_0$
& 14.94 & 11.80 & 10.00
& 35.22 & 45.08 & 61.80 \\

\multicolumn{7}{l}{\textit{Bootstrap}} \\
\quad Dep.\ Rademacher
& 8.48 & 7.32 & 7.30
& 24.46 & 34.48 & 53.26 \\
\quad Dep.\ Mammen
& 11.40 & 10.02 & 9.66
& 28.86 & 39.02 & 56.68 \\
\quad Gaussian DWB
& 9.76 & 8.64 & 7.98
& 26.56 & 37.38 & 55.54 \\
\quad MBB
& 14.54 & 11.84 & 10.56
& 31.96 & 41.08 & 57.64 \\

\bottomrule
\end{tabular}}

\begin{tablenotes}[flushleft]
\footnotesize
\item \textit{Notes:} All reported tests are equal-tail two-sided $z$-tests for $H_0:\theta=0$ at the nominal $\alpha=0.05$ level. The size is the rejection frequency when $\theta=0$, while the power is the rejection frequency when $\theta=0.20$. Critical values are given by the $\alpha/2$ and $1-\alpha/2$ quantiles of the respective reference distributions. ASY $r_0$ and the bootstrap procedures are as in Table~\ref{tab:mc_size_main}. Dep.\ Rademacher, Dep.\ Mammen, and Gaussian DWB denote the corresponding dependent wild bootstrap procedures, and MBB denotes the moving-block bootstrap.
\end{tablenotes}

\end{threeparttable}
\end{table}

\begin{table}[!ht]
\centering
\caption{Empirical size and power of equal-tail two-sided $z$-tests based on the OLS estimator under serially dependent heteroskedastic Student-$t_5$ innovations}
\label{tab:reg-t5}
\begin{threeparttable}
\small
\setlength{\tabcolsep}{5.0pt}
\renewcommand{\arraystretch}{1.05}

{
\begin{tabular}{
l
@{\hspace{1.0em}}
S[table-format=2.2]
S[table-format=2.2]
S[table-format=2.2]
@{\hspace{1.5em}}
S[table-format=2.2]
S[table-format=2.2]
S[table-format=2.2]
}
\toprule
Method
& \multicolumn{3}{c}{Size (\%)}
& \multicolumn{3}{c}{Power (\%)} \\
\cmidrule(lr){2-4}
\cmidrule(lr){5-7}
\multicolumn{1}{r}{$T$}
& {100} & {200} & {400}
& {100} & {200} & {400} \\
\midrule

\multicolumn{7}{l}{\textit{Panel A: Restricted $z$-statistic, $\wtilde{z}_{LM}$}} \\

\multicolumn{7}{l}{\textit{Asymptotic}} \\
\quad ASY $r_0$
& 9.08 & 8.14 & 6.52
& 29.96 & 41.80 & 61.88 \\

\multicolumn{7}{l}{\textit{Bootstrap}} \\
\quad Dep.\ Rademacher
& 5.72 & 5.88 & 4.86
& 22.32 & 35.76 & 55.84 \\
\quad Dep.\ Mammen
& 12.48 & 11.82 & 9.64
& 35.00 & 47.36 & 65.10 \\
\quad Gaussian DWB
& 6.80 & 6.48 & 5.30
& 24.62 & 37.88 & 57.84 \\
\quad MBB
& 11.76 & 10.76 & 8.64
& 33.40 & 44.86 & 62.38 \\

\midrule

\multicolumn{7}{l}{\textit{Panel B: Unrestricted $z$-statistic, $\what{z}_{W}$}} \\

\multicolumn{7}{l}{\textit{Asymptotic}} \\
\quad ASY $r_0$
& 14.54 & 11.58 & 8.60
& 36.58 & 46.48 & 63.58 \\

\multicolumn{7}{l}{\textit{Bootstrap}} \\
\quad Dep.\ Rademacher
& 8.54 & 7.22 & 5.80
& 26.08 & 36.50 & 55.18 \\
\quad Dep.\ Mammen
& 12.04 & 10.20 & 8.30
& 31.66 & 41.70 & 58.52 \\
\quad Gaussian DWB
& 9.82 & 8.08 & 6.32
& 28.88 & 39.26 & 57.50 \\
\quad MBB
& 15.00 & 11.96 & 9.22
& 34.18 & 43.46 & 59.30 \\

\bottomrule
\end{tabular}}

\begin{tablenotes}[flushleft]
\footnotesize
\item \textit{Notes:} All reported tests are equal-tail two-sided $z$-tests for $H_0:\theta=0$ at the nominal $\alpha=0.05$ level. The size is the rejection frequency when $\theta=0$, while the power is the rejection frequency when $\theta=0.20$. Critical values are given by the $\alpha/2$ and $1-\alpha/2$ quantiles of the respective reference distributions. ASY $r_0$ and the bootstrap procedures are as in Table~\ref{tab:mc_size_main}. Dep.\ Rademacher, Dep.\ Mammen, and Gaussian DWB denote the corresponding dependent wild bootstrap procedures, and MBB denotes the moving-block bootstrap.
\end{tablenotes}

\end{threeparttable}
\end{table}

Tables~\ref{tab:reg-normal} and \ref{tab:reg-t5} confirm that the regression findings are also robust to the innovation distribution.
For both statistics and all sample sizes, Dep.\ Rademacher has the smallest size distortion among the bootstrap procedures and rejects less often than Gaussian DWB, whereas Dep.\ Mammen and MBB remain appreciably oversized.
Size distortions are again generally larger for the unrestricted $z$-statistic $\what z_W$.
These results reinforce the main regression findings obtained under the centered $\chi_1^2$ design.

\subsection{Comparison with conventional LM and Wald bootstrap tests}
\label{app:symmetric_tests}

The main Monte Carlo results use equal-tail two-sided $z$-tests, so that the lower and upper tails of the bootstrap distribution are assessed separately.
For comparison, we also consider the conventional LM and Wald bootstrap tests.
For a scalar restriction, the LM test based on $\wtilde z_{LM}$ and the Wald test based on $\what z_W$ are equivalent to rejecting when
\[
    |Z| > q_{1-\alpha,T}^{*}(|Z^*|),
\]
where $Z$ denotes the corresponding restricted or unrestricted $z$-statistic and
$q_{1-\alpha,T}^{*}(|Z^*|)$ is the conditional $1-\alpha$ quantile of $|Z^*|$.
To keep the comparison focused, Table~\ref{tab:symmetric-comparison}
reports only empirical size for the main centered $\chi_1^2$ design under
serial dependence and stationary conditional heteroskedasticity.

\begin{table}[!ht]
\centering
\caption{Empirical size of equal-tail two-sided $z$-tests and conventional LM/Wald bootstrap tests}
\label{tab:symmetric-comparison}
\small
\begin{threeparttable}
\setlength{\tabcolsep}{4.0pt}
\renewcommand{\arraystretch}{1.05}
%
%
%
%
%
%
%
%
%
%
%
%
{
\begin{tabular}{llrrrrrr}
\toprule
&& \multicolumn{3}{c}{$\wtilde z_{LM}$}
& \multicolumn{3}{c}{$\what z_W$} \\
\cmidrule(lr){3-5}
\cmidrule(lr){6-8}
Method & Test
& {100} & {200} & {400}
& {100} & {200} & {400} \\
\midrule

\multicolumn{8}{l}{\textit{Panel A: Nonlinear GMM}} \\

\quad Dep.\ Rademacher
& Equal-tail $z$
& 7.60 & 6.94 & 5.58
& 14.06 & 11.22 & 7.84 \\
& LM/Wald
& 7.56 & 6.62 & 5.46
& 14.16 & 10.92 & 7.80 \\

\quad Dep.\ Mammen
& Equal-tail $z$
& 14.08 & 12.18 & 9.54
& 18.52 & 14.24 & 10.44 \\
& LM/Wald
& 8.72 & 6.96 & 5.66
& 16.18 & 12.28 & 8.80 \\

\quad Gaussian DWB
& Equal-tail $z$
& 7.86 & 7.24 & 6.16
& 14.76 & 11.60 & 8.46 \\
& LM/Wald
& 7.50 & 6.86 & 5.80
& 14.50 & 11.40 & 8.28 \\

\quad MBB
& Equal-tail $z$
& 13.20 & 11.32 & 8.82
& 15.32 & 12.10 & 9.22 \\
& LM/Wald
& 3.74 & 3.66 & 3.60
& 10.70 & 8.86 & 6.40 \\

\midrule

\multicolumn{8}{l}{\textit{Panel B: Linear regression}} \\

\quad Dep.\ Rademacher
& Equal-tail $z$
& 5.88 & 6.08 & 5.14
& 8.68 & 7.58 & 6.02 \\
& LM/Wald
& 5.70 & 5.92 & 5.10
& 8.54 & 7.34 & 5.92 \\

\quad Dep.\ Mammen
& Equal-tail $z$
& 13.38 & 12.10 & 9.88
& 12.48 & 10.94 & 8.92 \\
& LM/Wald
& 2.62 & 3.92 & 4.20
& 10.62 & 9.04 & 7.14 \\

\quad Gaussian DWB
& Equal-tail $z$
& 6.74 & 6.92 & 5.84
& 9.54 & 8.30 & 6.88 \\
& LM/Wald
& 6.68 & 6.60 & 5.50
& 9.50 & 8.40 & 6.76 \\

\quad MBB
& Equal-tail $z$
& 12.12 & 11.18 & 8.88
& 14.92 & 13.04 & 9.72 \\
& LM/Wald
& 3.72 & 3.72 & 3.94
& 9.38 & 8.16 & 6.56 \\

\bottomrule
\end{tabular}}
\begin{tablenotes}[flushleft]
\footnotesize
\item \textit{Notes:} Entries are rejection frequencies in percent at the
nominal 5\% level. Rows labelled ``Equal-tail $z$'' use the 0.025 and 0.975
quantiles of the bootstrap distribution of the corresponding restricted or
unrestricted $z$-statistic. Rows labelled ``LM/Wald'' report the conventional
LM test in the $\wtilde z_{LM}$ columns and the conventional Wald test in the
$\what z_W$ columns. For a scalar restriction, these are equivalent to comparing
$|z|$ with the 0.95 quantile of $|z^*|$. The design and bootstrap implementation
are otherwise identical to those in
Table~\ref{tab:mc_size_main}.
\end{tablenotes}
\end{threeparttable}
\end{table}

The comparison illustrates why a rejection frequency close to the nominal level for the conventional LM or Wald test need not imply an accurate approximation to the signed distribution of the corresponding $z$-statistic.
For a scalar restriction, the conventional test folds the bootstrap distribution around zero, since
\[
    \Pro^{*}(|Z^{*}|>q)
    =
    \Pro^{*}(Z^{*}<-q)
    +
    \Pro^{*}(Z^{*}>q).
\]
It therefore combines discrepancies in the two tails rather than assessing them separately.
A sample-specific displacement or asymmetry of the bootstrap distribution can make one tail too easy to reject and the other too difficult to reject, while the distribution of $|Z^{*}|$ can nevertheless produce an overall rejection frequency near the nominal level.

This effect is particularly apparent for the restricted $z$-statistic $\wtilde z_{LM}$.
For nonlinear GMM at $T=100$, for example, the MBB rejection frequency changes from $13.20\%$ under the equal-tail $z$-test to $3.74\%$ under the LM test, while the corresponding Dep.\ Rademacher frequencies are $7.60\%$ and $7.56\%$.
The regression experiment gives an even sharper contrast for Dep.\ Mammen, whose $T=100$ rejection frequency changes from $13.38\%$ to $2.62\%$.
Thus, the apparently improved size control of some conventional LM and Wald tests reflects in part the loss of information about directional errors in the bootstrap approximation.
By contrast, the relatively small differences between the two testing rules for Dep.\ Rademacher indicate that its finite-sample performance is much less affected by this folding.

{
\subsection{Bandwidth sensitivity}\label{app:bandwidth_mc}

Table~\ref{tab:mc_bandwidth} reports empirical size for the restricted statistic $\wtilde z_{LM}$ when the scale $c_b$ in $b_T=\lceil c_bT^{1/3}\rceil$ takes the values $0.5$, $1$, and $2$; the HAC bandwidth, the multiplier dependence scale, and the MBB block length change together.
The rejection frequencies vary moderately with the bandwidth, but the main pattern persists.
Across the 36 comparisons for both statistics, both models, the three bandwidths, and the three sample sizes, Dep.\ Rademacher has the smallest size distortion among the bootstrap procedures and rejects less often than Gaussian DWB in all but one case: for nonlinear GMM with $c_b=2$ and $T=100$, the rejection frequencies are $9.36\%$ and $9.10\%$, a difference of $0.26$ percentage points with a paired Monte Carlo standard error of $0.22$.
Dep.\ Mammen and MBB remain more oversized than Dep.\ Rademacher throughout.
The results for the unrestricted statistic $\what z_W$, which are not tabulated, show a similar pattern, with generally larger size distortions.

\begin{table}[!htb]
\centering
\caption{Empirical size (\%) of equal-tail two-sided tests based on the restricted $z$-statistic $\wtilde z_{LM}$ under alternative bandwidths}
\label{tab:mc_bandwidth}
\begin{threeparttable}
\small
\setlength{\tabcolsep}{3.2pt}
\renewcommand{\arraystretch}{1.03}
\begin{tabular}{lS[table-format=2.2]S[table-format=2.2]S[table-format=2.2]S[table-format=2.2]S[table-format=2.2]S[table-format=2.2]S[table-format=2.2]S[table-format=2.2]S[table-format=2.2]}
\toprule
& \multicolumn{3}{c}{$c_b=0.5$} & \multicolumn{3}{c}{$c_b=1$} & \multicolumn{3}{c}{$c_b=2$} \\
\cmidrule(lr){2-4}\cmidrule(lr){5-7}\cmidrule(lr){8-10}
\multicolumn{1}{r}{$T$} & {100} & {200} & {400} & {100} & {200} & {400} & {100} & {200} & {400} \\
\midrule

\multicolumn{10}{l}{\textit{Panel A: Nonlinear GMM}} \\
\quad ASY $r_0$ & 10.70 & 9.90 & 8.96 & 9.68 & 8.30 & 6.34 & 11.76 & 8.98 & 7.84 \\
\quad Dep.\ Rademacher & 8.38 & 7.94 & 6.96 & 7.60 & 6.94 & 5.58 & 9.36 & 7.64 & 6.90 \\
\quad Dep.\ Mammen & 14.94 & 12.94 & 10.78 & 14.08 & 12.18 & 9.54 & 16.06 & 13.64 & 11.26 \\
\quad Gaussian DWB & 8.64 & 8.50 & 7.48 & 7.86 & 7.24 & 6.16 & 9.10 & 8.10 & 7.58 \\
\quad MBB & 13.74 & 12.54 & 10.74 & 13.20 & 11.32 & 8.82 & 13.02 & 11.46 & 10.28 \\

\midrule

\multicolumn{10}{l}{\textit{Panel B: Linear regression}} \\
\quad ASY $r_0$ & 10.06 & 10.00 & 7.76 & 9.60 & 8.34 & 6.76 & 11.00 & 8.80 & 7.04 \\
\quad Dep.\ Rademacher & 7.00 & 7.56 & 5.70 & 5.88 & 6.08 & 5.14 & 4.90 & 5.26 & 5.24 \\
\quad Dep.\ Mammen & 14.60 & 13.04 & 10.54 & 13.38 & 12.10 & 9.88 & 13.50 & 12.42 & 10.46 \\
\quad Gaussian DWB & 8.06 & 8.36 & 6.20 & 6.74 & 6.92 & 5.84 & 6.44 & 6.36 & 5.88 \\
\quad MBB & 13.40 & 12.24 & 9.98 & 12.12 & 11.18 & 8.88 & 12.64 & 11.20 & 8.96 \\
\bottomrule
\end{tabular}
\begin{tablenotes}[flushleft]
\footnotesize
\item \textit{Notes:} Entries are null rejection frequencies at the nominal $5\%$ level in the main design of Table~\ref{tab:mc_size_main}.
The reference dependence profile and HAC estimator use $b_T=\lceil c_bT^{1/3}\rceil$; the Mammen calibration is held fixed, and the MBB block length equals $b_T$; $c_b=1$ is the specification of Table~\ref{tab:mc_size_main}.
The procedures are those of Table~\ref{tab:mc_size_main}.
\end{tablenotes}
\end{threeparttable}
\end{table}
}

\section{Data and Construction for the Short-Rate Application}
\label{app:cir_data}

The short-rate series is the monthly 3-Month Treasury Bill Secondary Market Rate, Discount Basis (TB3MS), obtained from Federal Reserve Economic Data (FRED),
maintained by the Federal Reserve Bank of St.\ Louis. 
The underlying source is the Board of Governors of the Federal Reserve System, H.15 Selected Interest Rates. 
The series is measured in percent per annum, is not seasonally adjusted, and is reported as a monthly average of business-day observations. 
Our sample covers January 1985 through August 2026 and contains 500 monthly observations with no missing values. 

For each horizon $h\in\{6,9,12\}$ months, observations are formed as $(r_t,r_{t+h})$, giving usable sample sizes $T=494$, $491$, and $488$, respectively. 
The empirical moment contribution is
\[
    \bg_t(\mu,\kappa)
    =
    \bz_t
    \left[
        r_{t+h}
        -
        \mu
        -
        (r_t-\mu)\exp(-\kappa h/12)
    \right],
    \qquad
    \bz_t=(1,r_t/10,(r_t/10)^2)'.
\]
The short rate remains in percentage points in the conditional-mean equation; division by ten is used in the instruments and defines their relative weighting in the one-step criterion. Thus the one-step estimator is conditional on this instrument normalization. 

For every horizon, the common bandwidth is
\[
    b_T=\lceil T^{1/3}\rceil=8.
\]
All procedures use the same observed statistic,
with the HAC covariance estimator based on the
Rademacher-induced lag weights described in
Section~\ref{sec:mc}.
The Gaussian multiplier covariance matches the
Rademacher covariance exactly, while Mammen's is
approximately matched.
The MBB resamples blocks of moment contributions
with Hall--Horowitz recentering and block length
$L=b_T$.
All bootstrap results use $B=9{,}999$ draws,
with the same HAC studentizer recomputed in each draw.


\section{Additional Empirical Results}
\label{app:cir_additional}

The main text reports the one-step GMM results for the nine-month horizon.
To assess sensitivity to the choice of horizon without reporting the full set of specifications, we report the corresponding one-step results for $h=6$ and $h=12$ months.
These horizons bracket the main $h=9$ specification and illustrate how the evidence changes with the forecast horizon.
We then report a two-step GMM robustness check for the main $h=9$ specification. 

\subsection{Horizon sensitivity}

\begin{table}[!htb]
\centering
\caption{Equal-tail two-sided $z$-tests of a two-year mean-reversion half-life,
$h=6$ months}
\label{tab:cir-h6}
\small
\begin{threeparttable}
{
\begin{tabular}{lrrr}
\toprule
Method & $p$-value & $q_{0.025}$ & $q_{0.975}$ \\
\midrule
\multicolumn{4}{l}{\textit{Asymptotic}} \\
\quad ASY $r_0$ & \multicolumn{1}{l}{$0.011^{**}$} & -1.96 & 1.96 \\
\multicolumn{4}{l}{\textit{Bootstrap}} \\
\quad Dep.\ Rademacher & \multicolumn{1}{l}{$0.031^{**}$} & -2.32 & 2.30 \\
\quad Dep.\ Mammen & \multicolumn{1}{l}{$0.022^{**}$} & -2.23 & 2.29 \\
\quad Gaussian DWB & \multicolumn{1}{l}{$0.026^{**}$} & -2.24 & 2.23 \\
\quad MBB & \multicolumn{1}{l}{$0.023^{**}$} & -2.23 & 2.43 \\
\bottomrule
\end{tabular}}
\begin{tablenotes}[flushleft]
\footnotesize
\item \textit{Notes:}
The sample is January 1985--August 2026 ($T=494$).
The one-step GMM estimates are $\what\mu=2.54$ and $\what\kappa=0.15$, implying a mean-reversion half-life of 4.47 years.
All procedures use the same observed restricted $z$-statistic $\wtilde z_{LM}=-2.53$, computed using the HAC covariance estimator with Rademacher-induced lag weights, as in Section~\ref{sec:mc}.
Bootstrap $p$-values and quantiles are based on 9,999 draws with recomputed HAC studentization.
$^{**}$ denotes rejection at the 5\% level. 
\end{tablenotes}
\end{threeparttable}
\end{table}

\begin{table}[!htb]
\centering
\caption{Equal-tail two-sided $z$-tests of a two-year mean-reversion half-life,
$h=12$ months}
\label{tab:cir-h12}
\small
\begin{threeparttable}
{
\begin{tabular}{lrrr}
\toprule
Method & $p$-value & $q_{0.025}$ & $q_{0.975}$ \\
\midrule
\multicolumn{4}{l}{\textit{Asymptotic}} \\
\quad ASY $r_0$ & \multicolumn{1}{l}{0.100} & -1.96 & 1.96 \\
\multicolumn{4}{l}{\textit{Bootstrap}} \\
\quad Dep.\ Rademacher & \multicolumn{1}{l}{0.192} & -2.40 & 2.43 \\
\quad Dep.\ Mammen & \multicolumn{1}{l}{0.199} & -2.44 & 2.26 \\
\quad Gaussian DWB & \multicolumn{1}{l}{0.193} & -2.36 & 2.36 \\
\quad MBB & \multicolumn{1}{l}{0.193} & -2.46 & 2.27 \\
\bottomrule
\end{tabular}}
\begin{tablenotes}[flushleft]
\footnotesize
\item \textit{Notes:}
The sample is January 1985--August 2026 ($T=488$).
The one-step GMM estimates are
$\what\mu=2.72$ and $\what\kappa=0.21$, implying a mean-reversion
half-life of 3.23 years.
All procedures use the same observed restricted $z$-statistic $\wtilde z_{LM}=-1.65$, computed using the HAC covariance estimator with Rademacher-induced lag weights, as in Section~\ref{sec:mc}.
Bootstrap $p$-values and quantiles are based on 9,999 draws with
recomputed HAC studentization.
\end{tablenotes}
\end{threeparttable}
\end{table}

Tables~\ref{tab:cir-h6} and~\ref{tab:cir-h12} report the corresponding one-step GMM results for $h=6$ and $h=12$ months, respectively.
The horizon comparison shows a clear weakening of the evidence against the two-year half-life null as the horizon increases.
At $h=6$, all asymptotic and bootstrap procedures reject at the 5\% level, whereas at $h=12$ none rejects.
The main $h=9$ specification lies between these cases and is therefore informative about differences among the distributional approximations.

\subsection{Two-step GMM robustness}
\label{app:cir_twostep}

As a robustness check on the estimation weighting matrix, we re-estimate the main $h=9$ specification by two-step GMM.
Table~\ref{tab:cir-two-h9} reports the results, which preserve the main qualitative comparison between asymptotic and bootstrap inference.

\begin{table}[!htb]
\centering
\caption{CIR conditional-mean two-step GMM results, $h=9$ months}
\label{tab:cir-two-h9}
\small
\begin{threeparttable}
\begin{tabular}{lrrrrrrr}
\toprule
\multicolumn{8}{l}{\textit{Panel A: Two-step GMM estimates}} \\
Kernel
& $\widehat\mu$
& $\widehat\kappa$
& $SE(\widehat\kappa)$
& Half-life
& $J$
& $p_J$
& $T$ \\
\midrule
$r_0$
& 2.264 & 0.178 & 0.068 & 3.89 & 0.561 & 0.454 & 491 \\
\midrule
\multicolumn{8}{l}{
\textit{Panel B: Equal-tail two-sided $z$-test of
$H_0:\kappa=\log(2)/2$ (two-year half-life)}} \\
{Method}
& $p$-value
& $c_{0.05}$
& $c_{0.95}$
& $c_{0.025}$
& $c_{0.975}$
& Reject 5\% \\
\midrule
\multicolumn{8}{l}{\textit{Asymptotic}} \\
\quad ASY $r_0$
& 0.0234 & -1.645 & 1.645 & -1.960 & 1.960 & Yes \\
\multicolumn{8}{l}{\textit{Bootstrap}} \\
\quad Dep.\ Rademacher
& 0.0608 & -1.999 & 2.016 & -2.393 & 2.374 & No \\
\quad Dep.\ Mammen
& 0.0602 & -2.025 & 1.927 & -2.362 & 2.221 & No \\
\quad Gaussian DWB
& 0.0556 & -1.992 & 2.060 & -2.313 & 2.380 & No \\
\quad MBB
& 0.0734 & -2.085 & 1.857 & -2.488 & 2.225 & No \\
\bottomrule
\end{tabular}
\begin{tablenotes}[flushleft]
\footnotesize
\item \textit{Notes:}
The sample is January 1985--August 2026.
Observations overlap.
Instruments are $(1,r_t/10,(r_t/10)^2)\prime$.
The common HAC bandwidth and MBB block length are 8.
Bootstrap inference uses equal-tail two-sided $z$-tests with 9,999 draws
and recomputed HAC studentization.
All procedures use the same observed restricted $z$-statistic $\wtilde z_{LM}=-2.27$, computed using the HAC covariance estimator with Rademacher-induced lag weights, as in Section~\ref{sec:mc}. The estimation weight is the inverse of this HAC estimator.
Standard errors use the sandwich with the estimation weight.
The $J$ test has one degree of freedom.
\end{tablenotes}
\end{threeparttable}
\end{table}

The two-step estimates imply a mean-reversion half-life of about 3.9 years,
close to the one-step estimate for the same horizon.
As in the one-step results, the asymptotic test rejects the two-year half-life null at the $5\%$ level ($p=0.023$), whereas none of the four bootstrap procedures rejects, with $p$-values between $0.056$ and $0.073$.
Thus, the main contrast between asymptotic and bootstrap inference at the
nine-month horizon is robust to the use of two-step GMM.

\end{document}